\documentclass[a4paper,reqno]{amsart}

\usepackage{hyperref}
\usepackage{inputenc,csquotes}
\usepackage{stmaryrd}
\SetSymbolFont{stmry}{bold}{U}{stmry}{m}{n}

\usepackage[english]{babel} 
\usepackage{graphicx,color,comment}  
\usepackage[dvipsnames]{xcolor} 
\usepackage{enumitem} 
\usepackage{rotating} 
\usepackage[absolute]{textpos}
\allowdisplaybreaks[3]
\usepackage{mdframed}
\usepackage{xfrac}
\usepackage{nicefrac}
\usepackage{comment,cancel}
\usepackage{stmaryrd}

\usepackage{amsmath,amssymb,amsfonts,amsthm} 
\usepackage{mathtools}  
\definecolor{light-gray}{gray}{0.95}
\usepackage[color=light-gray]{todonotes}
\usepackage{tikz-cd}  
\usepackage{tikz}
\usetikzlibrary{decorations.pathreplacing}
\usetikzlibrary{tqft}

\usepackage[all,cmtip]{xy} 
\usepackage[bbgreekl]{mathbbol} 
\usepackage{bbold}
\usepackage{slashed}
\usepackage{mathrsfs,slashed}

\usepackage{mathrsfs}

 \usepackage[backend=biber, giveninits=true, style=alphabetic, sorting=nyt,isbn=false, maxalphanames=5,url=false,doi=false, maxbibnames=99]{biblatex}
\usepackage{thmtools}

\declaretheoremstyle[
  spaceabove=5pt, spacebelow=5pt,
  headfont=\bfseries,
  notefont=\normalfont, notebraces={(}{)},
  bodyfont=\normalfont,
  postheadspace=1em,
  qed=$\Diamond$
]{pluto}
    \declaretheorem[style=pluto,name=Definition,    numberwithin=section]{definition}
    
    \declaretheorem[style=pluto,name=Proposition/Definition,    numberwithin=section]{propdef}

    \declaretheorem[style=pluto,name=Assumption,
    ]{assumption}

\declaretheoremstyle[
  spaceabove=5pt, spacebelow=5pt,
  headfont=\itshape,
  notefont=\normalfont, notebraces={(}{)},
  bodyfont=\normalfont,
  postheadspace=1em,
  qed=$\Diamond$
]{pluto2}
    \declaretheorem[style=pluto2,name=Remark,    sibling=definition]{remark}

\declaretheoremstyle[
  spaceabove=5pt, spacebelow=5pt,
  headfont=\bfseries,
  notefont=\normalfont, notebraces={(}{)},
  bodyfont=\itshape,
  postheadspace=1em,
]{pluto3}
    \declaretheorem[style=pluto3,name=Theorem,    sibling=definition]{theorem}
    \declaretheorem[style=pluto3,name=Lemma,    sibling=definition]{lemma}
    \declaretheorem[style=pluto3,name=Corollary,    sibling=definition]{corollary}

\declaretheoremstyle[
  spaceabove=5pt, spacebelow=5pt,
  headfont=\normalfont,
  notefont=\normalfont, notebraces={(}{)},
  bodyfont=\normalfont,
  postheadspace=1em,
  qed={}
]{pluto4}
    \declaretheorem[style=pluto4,name=Example,    sibling=definition]{example}

\DeclareMathAlphabet{\mathbbmsl}{U}{bbm}{m}{sl}

\newcommand{\pp}{\partial}
\renewcommand{\d}{\mathbb{d}}
\renewcommand{\top}{\mathrm{top}}
\newcommand{\tot}{\mathrm{tot}}
\newcommand{\into}{\hookrightarrow}

\newcommand{\wt}[1]{\widetilde{#1}}
\newcommand{\wh}[1]{\widehat{#1}}

\newcommand{\src}{\mathrm{src}}

\newcommand{\BV}{{\mathsf{BV}}}
\newcommand{\BFV}{{\mathsf{BFV}}}

\newcommand{\Tr}{\mathrm{Tr}}
\newcommand{\tr}{\mathrm{tr}}
\renewcommand{\Im}{\mathrm{Im}}
\renewcommand{\ker}{\mathrm{Ker}}

\newcommand{\id}{\mathrm{id}}

\newcommand{\fG}{\mathfrak{G}}
\newcommand{\fg}{\mathfrak{g}}

\newcommand{\cC}{\mathcal{C}}
\newcommand{\cO}{\mathcal{O}}

\newcommand{\cE}{\mathcal{E}}

\newcommand{\fX}{\mathfrak{X}}
\newcommand{\cEL}{\mathcal{EL}}

\newcommand{\G}{\mathscr{G}}

\newcommand{\bom}{\boldsymbol{\omega}}

\newcommand{\bHo}{\mathbf{H}_\circ}
\newcommand{\bH}{\mathbf{H}}

\newcommand{\bh}{\mathbf{h}}

\newcommand{\bE}{{\mathbf{E}}}

\newcommand{\bC}{{\mathbf{C}}}
\newcommand{\bB}{{\mathbf{B}}}

\newcommand{\bL}{\mathbf{L}}

\newcommand{\bJ}{\mathbf{J}}
\newcommand{\btheta}{{\boldsymbol{\theta}}}
\newcommand{\bK}{\mathbf{K}}

\newcommand{\bbI}{\mathbb{I}}
\newcommand{\bbE}{\mathbb{E}}
\renewcommand{\i}{\mathbb{i}}
\renewcommand{\L}{\mathbb{L}}
\newcommand{\bbR}{\mathbb{R}}
\newcommand{\bbP}{\mathbb{P}}
\newcommand{\bbh}{\mathbb{h}}

\newcommand{\fGred}{\underline{\mathfrak{G}}}

\newcommand{\fGo}{\mathfrak{G}_\circ}

\newcommand{\Ho}{{H}_\circ}

\newcommand{\uomegao}{\underline{\omega}}

\newcommand{\uS}{\underline{\mathcal{S}}{}}

\newcommand{\uuS}{\underline{\underline{\mathcal{S}}}{}}

\newcommand{\uuomegao}{\underline{\underline{\omega}}{}}
\newcommand{\oloc}{\Omega_{\loc}}

\newcommand{\loc}{\mathrm{loc}}

\newcommand{\tbox}[2]{\mbox{\parbox{#1}{\center\sffamily\footnotesize #2}}}

\newcommand{\bd}{\mathbb{d}}

\newcommand{\Ad}{\mathrm{Ad}}
\newcommand{\ad}{\mathrm{ad}}

\title{Phase spaces in field theory\\ Reduction vs. resolution}
\author{Aldo Riello}
\address{Perimeter Institute for Theoretical Physics, 31 Caroline St. N., Waterloo, ON N2L 2Y5, Canada}
\address{Department of Applied Mathematics, University of Waterloo, Waterloo, Ontario, Canada}
\email{ariello@perimeterinstitute.ca}
\author{Michele Schiavina}
\address{Department of Mathematics, University of Pavia, Via Ferrata 5, 27100 Pavia, Italy}
\address{INFN Sezione di Pavia, via Bassi 6, 27100 Pavia, Italy}
\date{}

\begin{document}

\begin{abstract}
    In this note we review the concept of phase space in classical field theory, discussing several variations on the basic notion, as well as the relation between them. In particular we will focus on the case where the field theory admits local (gauge) symmetry, in which case the physical phase space of the system emerges after a (usually singular) quotient with respect to the action of the symmetry group. We will highlight the symplectic and Poisson underpinnings of the reduction procedure that defines a phase space, and discuss how one can replace quotients with graded smooth objects within classical field theory via cohomological resolutions, a practice that goes under the name of Batalin--Vilkovisky formalism. Special attention is placed on the reduction and resolution of gauge theories on manifolds with corners, which famously depend on a number of arbitrary choices.  We phrase these choices in terms of homotopies for the variational bicomplex, and define a homotopy version of Noether's current and conservation theorem.
\end{abstract}
\maketitle

\tableofcontents
\newpage

\section{Introduction}
Classical Lagrangian field theory is a multidimensional generalisation of classical mechanics where one looks at variational problems for spaces of sections of fibre bundles (the fields of the theory) in order to encode a PDE in the Euler--Lagrange equations of an (action) functional. This PDE defines the dynamics of the model. Its solutions form the critical locus of the action and describe classical physical configurations.

In classical mechanics, the dynamical PDE is in fact an ODE, and one can unambiguously define a (finite dimensional) symplectic manifold, called \emph{phase space}, whose points correspond to physical configurations, i.e.\ solutions to the Euler--Lagrange equations. On this phase space the dynamics appears ``frozen''. However, in a time-translation invariant system, dynamical evolution can be traded for the flow of a Hamiltonian vector field associated to time-translation symmetry.  
In other words, one can establish equality between a space of ``phases'' understood as solutions, and ``phases'' understood as possible momentum/position configurations at a given time.

The multidimensional generalisation brought upon by field theory makes the relationship between these two perspectives, which underline the Lagrangian and the Hamiltonian descriptions, less straightforward. 
Classical mechanics is a 1-dimensional field theory, where fields are maps from a (time) interval to some target manifold. Therefore, one obvious source of complexity in passing from classical mechanics to field theory lies in the more general nature of the source manifolds, especially when they admit corners, i.e., codimension 2 boundaries. 
Another source of complexity arises when the field theory exhibits? a local (a.k.a.\ gauge) symmetry. In fact, the interplay of corners and local symmetries gives rise to a rich geometry on the infinite dimensional phase space.

In the general field-theoretic scenario, there are multiple ways to associate a phase space to the theory, each of which has to deal with the above challenges. In the presence of symmetries, field configurations that satisfy the equations of motion are acted upon by a distribution, and physical configurations are equivalence classes under the relation defined by it.
Importantly, note that, even assuming that the locus of solutions can be described as a smooth submanifold of all field configurations, the space of equivalence classes is often a singular space.

In the first part of this work we aim to give a description of the various methods one can employ to describe the phase space of the system as a (possibly singular) symplectic (or Poisson) manifold.

We will compare the (presymplecic) \emph{covariant phase space} of the system, naturally associated to field configurations on the source manifold $M$ with the (symplectic) \emph{geometric phase space}, which is instead naturally associated to codimension $1$ submanifolds $\Sigma\hookrightarrow M$. While the gauge distribution is encoded directly in the degenerate directions of the presymplctic structure on the covariant phase space, within the geometric phase space description this defines a set of \emph{constrained configurations} $\cC_\Sigma$, understood as necessary conditions to extend field configuration on $\Sigma$ to solutions of the Euler-Lagrange equations in an arbitrarily small tubular neighborhood of $\Sigma$.

In good cases $\cC_\Sigma$ is a coisotropic submanifold of the geometric phase space, and its characteristic distribution again defines an equivalence relation among constrained configurations, giving rise to an alternative reduction problem fully understood within the language of symplectic and Poisson geometry. The resulting space of inequivalent field configurations $\underline{\cC}_\Sigma$ is known as the reduced phase space of the system and it is expected to be equivalent to the presymplectic reduction of the covariant phase space, restoring the well-definiteness of the notion of phase space.
While the two procedures are in principle able to describe the same result, the latter method has several advantages from a technical point of view, among which---importantly---a reduction of the dimension of the source manifold on which fields are supported, which renders the explicit description of the physical phase space of the system more accessible.

In the presence of local symmetries \emph{and} codimension-2 corners, the reduction of $\cC_\Sigma$ becomes more involved, and one can naturally factor it in two distinct ``stages''. Roughly speaking, the first stage quotients all symmetries that are trivial on the corner, and yields a symplectic manifold. We call this the \emph{constraint-}reduced phase space. One can then proceed to quotient out the action of the residual ``corner gauge symmetries''. This second-stage reduction then leads to the \emph{fully-reduced} phase space $\underline{\underline{\cC}}{}_\Sigma$, which is generally only Poisson. We call its symplectic leaves (classical) superselection sectors.

In any case, it is often not possible to fully avoid the singularities ensuing from the reduction procedure, and an alternative description of such quotients becomes desirable. That is the problem of ``resolution'' of quotient manifolds, and it is conveniently handled within the Batalin--(Fradkin)--Vilkovisky formalism. This method, originally devised to discuss quantisation of gauge field theories, provides a smooth resolution of said quotients in terms of the zero cohomology of a Hamiltonian differential graded manifold.

In the second part of this work we outline these resolution techniques, and discuss how they can be used to describe (\textit{i})) the quotient of the space of solutions to the Euler--Lagrange equations of a gauge field theory by the symemtry distribution, or (\textit{ii}) the reduction $\underline{\cC}_\Sigma$ of its constraint set $\cC_\Sigma$, and how the two constructions are related. Finally, we will comment on how the cohomological approach is able to handle the description of phase spaces for gauge field theory on manifolds with corners in a natural way.

In our discussion, we will focus on the following aspects:
\begin{enumerate}
    \item To induce field theoretic data on higher codimension strata one must commit to a number of arbitrary choices, typically of representatives of (appropriate) cohomology classes. In this context the problem is known as ``corner ambiguity'', because many of the objects needed to define the phase space are in fact only defined up to higher codimension data. Following \cite{SchiavinaSchnitzer}, we rephrase this problem in terms of choices of homotopy operators for the variational bicomplex, which allows us to make consistent choices at every step. In this spirit, we revisit the definition of Noether's current in terms of such homotopies, and thus prove a generalisation of her  two celebrated theorems on variational problems with symmetry. 
    
    \item Studying a gauge theory in a finite subregion of a (Lorentzian) manifold determines a subtle and physically relevant interplay of gauge and corner data. Specifically, the (fully) reduced phase space for a gauge field theory on a manifold with corners, turns out to be a Poisson manifold owing to the reduction of residual ``corner'' gauge symmetries, which then determine a foliation of phase space into symplectic ``superselection sectors'', associated to codimension $2$ corners.\footnote{This notion is a classical counterpart of the celebrated superselections sectors discovered in algebraic quantum field theory. See \cite{GiuliniSuperselectionRules} for a review on the topic. See also the discussion in Section \ref{sec:remarks}.}
    \item We discuss the extension of the Batalin--Fradkin--Vilkovisky resolution procedure to corners of arbitrary codimension, and recover the superselection sectors of the reduced phase space from the ensuing BV Hamiltonian dg manifold structure in codimension $2$.
\end{enumerate}

This note is structured as follows. 

In Section \ref{sec:HamiltonianMechRed} we review the basics of Hamiltonian reduction for finite dimensional Hamiltonian dynamical systems, and discuss? the notion of a reduced (Poisson) phase space for $G$ invariant dynamical systems.

In Section \ref{sec:LFT} we pass from classical mechanics to field theory, recall the notion of symmetry in this context, and derive some of its consequences. We employ here an approach based on the homological properties of the variational bicomplex, which leads us to prove a new, more general, homotopy version of Noether's theorems 1 and 2, comparing with the usual formulation.

In Section \ref{sec:PS} we discuss and compare various notions of phase space for field theory with (local) symmetries: the covariant phase space and its geometric/canonical counterpart, which leads to the notion of constraint-reduced and fully-reduced phase spaces, with particular attention to the case of field theory on manifolds with corners. In this scenario we discuss the emergence of classical superselection sectors, drawing a parallel with the discussion of section \ref{sec:HamiltonianMechRed}.

In Section \ref{sec:BV} we discuss the counterpart to reduction of field configurations under the action of symmetries by means of the cohomological resolution implemented by the Batalin--Vilkovisky formalism (and its descendants).

Finally, we conclude with some remarks on the notion of symmetry and superselection in field theory in Section \ref{sec:remarks}, discussing some open questions and further developments.

\section{Hamiltonian mechanics and reduced phase space}\label{sec:HamiltonianMechRed}
Consider a symplectic manifold $(M,\omega)$, endowed with a distinguished function $H\in C^\infty(M)$. As we will see later on, this is a model for the phase space of a mechanical system with time-independent Hamiltonian. Hamilton's equation are simply a local presentation of the flow ODE's for the associated hamiltonian vector field $X_H$:
\[
i_{X_H}\omega = - dH.
\]

In the Lie algebra of vector fields over $M$ we consider those that are Hamiltonian
\[
\fX_{\mathrm{Ham}}(M)\doteq \{X\in\fX(M)\ |\ \exists F_X\in C^\infty(M), i_X\omega = -dF_X\}.
\]

On the symplectic manifold under study, we might want to consider the action of a Lie group $G$, which will encode the notion of symmetry for the mechanical system. Let us recall a few standard notions \cite{OrtegaRatiu}
\begin{definition}
    Let $G$ be a connected 
    Lie group and $\Phi\colon G\circlearrowright M$ be a smooth right action. Denote by $\rho\colon \fg\to\fX(M)$ the associated Lie algebra action. $\Phi$ is said  to be
    \begin{enumerate}
        \item Symplectic, iff $\Im(\rho)\subset \fX_{\mathrm{Sp}}(M)$,
        \item Hamiltonian, iff $\Im(\rho)\subset  \fX_{\mathrm{Ham}}(M)$, in which case there exists $J\colon M\to \fg^*$ --- called a momentum map for the action $\Phi$ --- such that 
        \[
        i_{\rho(\xi)}\omega = -\langle dJ,\xi\rangle \in \Omega^1(M),
        \]
        \item Equivariant, iff it is Hamiltonian and $J\colon M\to \fg^*$ is equivariant for the coadjoint\footnote{The coadjoint action of $\xi\in\fg$ on $\fg^*$, is the left action given by \emph{minus} the transpose of the adjoint action of $\xi$ on $\fg$, i.e.\ $\mathrm{coad}(\xi)=-\ad^*(\xi) : \fg^*\to\fg^*$.} action of $G$ on $\fg$, $L_{\rho(X)}J=-\ad^*(\xi).J$.
    \end{enumerate}
\end{definition}

Using the fact that a symplectic manifold $(M,\omega)$ is endowed with the Poisson structure $\{f,g\} \doteq i_{X_f}i_{X_g}\omega$, for all $f,g\in C^\infty(M)$, the equivariance condition (3) can be equivalently rephrased as:
\begin{enumerate}
    \item[$(3')$] The momentum map $J$ is a Poisson map between $(M,\{\cdot,\cdot\})$ and $\fg^*$ equipped with the canonical Kirillov--Konstant--Soruieau (KKS) Poisson structure.
    \item[$(3'')$] The comomentum map 
\[
    J^*\colon (\fg,[\cdot,\cdot]) \to (C^\infty(M),\{\cdot,\cdot\}),\ \xi \mapsto J^*(\xi) \doteq \langle J,\xi\rangle
\]
is a Lie algebra homomorphism
\end{enumerate}

(If $\Phi$ is a left action, (1) and (2) remain unchanged while (3) needs to be adapted so that ($3''$) states that $J^*$ is an \emph{anti}-homomorphism.)

In the time independent case, the Hamiltonian evolution of a dynamical system can then be rephrased as a Hamiltonian\footnote{Equivariance is obvious because the Lie algebra is Abelian and 1d.} $\bbR$-action on $M$ with $J_\bbR\equiv H$ the associated momentum map. As explained in the introduction, the Abelian group $\mathbb{R}$ can be understood as the symmetry group of time translations. We refer to the triplet $(M,\omega,H)$ as a \emph{dynamical system}. 
One can then consider dynamical systems which are invariant under the action of another group $G$:
\begin{definition}
    Let $(M,\omega,H)$ be a dynamical system, i.e.\ a symplectic manifold with a Hamiltonian $\bbR$-action, and let $\Phi\colon G\circlearrowright M$ be a symplectic $G$-action. If $H\in C^\infty(M)^G$ we say that $\Phi$ is a \emph{symmetry} of the dynamical system, and we say that $(M,\omega,H,\Phi)$ is a $G$-dynamical system. A symmetry is called Hamiltonian if $\Phi$ is an Hamiltonian $G$-action with (equivariant) momentum map $J$, in which case we will say that $(M,\omega,H,\Phi,J)$ is an (equivariant) Hamiltonian $G$-dynamical system.
\end{definition}

In this context we have a classical theorem of Noether, which in the Hamiltonian language can be rephrased as follows:
\begin{theorem}[Noether, Souriau, Smale]\label{thm:NSSthm}
    Let $(M,\omega,H,\Phi,J)$ be a Hamiltonian $G$-dynamical system. Then, for every $\xi\in\fg$, $J_\xi\equiv \langle J,\xi\rangle\equiv J^*(\xi)$ is a constant of motion, i.e.\ $L_{X_H}J_\xi=0$. 
\end{theorem}
\begin{proof}
    $L_{X_H}J_\xi = - i_{X_H}i_{\rho(\xi)}\omega = - L_{\rho(\xi)} H = 0$.
\end{proof}

\begin{remark}
    The above definitions and results are readily generalised to Hamiltonian $G$-dynamical systems over \emph{Poisson}, rather than symplectic, manifolds. Denoting such systems $(M,\{\cdot,\cdot\}, H,\Phi,J)$, their definition can be obtained from the above simply by replacing throughout the following defining relations for the Hamiltonian vector field $X_H$ and the momentum map $J$:
    \[
    X_H \doteq \{\cdot, H\}
    \quad \text{and}\quad 
    \rho(\bullet) = \{ \cdot , J_\bullet\}.
    \]
    For example, the proof of the Noether--Souriau--Smale theorem for a Poisson Hamiltonian $G$-dynamical system reads $L_{X_H} J_\xi = \{ J_\xi, H\} = - L_{\rho(\xi)} H = 0$.
\end{remark}

The Noether--Souriau--Smale theorem is foundational in that it implies that the dynamical system must be tangent to level sets of the momentum map, i.e.\ to the surfaces where the constants of motion have a fixed value. This tells us that, in the presence of symmetries, the effective dynamics is constrained, and can be \emph{reduced}, following the Hamiltonian reduction paradigm. The following theorem is a summary of a series of results in this area \cite{MarsdenWeinstein,Meyer,Arms_shift,MarsdenRatiuPoissonRed}; see \cite{OrtegaRatiu, Marsdenstages}.

We denote $G_\mu\doteq\{g\in G\ |\ \Ad^*(g).\mu = \mu\}$ and $\cO_\mu\doteq \Ad^*(G).\mu \simeq G/G_\mu$ respectively the stabiliser and the orbit of $\mu\in \fg^*$ under the coadjoint action.

\begin{theorem}[Hamiltonian reduction]\label{thm:MW-Poisson-red}
Let $\Phi:G\circlearrowright (M,\omega)$ be a free and proper Hamiltonian action with equivariant momentum map $J\colon M\to \fg^*$. 
Then, for any $\mu\in\Im(J)$
\begin{enumerate}[label=(\roman*),leftmargin=*,wide=0pt]
    \item The pair $(\underline{M}_\mu, \underline{\omega}_\mu)$ defined by
\[
\underline{M}_\mu\doteq J^{-1}(\mu)/G_\mu 
{\quad \text{and}\quad}
\pi^*_\mu\underline{\omega}_\mu \doteq \iota^*_\mu\omega,
\]
where $\iota_\mu^*\colon J^{-1}(\mu)\hookrightarrow M$ is the embedding of the level set and $\pi_\mu\colon J^{-1}(\mu)\to J^{-1}(\mu)/G_\mu$ is the quotient map, is a symplectic manifold called the \emph{symplectic reduction of $(M,\omega,\Phi,J)$ at $\mu$}.

 \item The following spaces are canonically symplectomorphic,  
\begin{align}
(\underline{M}_{\mathcal{O}_\mu},\underline{\omega}_{\mathcal{O}_\mu})\simeq(\underline{M}_\mu,\underline{\omega}_\mu)    
\end{align}
with
\begin{align}
    \underline{M}_{\mathcal{O}_\mu}\doteq J^{-1}(\cO_\mu)/G
\quad \text{and}\quad
    \pi_{\mathcal{O}_\mu}^*\underline{\omega}_{\mathcal{O}_\mu} \doteq \iota_{\mathcal{O}_\mu}^*(\omega - J^*\omega_\text{KKS}),
\end{align}
where $\iota_{\mathcal{O}_\mu}\colon J^{-1}(\mathcal{O}_\mu)\hookrightarrow M$ is the level set embedding, $\pi_{\mathcal{O}_\mu}\colon J^{-1}(\mathcal{O}_\mu)\to J^{-1}(\mathcal{O}_\mu)/G$ is the quotient map, and $\omega_\text{KKS}\in\Omega^2(\fg^*)$ is the Kirillov--Konstant--Souriau canonical 2-form.\footnote{The pullback of $\omega_\text{KKS}$ to any coadjoint orbit is symplectic.}
We call $(\underline{M}_{\mathcal{O}_\mu},\underline{\omega}_{\mathcal{O}_\mu})$ the \emph{symplectic reduction of of $(M,\omega,\Phi,J)$ at $\mathcal{O}_\mu$}.

\item The space of leaves of the foliation induced by the action $\rho:\fg \to \fX(M)$ is a smooth Poisson manifold $(\underline{M},\{\cdot,\cdot\}_{\underline{M}})$, foliated by the symplectic reductions at $\cO_\mu$:
\[
\underline{M} \doteq M/G = \bigsqcup_{\mathcal{O}_\mu\in \Im(J)/G}\underline{M}_{\mathcal{O}_\mu}.
\]

\item If furthermore $(M,\omega,H,\Phi,J)$ is an equivariant Hamiltonian $G$-dynamical system, then the triplet $(\underline{M}, \{\cdot,\cdot\}_{\underline{M}}, \underline{H})$ defined by
\[
\pi^*\underline{H} = H
\]
for $\pi:M\to\underline{M}$ the quotient map, is a \emph{(Poisson) dynamical system}, called the \emph{Hamiltonian reduction} of $(M,\omega,H,\Phi,J)$. The corresponding Hamiltonian vector field $X_{\underline{H}} = \{\cdot,\underline{H}\}_{\underline{M}} = \pi_* X_H$ is tangent to the symplectic leaves $\underline{M}_{\cO_\mu}\subset \underline{M}$ which indeed support a family of reduced dynamical systems $(\underline{M}_{\mathcal{O}_\mu},\underline{\omega}_{\mathcal{O}_\mu},\underline{H}_{\mathcal{O}_\mu})$ defined by
\[
\underline{H}_{\cO_\mu} \doteq \underline{H}|_{\cO_\mu}
\iff
\pi^*_{\cO_\mu}\underline{H}_{\cO_\mu} = \iota^*_{\cO_\mu} H .
\] 
and\footnote{By Theorem \ref{thm:NSSthm}, $X_H$ is tangent to the preimage $J^{-1}(\mathcal{O}_\mu)\subset M$ and can therefore be viewed as a vector field on it.} 
\[
X_{\underline{H}_{\cO_\mu}} = X_{\underline{H}}|_{\cO_\mu} = (\pi_{\cO_\mu})_*X_H.
\]
We call the dynamical system $(\underline{M}_{\mathcal{O}_\mu},\underline{\omega}_{\mathcal{O}_\mu},\underline{H}_{\mathcal{O}_\mu}) \subset \underline{M}$ the \emph{(classical) superselection sector of the $G$-dynamical system $(M,\omega,H,\Phi,J)$ at $\mathcal{O}_\mu$}.

\end{enumerate}
\end{theorem}

Theorem \ref{thm:MW-Poisson-red} can be extended quite significantly by relaxing various assumptions we have made for the sake of simplicity. A particularly interesting generalisation for us is that for non-free, proper, actions, where the symplectic reductions will be given by stratified symplectic manifold \cite{DiezHuebschmann-YMred,DiezRudolph_slice}.

\begin{remark}\label{rmk:reduced-all-relevant-info}
    It is important to note that the space $\underline{M}$ retains \emph{all} the relevant information on the dynamical system, even after quotienting by the group of symmetries. This will become relevant in the somewhat philosophical discussion surrounding any potential difference that may exist between local and global symmetries in a gauge system, see Section \ref{sec:remarks}.
\end{remark}

One can similarly introduce $(\underline{M}_\mu, \underline{\omega}_\mu, \underline{H}_\mu)$, the reduced dynamical systems at $\mu$, isomorphic to $(\underline{M}_{\mathcal{O}_\mu},\underline{\omega}_{\mathcal{O}_\mu},\underline{H}_{\mathcal{O}_\mu})$. However, the correspondence between $\cO_\mu$ and $\mu$ is one-to-many, and only the collection $\{\underline{M}_{\cO_\mu}\}_{\cO_\mu}$ is a partition of $\underline{M}$. This is why we ultimately focus on the reduction at $\cO_\mu$ and reserve the name ``superselection sectors" to them.

As a result of the symplectic reduction theorem for dynamical systems, we also have a function $\underline{H}\in C^\infty(\underline{M})$ such that $\pi^*\underline{H}=H$ where $\pi\colon M\to \underline{M}$, since $H$ is invariant and thus basic. Being a (generically nontrivial) Poisson manifold, $\underline{M}$ admits Casimir functions $Z(\underline{M})\subset C^\infty(\underline{M})$ (the center of the Poisson algebra), and in particular we have $\{\underline{H},c\} = 0, \forall c\in Z(\underline{M})$. Since the Casmirs are constant over the symplectic leaves, we again have that $\underline{H}$ is tangent to the symplectic leaves, i.e.\ the reduced phase spaces at fixed orbit $\mathcal{O}_\mu\in \Im(J)/G$, and its restriction obviously coincides with $\underline{H}_{\mathcal{O}_\mu}$.

\subsection{Example: Rotations in a central potential} \label{sec:example-rotations}
Consider the group $G=SO(3)$ and its (left) action on $\bbR^3$ by $\Phi_O(q) = Oq$. We can lift this action to the phase space $\wt{\Phi}\colon G\circlearrowright T^*\bbR^3$ whose canonical symplectic structure is given by the differential of the tautological 1-form, $\omega = d\theta$. 

Since $[(d_q\Phi_{O})^*]^{-1} p = (O^*)^{-1}p =Op$, the cotangent lift is given by 
    \[
    \wt{\Phi}_O(p,q) = [(d_q\Phi_{O})^*]^{-1} p, \Phi_{O}q) = (Op,Oq).
    \]
    Denoting by $\langle\cdot,\cdot\rangle$  the canonical pairing between $\bbR^3$ and its dual, the infinitesimal action  of the Lie algebra $\mathfrak{so}(3)$ on $\bbR^3$ is
    \[
    \rho(o) = \left\langle oq, \frac{\pp}{\pp q}\right\rangle 
    \]
    whence we readily compute its cotangent lift to be 
    \[
    \wt{\rho}(o)(p,q) = \left\langle oq, \frac{\pp}{\pp q}\right\rangle + \left\langle op, \frac{\pp}{\pp p}\right\rangle .
    \]

    The Hamiltonian flow equation follows from the fact that cotangents lifts leave the tautological 1-form $\theta= \langle p, dq\rangle$ invariant:
    \[
     0 = L_{\wt\rho(o)}\theta = i_{\wt\rho(o)}d\theta + d i_{\wt\rho(o)}\theta = i_{\wt\rho(o)}\omega + d i_{\wt\rho(o)}\theta, 
    \]
    whence the resulting Hamiltonian function---and thus the momentum map---is 
    \[
    \langle J(p,q) , o \rangle = i_{\wt\rho(o)}\theta = \langle p, \rho(o)(q)\rangle = \langle p, o q\rangle \equiv \langle p\wedge q, o\rangle,
    \]
    As can be easily checked also by direct inspection
    \begin{align*}
        i_{\wt{\rho}(o)}\omega &
        = i_{\wt{\rho}(o)} \langle dp \wedge dq\rangle  
        = -d \langle J,o\rangle
    \end{align*}

    The momentum map, taking values in $\mathfrak{so}(3)^*\simeq \bbR^3$, is then the angular momentum $J=p\wedge q \equiv \wt{p} \times q$, where $\wt{p}$ is $p$ mapped to $\bbR^3$ using the Euclidean metric. To see how equivariance comes into play here, one simply observes that 
    \[
    (p,q)\mapsto (Op,Oq), \implies \wt{p}\times q \mapsto O\wt{p}\times Oq = O(\wt{p}\times q).
    \]

    The submanifold 
    \[
    M_\mu = \{(p,q) \in (\bbR^3)^*\times\bbR^3\ |\ \mu = J(p,q) = \wt{p}\times q,\ \mu\not=0\},
    \]
    is preserved by the action of the subgroup $SO(2)$ embedded in $SO(3)$ as rotations around the axis defined by $\mu$.

    Given any ``Hamiltonian'' function $H$ that is $SO(3)$ invariant, for example that describing a particle in a central potential, we have an equivariant, Hamiltonian, $SO(3)$-dynamical system $(T^*\bbR^3, \omega, H, \Phi, J)$. Its reduced phase space is then isomorphic to the stratified Poisson manifold
    \[
    \underline{M}=T^*\bbR^3/SO(3)\simeq \underline{M}_0 \sqcup \bigsqcup_{\ell\in \bbR_{>0}}\underline{M}_\ell
    \]
    where $\underline{M}_0 = \{(p,q)\in(\bbR^3)^*\times \bbR^3 \ |\ |\tilde p \times q| = 0 \}/SO(3) \simeq T^*\mathbb{R}/\mathbb{Z}_2$ and $\underline{M}_\ell = \{(p,q) \in (\bbR^3)^*\times\bbR^3\ |\ |\wt{p}\times q|^2 = \ell^2 \not=0\}/SO(3)\simeq T^*\bbR_{>0}$.

    This example is important to guide our understanding of what happens in the general case of field theory. We note that the Poisson algebra on the reduced phase space has a nontrivial center, and a Casimir of this Poisson structure is represented in $C^\infty(T^*\bbR^3)$ as the function $\ell^2(p,q) = |\wt{p}\times q|^2\in \mathbb{R}_{\geq0}$, which is also the pullback along the momentum map $J$ of the Casimir $C_2$ of the Poisson structure on $\mathfrak{so}(3)^*\simeq \mathbb{R}^3$ given by the Euclidean length of vectors $|\cdot|^2$ provided by the Killing form.

    \begin{remark}[Quantisation]
    Upon quantisation, $T^*\bbR^3$ with vertical polarisation yields the Hilbert space $\mathcal{H}=L^2(\bbR^3)$. In $\mathcal{H},$ the Hamiltonian is mapped to a (positive) self-adjoint operator $\wh{H}$ while the angular momentum is thought of as a triplet of self-adjoint operators $\wh{J}$. While their commutators with $\wh{H}$ vanishes (by the hypothesis of rotational invariance), their mutual commutators represent the Poisson algebra of $J$---itself isomorphic to the $\mathfrak{so}(3)$ Lie algebra. The associated Casimir is $\wh{J}\,{}^2$.  The fact that $[\wh{H},\wh{J}\,{}^2]_{\mathrm{End}(\mathcal{H}}=0$ means that if we decompose the Hilbert space in the eigenspaces associated to $\wh{J}\,{}^2$, the Hamiltonian $\wh{H}$ reduces to a block diagonal operator with respect to such decomposition. By the Peter--Weyl theorem, the blocks are labeled by an irreducibile representation of $SO(3)$.
    We have here an example of the paradigm of quantisation of a Hamiltonian $G$-space, together with a quantisation of the momentum map $J$ as a representation of $SO(3)$ on $L^2(\bbR^3)$. This is to be compared with the quantisations of the leaves $\underline{M}_\ell$ of $\underline{M}$: the Bohr-Sommerfeld conditions impose the integrality of the norm of the angular momentum $\ell$, which then correspond to unirreps of $SO(3)$.
\end{remark}

    \section{Lagrangian field theory with boundary}\label{sec:LFT}
    
    Classical mechanics can be seen as a 1-dimensional field theory. When talking about the finite-time evolution of the system one considers the configuration space $C^\infty([0,1],N)$, where the manifold $N$ is the (joint) space in which the particles move (in the Example \ref{sec:example-rotations} here above $N=\bbR^3$). Consider e.g. the typical Lagrangian density\footnote{The function $L$ is thought of as a function on the first jet bundle of the trivial fibre bundle $[0,1]\times N \to [0,1]$, evaluated on (the prolongation of) a section $q\colon [0,1]\to N$. Here, $N$ is considered a Riemannian manifold whose metric induces the norm $\|\cdot\|$ on the tangent space of $N$.}
    \[
    Ldt = \left(\frac{m}{2}\|\dot{q}\|^2 - V(q)\right)dt.
    \]
    It is well-known that the variation of the action functional $S=\int_0^1Ldt$ yields two terms:\[
    \delta S = -\int_0^1 \underbrace{\left(\frac{d}{dt}\frac{\pp L}{\pp \dot{q}} - \frac{\pp L}{\pp q}\right)\delta q\ dt}_{E(L)} + \underbrace{m\dot{q}(1)\delta q(1)}_{p\delta q\vert_1} - \underbrace{m\dot{q}(0)\delta q(0)}_{p\delta q\vert_0}
    \]
    and restriction to the boundary $\pp[0,1]=\{0,1\}$ induces the map
    \[
    \pi\colon C^\infty([0,1],N) \to (T^*N\vert_0\times T^*N\vert_1, \omega_\pp\doteq -\pi^*_0\omega_{\mathrm{can}} +\pi^*_1\omega_{\mathrm{can}}).
    \]
    (Note that $\delta q\vert_0 \equiv dq_0$ where $q_0$ is a coordinate in $N_0\equiv C^\infty(\{0\},N)$.)

    \begin{figure}[h]
    \centering
    \begin{minipage}{0.45\textwidth}
    \centering
    \begin{tikzpicture}
    \draw[thick] (0,2.2) -- (0,0);
    \draw[decorate, decoration = {brace, raise=5pt}, gray, thick] (0,2.2) -- (0,0);
    \draw[->] (0.5,1.1) to (1,1.1) node [right] {$C^\infty([0,1],N)$};
    \draw[->] (0.5,0) to (1,0) node [right] {$T^*N$};
    \draw[->] (0.5,2.2) to (1,2.2) node [right] {$T^*N$};
    \filldraw[black] (0,0) circle (1pt) node[anchor=east]{$t=0$};
    \filldraw[black] (0,2.2) circle (1pt) node[anchor=east]{$t=1$};
    \end{tikzpicture}
    \end{minipage}%
    \begin{minipage}{0.55\textwidth}
    {Boundary  1-form on $T^*N^{\times2}$}
    
    {\small$pdq\vert_{t=1} - pdq\vert_{t=0}$.} 
    
    \medskip
    
    To (each) boundary of $[0,1]$ associate
    
    {symplectic phase space.}
    
    \medskip 
    
    Dynamics is identified by a
    
    {Lagrangian submanifold} $\subset(T^*N)^{\times 2}$.
    
    \end{minipage}
    \end{figure}
    All solutions of the Euler--Lagrange (evolution) equations $E(L)=0$ are mapped by $\pi$ to a Lagrangian submanifold $EL_{[0,1]}\subset T^*N\vert_0\times T^*N\vert_1$. Indeed, $EL_{[0,1]}$ is isotropic because $\pi^*\omega_\pp = \delta\left(\delta S - E(L)\right) = \delta E(L)\approx 0$ vanishes on the space of solutions $E(L)=0$ and $\pi^*$ is injective. As isotropic complement one can take a boundary condition (say $q$ fixed at $0$ and $1$). 

    One can interpret the Lagrangian submanifold coming from the solutions of the Euler--Lagrange equations in the interval interior as the graph of the Hamiltonian flow from initial data to final data. Uniqueness of the solution of the EL equation is given by the (unique) intersection of Lagrangian submanifolds given by the graph of the Hamiltonian flow, and---say---the initial value condition $(q(0),p(0))$ \cite{CattaneoPhaseSpace}.

    \begin{remark}
    It is important to fix the idea that, to a ``spacetime cylinder'', we associate an ``evolution'' Lagrangian submanifold in the product of the symplectic manifolds associated to each boundary component, taken with a natural orientation (and thus a sign for the canonical symplectic form) in this particular case.
    \end{remark}

    We conclude this brief overview of mechanics by noting that in this field-theoretic formulation of evolution in classical mechanics, we do not need to appeal to the time-\emph{independence} of the dynamics, nor to the corresponding time-translation symmetry.

    \subsection{Manifolds with boundaries and corners}
    The generalisation to field theory, which will be covered in detail in Section \ref{sec:PS}, replaces the interval $[0,1]$ with a manifold $M$, possibly with corners. (For the notion of manifold with corners used here see Definition \ref{def:corners}, borrowed by \cite{Joyce_corners}.) There are few reasons why one should consider formulating field theory on manifolds with boundaries and corners, here we list a few of the main ones.\\

    \paragraph{\bf Axiomatic approaches and functorial field theory} 
    
    Since the first axiomatic approaches of Atiyah and Segal \cite{Atiyah1988TopologicalQF,SegalTQFT} it was understood that it may be desirable to phrase quantisation of classical field theories as a functor between the category of cobordisms (possibly with additional geometric data) to some appropriate linear category. This, in generality, requires one to understand what a field theory assigns, both at classical and quantum level, to a manifold with corners on which it is cast.\footnote{The idea of studying a quantum field theory through composable amplitudes associated to bordering spacetime regions can be traced back all the way to Dirac's seminal 1933 paper \cite{Dirac1933}, which inspired Feynamn's doctoral work on the path integral.}\\

    \paragraph{\bf Causal Diamonds and Lenses} 
    The point above yields a rich structure in the particular case of topological and Riemannian field theories, where topological nontrivialities are front and centre. In field theories with propagating degrees of freedom on Lorentzian manifolds the causal structure determines the emergence of ``\emph{causal} building blocks'', given by causal diamonds or their doubly-truncated versions (Fig. \ref{fig:causalregions}, two centre figures), which are manifolds with corners. A ``smoother" alternative to the truncated causal diamond, is the causal lens (Fig. \ref{fig:causalregions}, right), whose boundary is the union of two spacelike disks, and whose corner is a circle. These causal regions have a trivial topology, which simplifies the axiomatic construction to some extent, but allow nontrivial boundary data. Therefore, both the gluing procedure and the Hamiltonian analysis of field theory on such causal regions with corners require a dedicated study.\\   

    \begin{figure}
        \centering
        \includegraphics[width=.95\textwidth]{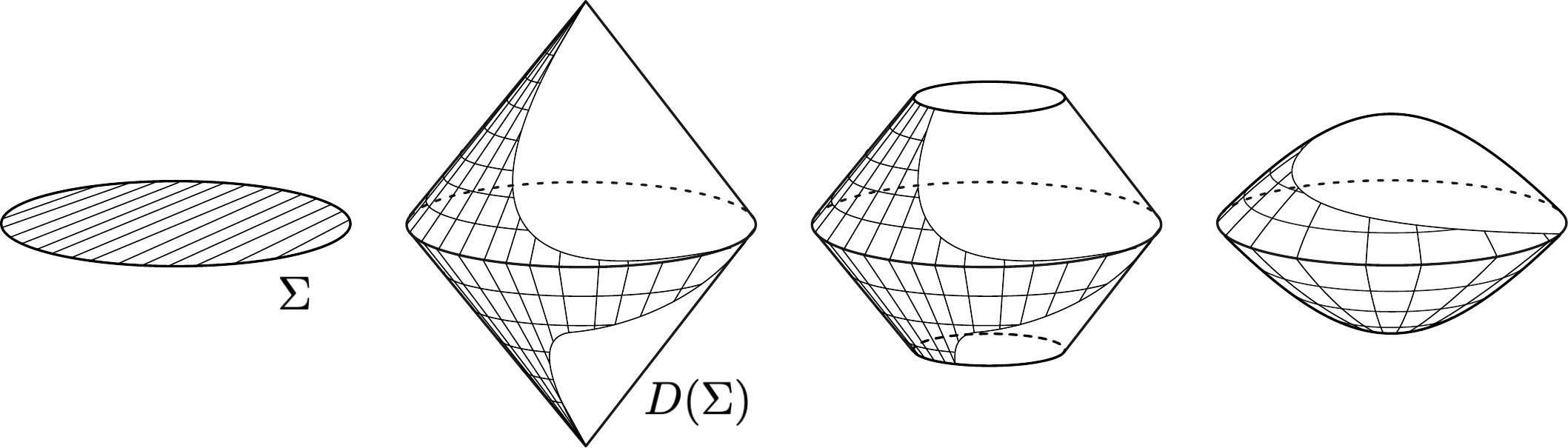}
        \caption{The spacelike codimension-1 surface $\Sigma$ and its Cauchy development $D(\Sigma)$, also called ``causal diamond", are followed by a truncated version of the causal diamond and a smooth deformation of the latter, called the ``causal lens''. Note that $\pp D(\Sigma)$ is a null surface with corners, while the boundary of the causal lens is a spacelike surface with corners.}
        \label{fig:causalregions}
    \end{figure}

    \paragraph{\bf Measurement problems}
    Another foundational question in quantum field theory revolves around the concept of measurement. While we are fairly accustomed to the success of ordinary, nonrelativistic, quantum information theory, much less is known in the realm of relativistic quantum measurements, starting from the very question of what should be considered to be a (sub)system. Recently, significant progress has been made in connecting information theory to (algebraic) quantum field theory \cite{FJLR}, and said developments hint at the fact that, in order to study quantum measurements, one needs to have control of the quantum theory on a closed, finite, causal region with corners.

    \subsection{Local Lagrangian field theory - the naive theory}
    In this section we give a definition of a local Lagrangian field theory with symmetries and follow the traditional construction. In Section \ref{sec:LFThomotopies} we will review these definitions in light of a more accurate analysis of the various choices involved, in terms of suitable choices of homotopies. For simplicity, throughout this and the following sections, we will assume $M$ compact (possibly with boundary and corners) and $E\to M$ a finite-rank vector bundle (results generalise easily to affine bundles, important for field spaces of connections). Denote by $\cE \doteq \Gamma(M,E)$ the space of \emph{global} sections of the bundle.

    We consider the variational bicomplex $\left(\Omega^{\bullet,\bullet}(J^\infty E),d_V,d_H\right)$ on the infinite jet bundle \cite{Anderson:1989} and, denoting 
    \[
    j^\infty \colon \cE\times M \to J^\infty E, \qquad j^\infty(\phi,x) = j^\infty_x\phi
    \]
    the infinite jet evaluation---which is surjective on the jet bundle of a vector bundle---we introduce:
    
    \begin{definition}[Local and horizontal forms]
        The bicomplex\footnote{This is essentially the variational bicomplex, just seen inside forms on the Fr\'echet manifold $\cE\times M$. For certain purposes it is easier to work here rather than on $J^\infty E$.} of local forms $(\oloc^{\bullet,\bullet}(\cE\times M), \d,d)$ is 
        \[
        \oloc^{\bullet,\bullet}(\cE\times M)\doteq(j^\infty)^*\Omega^{\bullet,\bullet}(J^\infty E)
        \]
        with
        \[
        \d(j^\infty)^*\alpha\doteq(j^\infty)^* d_V\alpha 
        \quad\text{and}\quad d(j^\infty)^*\alpha \doteq (j^\infty)^*d_H\alpha.
        \]
        Moreover, 
        \begin{enumerate}
            \item The complex of horizontal forms is $(\oloc^{0,\bullet}(\cE\times M),d)$.
            \item The complex of horizontal $p$-forms is $(\oloc^{p,\bullet}(\cE\times M),d)$.
            \item The \emph{higher} complex of horizontal forms is $(\oloc^{\geq1,\bullet}(\cE\times M), d))$.
            \item The \emph{higher} bicomplex of local forms $(\oloc^{\geq1,\bullet}(\cE\times M), d,\d)$
        \end{enumerate}
         
    \end{definition}

    \begin{remark}[Notation]
        If no risk of confusion arises, to avoid clutter we will often omit the manifold and the space of section over which local forms are defined---e.g.
        \[
        \oloc^{\bullet,\bullet} \equiv \oloc^{\bullet,\bullet}(\cE\times M).
        \]
    \end{remark}

    Looking at the variational bicomplex it is manifest that one can effectively work cohomologically. Among the various works that address the problem from this perspective we mention the monograph \cite{VinogradovCohomological}. (See also the recent review \cite{vinogradovreview} and references therein.) However, much of the structural advantage given by working with local forms, for our purposes, is based on the following theorem of Takens \cite{Takens77}, later extended by Anderson \cite{Anderson:1989}:

    \begin{theorem}[Anderson--Takens]\label{thm:Takens}
        The space of $(\geq1,\top)$-local forms on $\cE\times M$ splits into the image of the horizontal differential $d\oloc^{\geq1,\top-1}$ and a direct complement, denoted by $\Omega_{\src}^{\geq1,\top}$ and called the space of \emph{source}\footnote{Note that Takens called ``source'' only $(1,top)$-forms in the image of $\bbI$. Anderson uses this nomenclature too, and calls more generally forms in the image of $\bbI$ ``functional''---see \cite[Prop 3.1]{Anderson:1989} for an explanation of the name.} $(\geq1,\top)$-local forms:
        \[
        \oloc^{\geq1,\top} = d\oloc^{\geq1,\top-1} \oplus \Omega_\src^{\geq1,\top}.
        \]
        We call the projector\footnote{The projector is actually $i\circ \bbI$ with $i\colon \Omega_{\src}^{\geq1,\top} \to \oloc^{\geq1,\top}$ denoting the obviuous inclusion.} to the source component $\bbI\colon \oloc^{\geq1,\top}\to \Omega_{\src}^{\geq1,\top}$ the \emph{interior Euler operator}.
        
        The \emph{exterior Euler operator}, defined by $\bbE\doteq \bbI\d\colon \oloc^{p,\top} \to \Omega_{\src}^{p+1,\top}$ is nilpotent $\bbE^2=0$, and (obviously) satisfies
        \[
        \d \bL - \bbE(\bL) \in d\oloc^{1,\top-1}, \quad \forall \bL\in \oloc^{0,\top}.
        \]
    \end{theorem}

    Note that, given a local $0$-density $\bL\in\oloc^{0,\top}$ we can write
        \[
        d\btheta_{\bL} \doteq \d\bL - \bbE(\bL),
        \]
    where $\btheta_{\bL}$ is defined up to a $d$-closed $(1,\top-1)$ local form. Here, we invoke another important result of Takens, which tells us that $\btheta_\bL$ is indeed defined up to a d-\emph{exact} local $(1,\top -2)$ form.
    
    \begin{theorem}[Takens' acyclicity]
        For $p\geq 1$ the complex of horizontal $p$-forms $\left(\oloc^{p,\bullet}(\cE\times M),d\right)$ is exact except in top degree $\bullet = \top$.
    \end{theorem}

    \begin{definition}[Euler--Lagrange set]
        The \emph{Euler-Lagrange set} of $\bL\in\oloc^{0,\top}$ is 
        \[
        \cEL(\bL)\doteq\mathrm{Zero}(\bbE(\bL)).
        \]
        A configuration $\phi\in\cEL(\bL)\subset \cE$ is said to be \emph{on-shell}.
    \end{definition}
    
    We now provide a notion of Lagrangian symmetry. This notion will be slightly revised in Section \ref{sec:symms}, while its relation with the notion of symmetry  for classical mechanics, as provided in Section \ref{sec:HamiltonianMechRed}, will be elucidated in Sections \ref{sec:hamwithcorners} and \ref{sec:remarks}. A more significant generalisation of the current approach will be discussed in Section \ref{sec:BV}.
    
    Let $\Phi\colon \G\circlearrowright \cE$ be a group action, and denote by $\rho$ the associated Lie algebra action. According to the usual notion: The action $\rho$ is a \emph{symmetry} of the Lagrangian $\bL\in\oloc^{0,\top}$ iff 
    \[
    \forall \xi \in \fG=\mathrm{Lie}(\G)\quad \exists \bB_\xi\in\oloc^{0,\top-1}, \quad  \L_{\rho(\xi)}\bL = d\bB_\xi.
    \]

    In this context we have Noether's theorem \cite{Noether}:
    \begin{theorem}[Noether I, naive version]
        Let $\rho\colon \fG\to \fX(\cE)$ be a symmetry of the Lagrangian $\bL\in\oloc^{0,\top}(\cE\times M)$. Consider the map 
        \[
        \bJ\colon \fG \to \oloc^{0,\top-1}, \quad  \xi \mapsto \bJ_\xi\doteq\bB_\xi + \i_{\rho(\xi)}\btheta.
        \]
        Then, for every $\xi\in\fG$, $\bJ_\xi$ is $d$-closed on the Euler-Lagrange set $\cEL(\bL)$. We call $\bJ$ the Noether current associated to the symmetry $\rho$.
    \end{theorem}
    \begin{proof}
        Recalling $[\i_{\rho(\xi)},d]=\i_{\rho(\xi)} d + d \i_{\rho(\xi)} = 0$, we compute
        \[
        d \bJ_\xi = d\bB_\xi + d\i_{\rho(\xi)}\btheta = d\bB_\xi - \i_{\rho(\xi)} d\btheta = \- \L_{\rho(\xi)}\bL - \i_{\rho(\xi)}(\bd \bL - \bbE(\bL)) = \i_{\rho(\xi)}\bbE(\bL)
        \]
        which manifestly vanishes on $\cEL(\bL)$.
    \end{proof}

    Noether's second theorem, states that if $\mathfrak{G}$ is itselfs the space of sections of a vector bundle on $M$ (a ``gauge'' symmetry), then $\bJ_\xi$ is indeed exact on $\cEL(\bL)$. However, a rigorous proof of this theorem already requires more technology, to which we now turn.
    In fact, strong of this new technological hindsight, we will later provide a more general and rigorous proof of Noether's first theorem as well.

    \subsection{Homotopies}\label{sec:LFThomotopies}
    An important shortcoming of the traditional approach outlined above is that Noether's current depends on choices of primitives $\btheta$ and $\bB$. 
    Moreover, the choice of $\btheta$ affects the Hamiltonian theory, which will be constructed in Section \ref{sec:hamiltoniantheory}, and the fact that many Lagrangians $\bL$ lead to the same equations of motion is another aspect of the same problem. This issue has been clearly noticed in the physics literature \cite{LeeWald} (see \cite{Wald_Noethercharge,JakobsonKangMyers, IyerWald} for a study of its role in black hole thermodynamics, and \cite{WaldZoupas} at asymptotic null infinity).
    
    We introduce here a construction that exploits the full extent of the cohomological nature of the variational bicomplex, and the existence of contracting homotopies. The choice of a homotopy corresponds to a self-consistent way of fixing all the above ambiguities: Indeed, despite there is considerable freedom in choosing contracting homotopies, they must satisfy consistency conditions that constrain how ambiguities can be fixed across different degrees.
    As a result, given a symmetry, one can introduce a notion of Noether's current defined via a choice of homotopy, and thus formulate its conservation theorems.

    We start by recalling a series of results on the variational bicomplex that we formulate in terms of the space of local forms. We follow closely \cite{SchiavinaSchnitzer} and references therein. 
    The first result we present is a reformulation of \cite[Ch.5]{Anderson:1989} as it appears in \cite{SchiavinaSchnitzer}:
    
    \begin{theorem}\label{thm:horhomotopy}
    The higher horizontal complex $(\oloc^{\geq 1,\bullet},d_H)$ of forms in vertical degree greater or equal than one is a deformation retract of the trivial complex of $(\geq 1,\top)$ source forms:
    \begin{equation}\label{e:p-retract}
    \xymatrix{
          h^\geq \circlearrowright (\oloc^{\geq 1,\bullet},d) \ar@<1ex>[r]^-{\bbI} & \ \Omega_\src^{\geq 1, n} \ar@<1ex>[l]^-{i} 
        }
    \end{equation}
    where the degree $-1$ map $h^\geq\colon \oloc^{\geq 1,\bullet} \to \oloc^{\geq 1,\bullet -1}$ is Anderson's global homotopy (for a choice of symmetric connection $\nabla$ on $E\to M$), $\bbI$ is the source projection (the interior Euler operator) and $i$ is the trivial inclusion. 
    In particular, we have
    \begin{enumerate}
        \item the retraction $\bbI \circ i = \id_\src$,
        \item the homotopy equation $\id_\loc = h^\geq d + d h^\geq + i \circ \bbI$,
        \item the side condition $\bbI \circ h^\geq=0$.
    \end{enumerate}
    \end{theorem}
    
    Because of the specifics of $\bbI$, for $\alpha\in\oloc^{\geq1,k}$, one has the useful identities
    \[
    \alpha = 
    \begin{cases}
    h^\geq d\alpha & \text{if $k=0$},\\ 
    [d, h^\geq]\alpha & \text{if $1\leq k\leq \top-1$},\\
    d h^\geq(\alpha) + \bbI(\alpha) & \text{if $k = \top$}.
    \end{cases}
    \]
    Note that one can always change the homotopy $h^\geq$ in order to get a ``special homotopy''. In fact, one can replace $h$ with $\wt{h}=h(ih + hi)$ to obtain $\wt{h}i=0$, and $\wt{h}$ with $\wt{h}'=\wt{h}i\wt{h}$ to further obtain $\wt{h}'{}^2=0$ \cite{crainic}.

    \begin{remark}
        Note that a version of the deformation retract in Equation \eqref{e:p-retract} exists also independently for every complex of horizontal $p$-forms $\oloc^{p,\bullet}$. However, it was suggested in \cite{SchiavinaSchnitzer} that it is convenient to consider the full retract in Equation \eqref{e:p-retract}. This is because only when considering the entire higher complex $\oloc^{\geq 1,\bullet}$, one can apply the homological perturbation lemma (see \cite{crainic}) in the direction of $\d$. This yields a new deformation retract
        \begin{equation}\label{eq:tot-homotopy}
        \xymatrix{
          \wt{h}^{\geq 1} \circlearrowright (\oloc^{\geq 1,\bullet},\mathbf{d}\doteq d + \d) \ar@<1ex>[r]^-{\bbI} & \ (\Omega_\src^{\geq 1, n},\wt{\d}) \ar@<1ex>[l]^-{\wt{i}},
        }
        \end{equation}
        for some appropriate maps $\wt{h}^{\geq 1},\wt{i}$, and an induced differential $\wt{\d}$, so that 
        \[
        \id = \mathbf{d} \wt{h}^{\geq 1} + \wt{h}^{\geq 1} \mathbf{d} + \wt{i} \bbI .
        \]
        (Observe that the twisted map $\wt{\bbI} \equiv \bbI$ turns out to be unchanged. See \cite{SchiavinaSchnitzer} for further details.)
    \end{remark}

    Let $p:\cE \times M \to M$ be the projection on the second factor. 
    Then 
    \[
    p^*\colon \Omega^\bullet(M)\to \oloc^{0,\bullet}
    \]
    identifies forms in $\Omega^k(M)$ with  $\d$-closed forms in $\oloc^{0,k}(\cE\times M)$ \cite[Prop. 1.9]{Anderson:1989}. We refer to forms in the image of $p^*$ as ``field-independent" or ``constant" local forms.\footnote{Seeing $\alpha\in\oloc^{0,k}(\cE\times M)$ as a $k$-form-valued function on $\cE$, the adjective ``constant'' refers to a constant function on $\cE$, that is one that is annihilated by $\d$.}, i.e.\ We can use this fact to define a vertical homotopy $\bbh$ \cite[Prop. 4.1]{Anderson:1989} such that:
    \begin{theorem}
    The vertical complex $(\oloc^{\bullet, k}(\cE\times M),\d)$ is a deformation retract of the De Rham complex on $M$:
    \begin{equation}\label{eq:vert-homotopy}
        \xymatrix{
          \bbh \circlearrowright (\oloc^{\bullet, k}(\cE\times M),\d) \ar@<1ex>[r]^-{0^*} & \ (\Omega(M)^k,d) \ar@<1ex>[l]^-{p^*},
        }
        \end{equation}
    where $\bbh\colon \oloc^{p,q} \to \oloc^{p-1,q}$ is called the \emph{vertical homotopy}. In particular,
    \[
    \id = \bbh\d + \d\bbh + p^*0^*,    
    \] 
    where $0: M \to \cE \times M$ is the pullback along $j^\infty:\cE \times M \to J^\infty E$ of the zero section of $J^\infty E\to M$.\footnote{In other words, $0^*$ is the evaluation of a local form over $\cE\times M$ on the zero section in $\cE=\Gamma(M,E)$, which leaves us with a ($\cE$-constant) form on $M$.} Moreover, the vertical homotopy $\bbh$ can be chosen such that it commutes with $d$, i.e.\ $\bbh d + d \bbh = 0$, and 
    \[
    \bbh^2 = 0^*\circ\bbh = \bbh \circ p^* \equiv 0.
    \]
    \end{theorem}

    Consider the following definition from \cite{SchiavinaSchnitzer}, where $(\Omega^\bullet(M)[1],d_M)$ is the De Rham complex on $M$ shifted by 1 degree---whereby $k$-forms are given degree $k-1$:
    \begin{definition}
    The \emph{horizontal cone} is the complex\footnote{Recall that $\Omega(M)^k[1]\equiv(\Omega(M)[1])^k = \Omega^{k+1}(M)$.}
        \[
        C(p^*)\doteq (\Omega(M)^\bullet[1]\oplus\oloc^{0,\bullet}, d_C), \qquad d_C(a,b)=(-d_Ma, db + p^*a).
        \] 
    \end{definition}

    \begin{theorem}[\cite{SchiavinaSchnitzer}]
    Define the Euler Projector $\bbP$ as the map\footnote{The map $\bbP$ is indeed a projector, $\bbP^2=\bbP$, see \cite{SchiavinaSchnitzer}.} 
    \[
    \bbP \doteq \bbh\circ i \circ \bbE \equiv \bbh \circ i\circ \bbI \circ \d \colon \oloc^{0,\bullet} \to \oloc^{0,\bullet}.
    \] 
    Then, the cone $(C^\bullet(p^*), d_C)$ is a deformation retract of the trivial complex given by the image of the Euler projector:
    \[
    \xymatrix{
          h_C \circlearrowright \left(\Omega(M)^\bullet[1]\oplus\oloc^{0,\bullet}, d_C \right) \ar@<1ex>[r]^-{\wt{\bbP}} & \ \Im(\bbP) \ar@<1ex>[l]^-{i_{\mathrm{can}}},
        }
    \]
    with homotopy
    \[
    h_C(a,b) \doteq (0^*b, h^0 b),  \qquad d_C h_C + h_C d_C + i_{\mathrm{can}}\wt{\bbP} = \id,
    \]
    where $i_{\mathrm{can}}$ is the canonical inclusion, $\wt{\bbP}\doteq0\oplus \bbP$, and 
     $h^0 \doteq -\bbh h^\geq \d$.
    \end{theorem}

        \begin{lemma}
            Let $\bbP^0 \doteq \bbP + p^*0^*$. Then, $\bbP^0$ is a projector and one also has the following deformation retract:
                \[
            \xymatrix{
                h^0\circlearrowright(\oloc^{0,\bullet},d) \ar@<1ex>[r]^-{\bbP^0}    
                        & \Im(\bbP^0)
                        \ar@<1ex>[l]^-{\id}
                }
            \]
            with homotopy equation
            \[
            \id = dh^0 + h^0d + \bbP^0 = dh^0 + h^0d + \bbP + p^* 0^* .
            \]
        \end{lemma}
        
        \begin{proof}
        First we observe that $\bbP^0$ is a projector $\bbP^0\bbP^0 = \bbP^0$ owing to $0^*\bbP(F)=0$, which follows from $0^*\bbh = 0$. Moreover, we have
        \[
        [d,h^0] = -[d,\bbh {h}^{\geq 1} \d] = \bbh [d, {h}^{\geq 1}] \d = \bbh (\id - {i}\bbI)\d = \bbh \d - \bbP = \id - \cancel{\d \bbh} - p^*0^* - \bbP,
        \]
        where we used both the horizontal homotopy equation as well as the fact that $\bbh$ annihilates $\oloc^{0,\bullet}$.
        \end{proof}

    \begin{remark}
        Note that the map $h^0=-\bbh \wt{h}^{\geq 1}\d$, first introduced in \cite[Prop. 4.3 and p. 121]{Anderson:1989},
 fits into several different homotopy equations, e.g.\ that for the cone $C^\bullet$ as a retract of $\mathrm{Im}(\bbP)$ and that for $(\oloc^{0,\bullet},d)$ as a retract of $\Im(\bbP^0)$. In \cite[Corollary 3.8]{SchiavinaSchnitzer} the map $h^0$ was derived ``naturally'' from a diagram of homotopy equivalences. The nomenclature $h^0$ comes from the fact that $h^{\geq}$ acts on $\oloc^{\geq 1,\bullet}$, so that effectively one can think of $h^0$ as extending this map. However, care must be exercised, since the two maps satisfy different homotopy equations.
    \end{remark}

    With these tools, we can now define and characterise a Lagrangian field theory and derive its properties.

\subsection{Local Lagrangian Field Theory} We start with the following definition

    \begin{definition}[Lagrangian field theory]
        A Lagrangian Field Theory on a manifold $M$ is the data of a space of fields $\cE=\Gamma(M,E)$ and a Lagrangian density $\bL \in \oloc^{0,\top}(\cE\times M)$.
        The associated Euler-Lagrange Set is
        \[
        \cEL(\bL) \doteq \mathrm{Zero}(\bbE(\bL)) \subset \cE.
        \]
    \end{definition}
       
    Lagrangian field theories on a given $\cE$ come in equivalence classes $[\bL]$ defined by the equivalence relation of sharing the same Euler--Lagrange equations:
        \[
        \bL \sim \bL' \iff \bbE(\bL)=\bbE(\bL').
        \]

    Note that for $C$ is a constant on $\cE$, i.e.\ $\d C=0$, one has $\bL\sim \bL'=\bL + C + d\mathbf{Z}$. In fact the converse is also true:

    \begin{lemma}\label{lemma:equiv-Lagr}
        $\bL\sim\bL'$ if and only if any one of the following equivalent assertions is verified:
        \begin{enumerate}
            \item $\bbE(\bL)=\bbE(\bL')$,
            \item $\mathbb{P}(\bL)=\mathbb{P}(\bL')$,
            \item $\bL' - \bL \in \ker(\d) + \Im(d)$.
        \end{enumerate}
    \end{lemma}
    \begin{proof}
        As well known $(3)\implies (1)$. Now, within $\oloc^{0,\top}$, the kernel of the projection $\bbP \doteq \bbh\circ i\circ \bbI \circ \d$ is given by  the sum 
        \[
        \ker(\bbP)= \ker(\d) + \Im(d)
        \]
        (see \cite[thm 5.9]{Anderson:1989}, and compare with \cite[Theorem 5.2.6]{Blohmann} and \cite{SchiavinaSchnitzer})---note that the sum is not direct since $\ker(\d)$ and $\Im(d)$ have non-empty intersection. Therefore, $(2)\implies (3)$. Finally, recall that $\bbP = \bbh\circ i \circ \bbE$, whence $(1)\implies (2)$, since $\bbh \circ i$ is injective. 
    \end{proof}

    \begin{remark}
    Typically, one identifies Lagrangian densities when they agree on all sections of compact support in the interior of $M$. This leads to the identification of equivalence classes of Lagrangian densities with the horizontal cohomology of the variational bicomplex in top degree (see \cite{BFLS}). Here, on manifolds with nontrivial (top) De Rham cohomology, we are only slightly enlarging these equivalence classes by also identifying Lagrangians that differ by a constant (i.e.\ field-independent) top form $\ell \in \oloc^{0,\top}(\cE\times M)\cap\ker(\d)\simeq \Omega^\top(M)$. Indeed, the standard result (of \cite{BFLS}) that the horizontal complex is a resolution of the space of classes of (nonconstant) Lagrangian densities generally only holds for contractible manifolds $M$, while a more involved construction is needed in general (see \cite[Remark 3.12]{SchiavinaSchnitzer}).
    As a consequence of this generalisation, we are also led to incorporate in our homotopic version of Noether's theorem what we call ``external currents" (cf. Defintition \ref{def:NoetherCurrent} and Examples \ref{ex:Maxwell-soruces1} and \ref{ex:Maxwell-soruces2}).
    \end{remark}

    We note here that, given a homotopy $h=(h^\geq,\bbh)$, one can fix a representative for the class $[\bL]$ by $\bL^h \doteq \bbP[\bL]$. In fact, the representative $\bL^h = \bbh \bbE \bL$, is solely determined by the Euler-Lagrange equations of $[\bL]$ by means of a choice of vertical homotopy $\bbh$. Nevertheless, our definition of a Lagrangian Field Theory reflects the fact that we will here work with a \emph{given} choice of Lagrangian, e.g. one that is ``natural'' for its transformation properties under some group action. Once such a choice of $\bL$ is given, we will leverage the homotopical structures of local forms to \emph{consistently} fix across degrees the choice of representatives for $d$-exact quantities computed from $\bL$, such as the primitives $\bB_\xi$, and $\btheta$, which play a central role in the definition of the Noether's current.

    \begin{example}[Examples of $\bL^h$]\label{rmk:bLh}
        (1) Consider first a point particle in an external potential, $\cE_\text{pp} = C^\infty(\bbR,\bbR^3)$ and $\bL_\text{pp} = ( \frac{m}2 \dot{\vec q}\ {}^2  - V(\vec q\ ))dt$. Applying to its Lagrangian the projection $\bbP = \bbh \circ i \circ \bbE$ built using Anderson's vertical homotopy, one finds
        \begin{align*}
        \bL_\text{pp}^h = 
        \bbP(\bL_\text{pp}) 
        & = \left(-\frac{m}2\ddot{\vec q}\cdot \vec q - V(\vec q ) + V(\vec 0)\right)dt .
        \end{align*}
        (2) Consider next first order Maxwell theory (on a trivial bundle) written as a deformation of $BF$, i.e.\ $\cE_\text{1Max} = \Omega^2(M) \times \Omega^1(M)\ni (B,A)$ and $\bL_\text{1Max} = B \wedge dA - \frac12 B \wedge \star B$, one finds
            \[
            \bL^h_\text{1Max} = \bbP(\bL_\text{1Max}) = \frac12(B \wedge dA - dB \wedge A ) - \frac12 B \wedge \star B.
            \]
        (3) Finally, consider the Chern-Simons Lagrangian $\bL_\text{CS} = \Tr(A \wedge dA + \frac23 A\wedge A\wedge A)$, where $A$ is a $\mathfrak{su}(n)$-connection 1-form. It is easy to check that $\bL_\text{CS}$ is in the image of $\bbP$, i.e.\ $\bL_\text{CS} = \bbh \Tr(2F_A \wedge \d A)$ where $F_A = dA +  A\wedge A$ is the curvature of $A$.
    \end{example}

\subsection{Symmetries and Noether currents}\label{sec:symms}
    In this section we revisit, using the homotopies defined above, the standard theory of symmetries and their analysis in terms of Noether currents.
    We start with the following definition:
    
    \begin{definition}[Lagrangian symmetry]\label{def:sym}
        A Lie algebra action $\rho\colon\fG\to \fX(\cE)$ is a \emph{symmetry} of a Lagrangian field theory $(\cE,\bL)$ if it leaves $[\bL]$ invariant, i.e.\ if for all $ \xi\in \fG$
        \[
        \bbP(\L_{\rho(\xi)}\bL) = 0.
        \]
    \end{definition}

    By Lemma \ref{lemma:equiv-Lagr} this condition does not depend on the choice of representative $\bL$ of $[\bL]$, namely, if $\bL\sim\bL'$ and $\rho$ is a symmetry of $(\cE,\bL)$ then it is also a symmetry of $(\cE,\bL')$. Therefore, $\rho$ can be thought of as a symmetry of the class $(\cE,[\bL])$.

    The following theorem guarantees the compatibility of our definition of a symmetry with the Euler-Lagrange set: it says that the perturbation of an on-shell configurations along a symmetry is a Jacobi field, that is, a solution to the linearised equations of motion:

    \begin{theorem}\label{thm:inv-eom}
        Let $\rho$ be a symmetry of $(\cE,\bL)$. Then for all $\xi\in\fG$, 
        \[
        (\L_{\rho(\xi)} \bbE \bL)|_{\cEL} = 0.
        \]
    \end{theorem}

     \begin{proof}
    Since $\bbE \bL$ is a source $(1,\top)$-form, it can be written as
    \[
        \bbE \bL \equiv \bE_I \d \phi^I
    \]
    with $\bullet^I$ an abstract index in the fibre of $E\to M$ and $\bE_I$ a family of $(0,\top)$ local forms (the densitised Euler--Lagrange equations). We also denote $\rho(\xi)\phi^I \equiv \delta_\xi \phi^I$.
    Now, taking the Lie derivative along $\rho(\xi)$ of $\d\bL = \bbE \bL + d\btheta$, using the horizontal homotopy equation $\id = h^0d + d h^0 + \bbP + p^*0^*$ applied to $\L_{\rho(\xi)}\bL$, and recalling the definition of symmetry $\bbP(\L_{\rho(\xi)}\bL)=0$, we find:
        \begin{align*}
            \L_{\rho(\xi)}\bbE\bL & = \L_{\rho(\xi)} (\d\bL - d\btheta) = - d(\d h^0 \L _{\rho(\xi)}\bL + \L_{\rho(\xi)}\btheta) + \cancel{\d p^*0^* \L_{\rho(\xi)}\bL}.
        \end{align*}
    We deduce that $\bbI \L_{\rho(\xi)} \bbE\bL = 0$, whence
    \[
        0=\bbI( \i_{\rho(\xi)}\d\bbE\bL) = \bbI( (\L_{\rho(\xi)}\bE_I)\d \phi^I +  \bE_I \d \delta_\xi\phi^I )  = (\L_{\rho(\xi)}\bE_I)\d \phi^I + \bbI(\bE_I \d \delta_\xi\phi^I ).
    \]
    The theorem follows because $\bE_I \d \delta_\xi\phi$ vanishes at all $\phi \in \cEL(\bL)$ and therefore so does $\bbI(\bE_I \d \delta_\xi\phi )$ by linearity of $\bbI$.\end{proof}

    We now turn to the definition and properties of the (homotopy) Noether current.
    For convenience, on $\Omega^\bullet(M)[1]\oplus\oloc^{\bullet,\bullet}$ introduce the following collective notation---but note that this is \emph{not} a homotopy itself:
    \[
    h_C^\tot \doteq
    \begin{dcases}
        h_C & \text{on }\Omega^\bullet(M)[1]\oplus\oloc^{0,\bullet}\\
        \id \oplus h^{\geq} & \text{on }\Omega^\bullet(M)[1]\oplus\oloc^{\geq 1,\bullet}
    \end{dcases}
    \]

    \begin{definition}[Noether current]\label{def:NoetherCurrent}
        Let $\rho$ be a symmetry of $[\bL]$ and $(h^\geq,\bbh)$ horizontal and vertical homotopies, and $\bL$ a representative of $[\bL]$. Define the Cone Noether Current associated to $\bL$
        \[
        \bJ^{h_C}_\bullet \equiv (S_\bullet, \bJ^h_\bullet)\doteq [ \i_{\rho(\bullet)}, h^\tot_C] (0,\d \bL) :\fG \to \Omega^\top(M)[1] \oplus \oloc^{0,\top-1}.
        \]
        Its first component, $S_\bullet : \fG\to \Omega^{\top}$, is called the \emph{external $0$-current}, while its second component, $\bJ^h_\bullet : \fG \to \oloc^{0,\top-1}$ is called the \emph{Noether current}.
    \end{definition}    
    More explicitly, $\bJ^{h_C}_\bullet = (S_\bullet,J^h_\bullet)$ with
    \begin{align}\label{eq:ext-soruces}
                S_\bullet = 0^*\L_{\rho(\bullet)}\bL
    \end{align}
    and
    \begin{align}\label{eq:Jh}
        \bJ^h_\bullet = h^0\L_{\rho(\bullet)} \bL + \i_{\rho(\bullet)} h^{\geq} \bd \bL = - \bbh h^\geq \L_{\rho(\bullet)} \d \bL + \i_{\rho(\bullet)} h^{\geq} \bd \bL.
    \end{align}
    The 0-current $S_\bullet$ is named ``external'' because it does not depend on the dynamical fields. We will shortly see that it plays the role of an ``external source" for the Noether current $\bJ^h_\bullet$.

    The following is the technical step needed to prove the first (homotopy) Noether theorem:
    \begin{lemma}
        For all $\xi\in\fG$, $d_C \bJ^{h_C}_\xi = (0,\i_{\rho(\xi)}\bbE\bL)$. In particular, $\bJ^{h_C}_\xi$ is $d_C$-closed on the Euler-Lagrange set.
    \end{lemma}
    \begin{proof}
    Dropping from $\i_{\rho(\xi)}$ and $\L_{\rho(\xi)}$ the subscript $\bullet_{\rho(\xi)}$, we compute
    \begin{align*}
        d_C \bJ^{h_C} 
        & =  d_C (0, \i h^\geq \d \bL) + d_C h_C (0,\L \bL) = \i (0, -d h^\geq \d \bL) + d_C h_C (0,\L \bL)\\
        & =  \i (0,  h^\geq \cancel{d \d \bL}) + \i (0, \bbI \d \bL) - \i (0,\d\bL) - h_C\cancel{d_C (0,\L \bL)} - (0\oplus \bbP) (0,\L\bL) +  (0,\L \bL)\\
        &= (0, \i\bbE\bL) - (0, \cancel{\bbP\L\bL}) = (0, \i\bbE\bL),
    \end{align*}
    where in the first line we used $[d,\i]=0$, and in passing from the first to the second line, the homotopy equations $\id = dh^\geq+h^\geq d + \bbI $ and $\id = d_C h_C + h_C d_C + 0\oplus \bbP$; terms in the second line drop because they involve differentials of top forms, and in the third line because of the definition of symmetry.
    \end{proof}
    
    \begin{theorem}[Noether I]\label{thm:N1}
        Let $\rho$ and $\bL$ as above. Then, for all $\xi\in\fG$
        \[
        d \bJ^h_\xi = - p^*S_\xi + \i_{\rho(\xi)}\bbE\bL.
        \]
        In particular, if the Euler-Lagrange set is non-empty, $\cEL(\bL)\neq \emptyset$, the external 0-current $S_\xi \in \Omega^\top$ is $d$-exact. 
    \end{theorem}
    \begin{proof}
        From the lemma, $(0,\i\bbE\bL) = d_C \bJ^{h_C} = d_C(S, \bJ^h) = ( 0, d\bJ^h + p^*S)$.
    \end{proof}

Combining Noether's first theorem with Theorem \ref{thm:inv-eom} on the invariance of the equations of motion, we obtain the following equivariance property of $d\bJ^h_\bullet$:
\begin{corollary}\label{cor:equi-dJ}
    $
    \L_{\rho(\xi)} d \bJ^h_\eta = d \bJ^h_{[\xi,\eta]} + p^*S_{[\xi,\eta]}.
    $
\end{corollary}

    \begin{example}[Maxwell with sources]\label{ex:Maxwell-soruces1}
        An example of Noether's first theorem where the external sources $S_\xi$ play an important role is given by electromagnetism coupled to an external current $j_\mathrm{ext} \in \Omega^{\top-1}$, with Lagrangian $\bL = \frac14 F_A\wedge \star F_A - j_\mathrm{ext} \wedge A$, where $F_A=dA$ denotes the curvature of an $\bbR$-principal connection $A$ (over a trivial bundle), and $\star$ is the Hodge operator on forms, induced by a (possibly Lorentzian) metric $g$. Considering the gauge symmetry $\rho(\xi) A = d\xi$, one can verify that Anderson's homotopies give
        \[
        \bJ^{h_C}_\xi = ( {S_\xi} , \bJ^h_\xi) = (  j_\mathrm{ext} \wedge d\xi, \star F_A \wedge d\xi)
        \]
        Since $\bbE(\bL) = (d \star F_A - j_\mathrm{ext})\wedge \d A$, the Euler-Lagrange set is non-empty only if the external current is closed (conserved), $dj_\mathrm{ext}=0$---that is, only if the external $0$-current is $d$-exact, $S_\xi= d(j_\mathrm{ext} \xi)$.
    \end{example}

    Having given a homotopy version of Noether's first theorem, we turn now to her second. This involves the case where the symmetry is local, namely:
    \begin{definition}[Local symmetry]
        A Lie algebra action $\rho:\fG\to \mathfrak{X}(\cE)$ is local if the following conditions hold:
        \begin{enumerate}
            \item $\fG = \Gamma(M,\Xi)$ where $\Xi\to M$ is a vector bundle 
            \item $\fG$ is equipped with a Lie bracket 
            \[
            [\cdot,\cdot]\colon \fG\times \fG \to \fG
            \]
            which is local as a map of sections of vector bundles \cite{Blohmann,RSATMP},
            \item the action $\rho$, seen as a map $\fG\times \cE \to T\cE =\Gamma(M, VE)$, is local as a map of sections of vector bundles.
        \end{enumerate}
        If, furthermore, $\rho$ is a symmetry of $(\cE,\bL)$, it is called a \emph{local symmetry}.
    \end{definition}

    Local symmetries allow for the following powerful construction, based on the idea that when $\rho$ is a local symmetry, the action Lie algebroid $\mathsf{A}\doteq\mathfrak{G}\times\mathcal{E}\xrightarrow{\rho} T\mathcal{E}$ can be seen itself as the space of sections of a (vector) bundle over $M$.
    
    This allows us to identify a linear and local map\footnote{These are called ``dual-valued local forms" in the language of \cite[Section 2.2.]{RSATMP}.} $\mathbf{F}_\bullet : \fG \to \oloc^{0,k}$ with a local form $\mathbf{F}$ on $\oloc^{0,k}(\mathsf{A}\times M)$ that are linear in the $\fG$ entry, and thus canonically construct from $\mathbf{F}_\bullet$ a local $(1,k)$-form on $\mathsf{A}\times M$ by 
    \[
    \check{\mathbf{F}}\doteq\mathbf{F}_{\d \xi}\in\oloc^{1,k}(\mathsf{A}\times M).
    \]
        
    Note that, from $\check{\mathbf{F}}$, we can conversely reconstruct $\mathbf{F}$ by contraction with the tautological vector field $\check{R} \in \mathfrak{X}(\fG)\into\mathfrak{X}(\mathsf{A})$ such that $\i_{\check{R}}\d\xi=\xi$. In view of the linearity of $\mathbf{F}_\bullet$, this operation is in fact the same as applying Anderson's vertical homotopy, namely $\mathbf{F} = \bbh \check{\mathbf{F}}$.

    Understanding the complex of local forms on $\cE\times M$ as the subcomplex of local forms on $\mathsf{A}\times M$ which are basic with respect to the projection $\mathsf{A}=\fG\times\cE \to \cE$, we can reinterpret the homotopies used so far as relating to local forms on $\mathsf{A}$ rather than $\cE$. This allows us not to introduce new symbols for the homotopy operators in $\oloc^{\bullet,\bullet}(\mathsf{A}\times M)$.
    \begin{lemma}\label{lemma:check}
        Let $\mathbf{F}_\bullet : \fG \to \oloc^{0,k}$ be a linear and local map. Then, a choice of homotopies $(h^\geq,\bbh)$ induces the decomposition 
        \[
        \mathbf{F}_\bullet = \mathbf{f}^h_\bullet + d \mathbf{k}^h_\bullet.
        \]
        where
        \[
        \mathbf{f}^h_\bullet \doteq \bbh (h^\geq d \check{\mathbf{F}} +  \mathbb{I}\check{\mathbf{F}} )
            \quad\text{and}\quad
            \mathbf{k}_\bullet^h\doteq - \bbh h^\geq \check{\mathbf{F}}
        \]
        are linear and local maps from $\fG$ to $\oloc^{0,k}$ and $\oloc^{0,k-1}$, respectively.
    \end{lemma}
    
    \begin{proof}
        Recall the identification between the map $\mathbf{F}_\bullet$ with the local form on $\oloc^{0,k}(\mathsf{A}\times M)$, and construct the local $(1,k)$-form on $\mathsf{A}\times M$ by $\check{\mathbf{F}}\doteq\mathbf{F}_{\d \xi}\in\oloc^{1,k}(\mathsf{A}\times M)$. 
        
        Using the homotopy equation for $h^\geq$, we find
        \[
            \check{\mathbf{F}} = d h^\geq \check{\mathbf{F}}  +  h^\geq d \check{\mathbf{F}} + \mathbb{I}\check{\mathbf{F}} .
        \]
        Now, recall that Anderson's vertical homotopy $\bbh$ preserves the linearity in $\xi$ and commutes with $d$. Thus, applying $\bbh$ we find the desired decomposition by defining $\mathbf{f}_\bullet^h$ and $\mathbf{k}_\bullet^h$ as in the statement of the theorem, and using the above identification in reverse.
    \end{proof}

    \begin{theorem}[Noether II]\label{thm:N2}
        Let $\rho$ and $\bL$ be as in Noether's first theorem \ref{thm:N1}, and let $(h^\geq,\bbh)$ be a choice of homotopies. Furthermore, assume that $\rho$ is local. 
        Then, the Noether current $\bJ^h_\bullet$ decomposes into the sum of the Noether Constraint current $\mathbf{C}^h_\bullet$ and the differential of the Flux Noether 2-current $\mathbf{K}^h_\bullet$,\footnote{See the proof for an explicit expression of $\mathbf{C}^h_\bullet$ and $\mathbf{K}^h_\bullet$.}
        \[
        \bJ^h_\bullet = \mathbf{C}^h_\bullet + d \mathbf{K}^h_\bullet,
        \]
        such that $\mathbf{C}^h_\bullet$ vanishes on $\cEL(\bL)$ whenever the external sources vanish $S_\bullet=0$. More generally, if $S_\bullet\neq0$ and $\cEL(\bL)\neq\emptyset$, there is an external current $j^h_\bullet$ such that $S_\bullet = d j^h_\bullet$ and 
        \[
        \mathbf{C}^h_\bullet|_{\cEL(\bL)} = j^h_\bullet.
        \]
    \end{theorem}
    \begin{proof}
        If $\cEL(\bL)$ is empty then the theorem is trivially true. If $\cEL(\bL)$ is not empty, we proceed as follows. For ease of notation we remove subscripts and superscripts.
        Noting that  $\bbI\check{\bJ}=0$ because $\check{\bJ}\in\oloc^{1,\top-1}(\mathsf{A}\times M)$ is not a top form on $M$, application of Lemma \ref{lemma:check} to $\bJ$ and $S$ tells us that $\bJ = \mathbf{C} + d \mathbf{K}$ for 
        \[
        \mathbf{C} \doteq \bbh h^\geq d \check\bJ{}
            \quad\text{and}\quad
            \mathbf{K} \doteq - \bbh h^\geq \check\bJ{} ,
        \]
        and $S = s + d j$ for (here we think of $S\in\oloc^{0,\top}(\fG\times M)$)
        \[
        s \doteq \bbh \mathbb{I}\check S 
            \quad\text{and}\quad
            j \doteq - \bbh h^\geq \check{S}.
        \]
        Recall now Noether's first theorem \ref{thm:N1}, which states that $d \bJ = - S + \i\bbE\bL$ and that, consequently, $S$ is $d$-exact on shell if $\cEL(\bL)\neq\emptyset$---as assumed here. 

        Now, using Noether's first theorem, we compute
        \[
            \mathbf{C} = \bbh h^\geq d \check\bJ{} =  \bbh h^\geq (- \check S + \i_{\rho(\d \xi)}\bbE\bL) = j + \bbh h^\geq \i_{\rho(\d \xi)}\bbE\bL,
        \]
        The theorem follows if we show that $\bbh h^\geq \i_{\rho(\d \xi)}\bbE\bL$ vanishes on $\cEL$.
        To show that this is the case, we proceed as follows. First, observe that $\rho(\xi)\phi^I = D_\alpha^I\xi$ for some (field-dependent) linear differential operators $D_\alpha^I$, whence $\i_{\rho(\d\xi)}\bbE\bL = \bE_I D^I_\alpha \d \xi^\alpha$ or, for short, $\i_{\rho(\d\xi)}\bbE\bL = \bE D \d \xi$. Next, note that the homotopy $h^{\geq}$ effectively ``integrates by parts'' differential operators, which means that
        \[
        \bbh h^\geq \i_{\rho(\d \xi)}\bbE\bL = \bbh (D' \bE \cdot D'' \d \xi) = D' \bE \cdot D'' \xi
        \]
        for some other (field-dependent) linear differential operators $D'$ and $D''$. One therefore concludes by noting that $\bE$ vanishes on $\cEL$ by definition and therefore so does $D'\bE$. 
        \end{proof}

Combining Noether's second theorem with Theorem \ref{thm:inv-eom} on the invariance of the equations of motion, recalling that by construction $\d j^h_\bullet =0$, we find that the Noether constraint current is invariant on the Euler-Lagrange set:

\begin{lemma}\label{lem:inv-C}
Under the assumptions of Theorem \ref{thm:N2},
    $\left.\left(\L_{\rho(\xi)}\mathbf{C}^h_{\eta}\right)\right|_{\cEL} = 0$.
\end{lemma}

\begin{proof}
Recall $\mathbf{C}^h_\eta = j_\eta + D' \bE \cdot D'' \eta$ for all $\eta\in\fG$ from the proof of Noether's second theorem \ref{thm:N2}. 
Deriving this expression along $\rho(\xi)$ we find:
\[
\L_{\rho(\xi)}\mathbf{C}_\eta^h = D' (\L_{\rho(\xi)}\bE) \cdot D'' \eta + (\L_{\rho(\xi)} D')\bE \cdot D''\eta + D' \bE \cdot \L_{\rho(\xi)}(D''\eta).
\]
where we used $\L_{\rho(\xi)} j_\eta = 0$ since $j_\eta$ is a constant ($\d j_\eta = 0$).
Evaluating on $\cEL(\bL) = \mathrm{Zero}(\bbE \bL = \bE_I\d\phi^I)$, the last two terms on the right-hand side clearly vanish, while the first one vanishes in view of Theorem \ref{thm:inv-eom}.
\end{proof}

From this lemma it follows that external currents are in fact defined on the Abelianisation of $\fG$ whenever the constraint current $\bC^h_\bullet$ is equivariant:
\begin{theorem}\label{thm:jext=0}
    Under the assumptions of Theorem \ref{thm:N2}, assuming moreover that $\mathbf{C}^h_\bullet$ is equivariant. Then, if $\cEL(\bL)\neq \emptyset$, the commutator ideal belongs to the kernel of the external current map $j^h_\bullet : \fG \to \Omega^{\top-1}(M)$, i.e.
    \[
    [\fG,\fG]\in \ker(j^h_\bullet).
    \] 
    Therefore, $j^h_\bullet$ vanishes whenever the Abelianisation $\fG/[\fG,\fG]$ is trivial.
\end{theorem}
\begin{proof}
    All we need to show is that $j$ vanishes when evaluated on a commutator, i.e.\ that $j_{[\xi,\eta]} =0$ for all $\xi,\eta\in\fG$.
    Now, from the assumed equivariance of $\mathbf{C}$ and Noether's second theorem \ref{thm:N2}, we compute respectively: $\L_{\rho(\xi)}\mathbf{C}_\eta = \mathbf{C}_{[\xi,\eta]} = j_{[\xi,\eta]} + \varepsilon$, where $\varepsilon$ is a quantity that vanishes on $\cEL$. Thus, from Lemma \ref{lem:inv-C}, we conclude that $j_{[\xi,\eta]} = 0$ since we have shown that  $j_{[\xi,\eta]}$ is a constant on $\cE$ that takes on a zero value on $\cEL\subset \cE$. If the Abelianisation is trivial, then $j$ is identically zero on the whole of $\fG$.
\end{proof}

    \begin{example}[Maxwell with sources, continued]\label{ex:Maxwell-soruces2}
        Continuing from example \ref{ex:Maxwell-soruces1}, we see that in electromagnetism coupled to an external current $j_\mathrm{ext}$: the Noether Constraint current is $\mathbf{C}^h_\xi = (d\star F) \xi$ and the Noether Flux 2-current is $\mathbf{K}^h_\xi = \star F\xi$, while the external sources, when compatible with the equations of motion (recall, this requires that $dj_\mathrm{ext}=0$), are given by $S_\xi = dj^h_\xi$ with $j^h_\xi = \xi j_\mathrm{ext}$. We have used Anderson's homotopies throughout.
    \end{example}

    \begin{example}[Yang--Mills with sources]\label{ex:jext-YM}
    In non-Abelian Yang-Mills (YM) theory the constraint current is $\mathbf{C}_\xi = \tr(\xi d\star F)$. If the theory under study has a semisimple structure algebra, e.g. $\fg=\mathfrak{su}(n)$, one has $\fg=[\fg,\fg]$ and thus necessarily $j^h_\bullet = 0$ there.
    
    At the level of the YM Lagangian, one sees that the coupling of external sources studied in the case of electromagnetism generalises only if $\langle j_\mathrm{ext}, [\xi,\eta]\rangle =0$ for all $\xi,\eta\in\fG$ or, equivalently, only if $\ad^*(\xi).j_\mathrm{ext} =0$ for all $\xi \in \fG$. To see why, consider YM theory for a trivial principal bundle with gauge algebra $\fG \simeq C^\infty(M,\fg)$ and Lagangian\footnote{Like in the Abelian case, this emerges e.g. from a (truncated) Dirac-YM Lagrangian with non-dynamical fermion fields in the fundamental irrep of the gauge group, $(j_\mathrm{ext})^\mu_a = \bar \psi^{i\alpha}(\gamma^\mu)_{\alpha\beta} (\tau_a)_{ij} \psi^{j\beta}$ with $\tau_a$ a basis of $\fg$.} $\bL = \frac12 \tr(F\wedge \star F) + \langle j_\mathrm{ext} \wedge A\rangle$. 
    Now, for $\L_{\rho(\xi)}\bL = \langle j_{\mathrm{ext}},d_A\xi\rangle$ to be a boundary term one would need $d_A j_\mathrm{ext} =0$. This however fails to be compatible with the condition $\d j_\mathrm{ext}=0$ unless $j_\mathrm{ext}$ annihilates the commutator ideal, since in fact  $d_Aj_{\mathrm{ext}}=0$ implies $ 0 =  \d (d_A j_\mathrm{ext}) = \ad^*(\d A).j_\mathrm{ext}$. In this case, $d_A j_\mathrm{ext} = d j_{\mathrm{ext}}$, and the gauge transformations are again a symmetry as long as $j_\mathrm{ext}$ is a conserved current in the \emph{ordinary} sense, $d j_\mathrm{ext}=0$ (cf. \cite{AbbottDeser}).
    
    Loosely speaking, we can conclude that only \emph{Abelian} currents can be coupled as external sources to a YM gauge field theory.
    \end{example}

    \section{Phase Spaces in Field Theory}\label{sec:PS}
    This section is dedicated to the concept of phase space in field theory. We will define and compare two alternatively valid notions of phase space, which are (expected to be) compatible. This means that when all symmetries are reduced, the two alternative procedures should return the same result. This is already provably true for a large class of cases, but---to our knowledge---a fully general proof of this statement is still lacking.

    We argue that a choice of Lagrangian $\bL$ together with a choice of homotopy $(\bbh,h^\geq)$, yields a Hamiltonian dynamical system under suitable assumptions. We will focus, in particular, on cases where the symmetries can be realised in a ``Hamiltonian'' fashion. 
    
    \subsection{The Covariant Phase Space}\label{sec:covPS}
    
    Consider a Lagrangian $\bL$, and note that Theorem \ref{thm:horhomotopy} applied to $\d\bL$ states that $\d \bL = d h^\geq(\d\bL) + \bbI(\d\bL) \equiv dh^\geq(\d \bL) + \bbE(\bL)$. Recall also that the Euler-Lagrange set is defined as $\cEL(\bL) = \mathrm{Zero}(\bbE(\bL)) \subset \cE$.

    \begin{definition}
    Let $\bL\in\oloc^{0,\top}(\cE\times M)$ and $h^\geq$ be a homotopy as in Theorem \ref{thm:horhomotopy}, then
    \begin{enumerate}
        \item $\btheta^h\doteq {h}^{\geq1}\bd \bL$ is the \emph{pre-symplectic potential current} of $\bL$;
        \item $\bom^h\doteq \d\btheta^h$ is the \emph{pre-symplectic current} of $\bL$.
    \end{enumerate}
    \end{definition}

    \begin{theorem}\label{thm:dbom}
    One has
        \[
        d\bom^h = \d \bbE\bL,
        \]
    whereby, if $\iota:\cEL(\bL)\into \cE$ is smooth, $\iota^*\bom^h$ is conserved, i.e.\ $d\iota^*\bom^h =0$.
    \end{theorem}
    \begin{proof}
    We compute: 
    \[
    d\bom^h = - \d d\btheta^h = \d \bbE\bL^h - \cancel{\d^2\bL }.
    \]
    This means that $d\bom^h$ vanishes when pulled-back to the Euler-Lagrange set $\cEL$.
    \end{proof}

    \begin{remark}
        If $\cEL$ is not smooth, the second part of the theorem can be replaced by the following: $d\bom^h|_{\bar\phi}(\mathbb{X}_1,\mathbb{X}_2)=0$ if $\bar\phi\in\cEL$ and the $\mathbb{X}_i \in T_{\bar\phi}\cE$ are Jacobi fields, i.e.\ solutions to the linearised equations of motion at $\bar\phi$ \cite{LeeWald,Blohmann}.
    \end{remark}

    This result suggests the following definition
    (to be discussed below):

    \begin{definition}[Phase submanifold]\label{def:achronal}
    Henceforth, let 
    \[
    \iota_\Sigma : \Sigma \into M
    \]
    be a codimension-1 submanifold, possibly with boundary. Unless explicitly stated, whenever $M$ is equipped with a Lorentz structure we assume $M$ globally hyperbolic and $\Sigma$ achronal\footnote{Recall: $S\into M$ is achronal if all pair of points on $S$ \emph{cannot} be connected by a timelike curve within $M$.}, with non-empty Cauchy development\footnote{The Cauchy development $D(S)$ of a subset $S\subset M$ is the set of points $p$ such that all inextensible causal curves through $p$ intersect $S$.} $\emptyset \not= D(\Sigma)\subseteq M$. 
    We call $\Sigma$ the \emph{phase submanifold}.

    A phase submanifold such that $D(\Sigma)=M$ is called a \emph{Cauchy Surface}.
    \end{definition}

\begin{example}[Counterexamples]\label{rmk:degenerateSigma}
Note the following scenario in Minkowski space $\mathbb{R}^{1,3}$, which admit obvious generalisations to flat and curved spacetimes of different dimensions. If $\iota_\Sigma: \Sigma\hookrightarrow M$ is either the null hyperplane $\{z=t\}$ or the truncated null cone $\{r = t\}\cap \{ 1 \leq t\leq 2\}$, then its Cauchy development is empty, $D(\Sigma) = \emptyset$. In this cases, although achronal, $\Sigma$ fails to be a phase submanifold. This degeneracy can be lifted e.g. by ``completing" $\Sigma$ to include either the plane $\{z=-t\}$ or the the spacelike ball $\{r\leq 1\}\cap\{t=1\}$, respectively, in which cases the new Cauchy developments become the entire Minkowski space $\mathbb{R}^{1,3}$ and the double cone $\{t\geq r\}\cap\{t\leq 2-r\}$.
\end{example}

    \begin{definition}[(Homotopy) Covariant Phase Space]\label{def:cps}
        Let $\Sigma \into M$ be a phase submanifold as in Definition \ref{def:achronal}.
        Assume that $\iota : \cEL(\bL)\into \cE$ is smooth. 
        The \emph{(homotopy) covariant phase space} associated to $\Sigma$ and a choice of homotopies $h=(h^\geq,\bbh)$ is the \emph{pre}-symplectic\footnote{Some definitions require $\ker(\omega^h_{\cEL|\Sigma})$ to be a \emph{regular} integrable foliation. We forgo this restriction and allow $\ker(\omega^h_{\cEL|\Sigma})$ to be a singular foliation.} manifold $(\cEL(\bL), \omega^h_{\cEL|\Sigma})$ where $\omega^h_{\cEL|\Sigma} \doteq \iota^*\int_\Sigma\bom^h$.

        The \emph{symmetry-reduced (homotopy) covariant phase space} is the quotient of $\cEL(\bL)$ by the action of symmetries $\cEL(\bL)/\fG$.
    \end{definition}

\begin{remark}[Dependence on choices and notation]
    In a Lorentz-covariant theory, Theorem \ref{thm:dbom} suggests that the covariant phase space associated to $\Sigma$ is in fact best thought as associated to $D(\Sigma)$. 
    Indeed, from Theorem \ref{thm:dbom} it readily follows that two homotopic $\Sigma$ and $\Sigma'$ such that $D(\Sigma)=D(\Sigma')$ yield the same covariant phase space.  If moreover $\pp\Sigma=\emptyset$, then one can verify that the associated covariant phase space $(\cEL,\omega^h_{\cEL|\Sigma})$ is independent of the chosen (horizontal) homotopy $h$ on the complex of local forms.
    
    More generally, a choice of homotopy $h$ will give us consistent data in the horizontal cone so that changing the homotopy will give us a different presentation of the same structures in cohomology.
    Since we postpone to another work the study of how the constructions that follow vary at the change of the choice of homotopy, we will henceforth drop the homotopy superscript $\bullet^h$ from $\bom^h$, $\omega_{\cEL\vert\Sigma}^h$, etc.
\end{remark}

The requirement that $\cEL(\bL)\into \cE$ be a well defined (at least immersed) submanifold is in general a restrictive assumption that might be hard to verify, since it has to do with the regularity of solutions to PDE's. This is an obvious drawback of the covariant phase space approach, considered as an honest presymplectic manifold. As we will see in Section \ref{sec:BV} this assumption is superseded by a shift in approach, which aims at cohomologically describing $\cEL$ (and the associated quotient by possible symmetries, which takes care of presymplectic reduction \cite{HenneauxTeitelboim}) in terms of a differential graded manifold. 

Moreover, even when it is smooth, $\cEL(\bL)$ is generically not a local space of fields over $M$, i.e.\ it is not a space of sections of a bundle over $M$.
However, using the correspondence between solutions to the Euler-Lagrange equations and the set of their initial conditions, $\cEL$ is expected to be isomorphic to a space of fields over a codimension-1 surface $\Sigma$. Although this idea works \emph{directly} only in the absence of local symmetries (see below), the construction has interesting implications even in the general scenario.

An important motivation that leads us to the study of the symmetry-reduced covariant phase space $\cEL(\bL)/\fG$ is the following argument found in \cite{HenneauxTeitelboim}.

    Working formally\footnote{To make the argument rigorous, we can either consider $M$ Euclidean and closed or, for a globally hyperbolic $M$, we can introduce generalised actions $S
    (f)= \int_M f \bL$, for $f$ a compactly supported test function over $M$. For details on the latter approach see \cite{RejznerBook}.}, consider the action functional $S=\int_M \bL$ with $\phi_0$ a critical point and $\rho(\xi)$ a symmetry of $S$, i.e.\ $\d S|_{\phi_0}=0$ and $\rho(\xi)(S) = 0$. Then
    \[
    \rho(\xi)^{\phi_a}\frac{\pp S}{\pp\phi_a}\ = 0\quad  \implies\quad \frac{\pp}{\pp\phi_b}\left(\rho(\xi)^{\phi_a}\frac{\pp S}{\pp\phi_a}\right) = 0 \quad \implies\quad \rho(\xi)^{\phi_a}\frac{\pp^2 S}{\pp\phi_a\pp\phi_b}\Big\vert_{\phi_0}=0
    \]
    and the Hessian of $S$ is degenerate on symmetry vectors $\rho(\xi)$ on $\cEL$. This is of course just a different formulation of Theorem \ref{thm:inv-eom}.
    Now, the (inverse of the) Hessian of $S$ is crucially featured in any (rigorous) definition of the path ``integral'' of the field theory, e.g. via stationary phase formula\footnote{When a proper integration theory is not available, the saddle point approximation theorem becomes a definition.} where one needs to invert the Hessian at a critical point. This is another way of noting that the path ``integration'' should happen only after quotienting by symmetries, or along a (local) section of the quotient map $\cEL(\bL)\to \cEL(\bL)/\fG$ (a.k.a.\ ``gauge fixing").

    The covariant phase space of a theory modulo local symmetries can be difficult to access directly, as it requires the solution of PDE's on general manifolds---which is a challenging problem in analysis---and because, consequently, it is nonlocal in general---which is a challenging problem in geometry. Therefore, it is often convenient to construct it by means of a reduction procedure starting from another ``proxy'' symplectic space, which we will build in the next section. This is called the geometric phase space associated to the Cauchy surface $\Sigma$, and---differently from $\cEL$---it is local.

\subsection{The Geometric Phase Space}\label{sec:covPS2}
The next space of fields we are going to build is called the \emph{geometric phase space}. Morally, this is a (local) symplectic manifold, which the theory induces given a choice of a (phase) submanifold $\Sigma\into M$ \emph{before} any restriction to the Euler-Lagrange set $\cEL$.
In other words, we aim to define a (local) smooth manifold of $\Sigma$-fields $\mathcal{E}_\Sigma = \Gamma(\Sigma, E^h_\Sigma)$ for some\footnote{In general, this is not the pullback bundle $\iota_\Sigma^*E$, but a larger bundle, see the examples below.} $E^h_\Sigma \to \Sigma$  equipped with a (weak) symplectic 2-form $\omega^h_\Sigma$.

In order to achieve this, we follow a procedure due to Kijowski and Tulczyjew (KT), which, when unobstructed, produces a surjective submersion
\[
    \pi^h_\Sigma:\cE \to (\cE^h_\Sigma,\omega^h_\Sigma)
\]
such that the 2-form $\omega^h_\Sigma$ defined by 
\[
    (\pi^h_\Sigma)^*\omega^h_\Sigma \doteq \int_\Sigma \bom^h,
\]
is non-degenerate, i.e.\ $(\omega_\Sigma^h)^\flat : T\cE^h_\Sigma \to T^*\cE^h_\Sigma$ is injective \cite{KT,CattaneoPhaseSpace}.

To be more specific, one can always construct some smooth space of restricted field configurations $\wt{\cE}^h_\Sigma$ over $\Sigma$ endowed with the form $\wt{\omega}^h_{\Sigma}$, characterised by $\wt\pi_\Sigma^*\wt{\omega}^h_{\Sigma}= \int_{\Sigma}\bom^h$. In good cases, this form is pre-symplectic and has a smooth (pre-)symplectic reduction, which we then denote by $(\cE^h_\Sigma,\omega^h_\Sigma)$ and call \emph{geometric phase space}. 
It is moreover often possible (but cumbersome) to show that $\cE^h_\Sigma$ is local, i.e.\ it a space of sections of a fibre bundle on $\Sigma$. 

Note that the pre-symplectic reduction involved in the definition of the geometric phase space {\`a} la Kijowski--Tulczyjew is much simpelr and better behaved than that involved in the definition of the covariant phase space. Indeed, for many field theories of physical interest, including general relativity, the geometric phase space is well defined as a smooth Fréchet manifold, see e.g.\ \cite{CS_EH,CS_PCH,CattaneoPhaseSpace}.

\begin{example}\label{ex:scalarfieldex}
    For concreteness let us mention a simple but physically relevant example.
    Consider the scalar field theory $\bL = - \frac12 d\phi \wedge \star d\phi + V(\phi)\star1$ on $\cE = C^\infty(M) \ni \phi$ the space of sections of $\pi_2 : \mathbb{R}\times M \to M$. Then, $\btheta^h = \star d\phi \wedge \d \phi$ and $\int_\Sigma  \bom^h = \int_\Sigma (\d(\star d\phi) \wedge \d \phi)$. 
    
    If $\Sigma$ is compact and spacelike, the geometric phase space is then $\cE_\Sigma \simeq T^\vee C^\infty(\Sigma)$, that is, the space of sections of the symplectic fibre bundle $\pi_2 : (T^*\mathbb{R})\times \Sigma \to \Sigma$ tensored with the densities over $\Sigma$, which is nuclear Fréchet (see \cite{RSATMP} and therein). Considering as usual densitised momenta, one has the isomorphism $\cE_\Sigma \simeq \Omega^\top(\Sigma)\times \Omega^0(\Sigma)$ and thus the projection $(\Pi,\Phi) = \pi_\Sigma(\phi) = (\iota_\Sigma^* (\star d\phi) , \iota_\Sigma^*\phi)$; whence $\omega^h_\Sigma = \int_\Sigma \d \Pi \wedge \d \Phi$. 
    
    Conversely, if $\Sigma \simeq S \times [-1,1] \ni (y,u)$ is a null hypersurface ruled by $\ell = \pp_u$ and with $S$ closed and compact, one finds a different nuclear Fréchet space, namely $\cE_\Sigma \simeq C^\infty(\Sigma)$ with $\Psi = \pi_\Sigma(\phi) = \iota_\Sigma^*\phi$ and $\omega_\Sigma^h = \int_\Sigma \mathbf{vol}_S \ \pp_u\d \Psi \wedge \d \Psi$, where $\mathbf{vol}_S= \iota_\Sigma^* (i_\ell (\star 1))$ is the induced volume.
    Generalisation to Yang-Mills theories is computationally more involved but conceptually straightforward \cite{RSAHP}. (Note: one of the first and main references on the topic is \cite{AshtekarStreubel}, where gravity is treated as well; however, this work misses on certain zero-modes of the reduced symplectic structure. The omission is relevant in physical application \cite{AshtekarBonga2017}). 
\end{example}

To wrap up, we combine the definition of the geometric phase space with a mild assumption:

\begin{assumption}[Geometric Phase Space]\label{ass:GPS-local}
    The (homotopy) geometric phase space $(\cE^h_\Sigma,\omega^h_\Sigma)$ of the field theory $(\cE,\bL)$ associated to the phase submanifold $\Sigma\hookrightarrow M$ is a \emph{locally symplectic} space (\cite[Def. 3.1]{RSATMP}) where $\omega^h_\Sigma= \int_\Sigma \bom^h_\Sigma$, with $\pi_\Sigma^*\bom^h_\Sigma = \iota_\Sigma^*\bom^h$, and $\pi_\Sigma\colon \cE \to \cE^h_\Sigma$ a surjective submersion. Once again, we will drop the homotopy superscript $\bullet^h$ in what follows.
\end{assumption}

    We conclude this section with a side note.
    The KT procedure can be used to determine local data for any codimension-$1$ submanifold $\Sigma\hookrightarrow M$, whether achronal or not. However, for \emph{non}-achronal surfaces the physical interpretation of $(\cE_\Sigma,\omega_\Sigma)$ changes.
    For instance, consider $\Sigma = \pp M\not=\emptyset$.
    Already in the simple example of $M=[0,1]\times \Sigma$ with $\pp\Sigma=\emptyset$, we have that $\pp M=\Sigma_1\sqcup\Sigma_0$ and therefore $\cE_{\pp M} \simeq \cE_{\Sigma_1}\times \cE_{\Sigma_0}$ with the symplectic forms defined with opposite orientation, $\omega_{\pp M} = \omega_{\Sigma_1} - \omega_{\Sigma_0}$.  Then, in the example of a scalar field theory discussed above, the restriction of fields in the Euler-Lagrange set at initial and final times, $\cEL \to \cE_{\pp M}$, is not an isomorphism, but defines instead a Lagrangian submanifold of $\cE_{\pp M}$ \cite{CattaneoMnev_wave,CattaneoPhaseSpace} (or, more generally, Lagrangian relations \cite{ContrerasPSM}). This is completely analogous to what happens in the guiding example of classical mechanics as discussed at the beginning of Section \ref{sec:LFT}.

\subsection{Covariant vs. Geometric phase space}\label{sec:covVgeom}

One of the main motivations to  define the geometric phase space is to have a local proxy for the covariant phase space. Consider the following relation between $\cEL$ and the data in $\cE_\Sigma$ expressed via the commutative diagram
    \begin{equation}\label{eq:comm-diag-1}
    \xymatrix{
    \cEL \ar[d]_{\pi_{\Sigma|\cEL}} \ar[r]^\iota & \cE \ar[d]^{\pi_\Sigma} \\
    \cEL_{\Sigma} \ar[r]^{\wt\iota} & \cE_\Sigma
    }
    \end{equation}
together with
    \[
    \omega_{\cEL|\Sigma}=\iota^*\int_\Sigma \bom = \iota^*\pi^*_\Sigma\omega_\Sigma = (\pi_\Sigma\circ \iota)^*\omega_\Sigma = (\wt\iota\circ\pi_{\Sigma|\cEL})^*\omega_\Sigma = \pi_{\Sigma|\cEL}^*\wt\iota\,{}^*\omega_\Sigma.
    \]
Note that $\cEL_\Sigma \subset \cE_\Sigma$ are the $\Sigma$-fields compatible with the equations of motion.

In the absence of local symmetries, for a theory characterised by hyperbolic Euler-Lagrange equations $\bbE(\bL)=0$, one expects the covariant and geometric phase spaces to be isomorphic, $\cEL \simeq \cE_\Sigma$, when $\Sigma$ is a Cauchy surface, i.e.\ $M=D(\Sigma)$ \cite{Peierls, souriau1970structure, AshtekarMagnon1982, Zuckerman,crnkovicwittencovphasespace, Crnkovic_1988, LeeWald, Khavkine2014,BentabolVillasenor}.

From the diagram, this translates to the expectation that both $\pi_{\Sigma|\cEL}\colon \cEL \to \cEL_\Sigma$ and $\wt\iota\colon \cEL_\Sigma\to\cE_\Sigma$ are isomorphisms.
More specifically, the first expectation has to do with solutions of $\bbE(\bL)=0$ being in 1-to-1 correspondence with the set of initial data at $\Sigma$, while the second expectation can be phrased as the absence of ``constraints" governing which local fields in $\cE_\Sigma$ provide viable initial conditions for a solution of the equations of motion. Note that for this to make sense one needs $\Sigma\into M$ to be achronal.

\begin{example}
    Back to the example of the scalar field of Example \ref{ex:scalarfieldex} on globally hyperbolic cylinders $M=\Sigma \times [0,1]$, with $\pp\Sigma=\emptyset$, one readily observes that the restriction of all possible solutions of the Klein-Gordon equation for the field $\phi$ onto a spacelike $\Sigma$ (a connected component of the boundary $\pp M$) yields exactly all possible pairs $(\Phi,\Pi)$ of a scalar $\Phi=\phi|_\Sigma$ and its ``normal derivative'' $\Pi\equiv \dot{\phi}\vert_{\Sigma}$, so that $\cEL\simeq\cEL_\Sigma \simeq \cE_\Sigma \simeq T^\vee C^\infty(\Sigma)$ (see \cite{CattaneoMnev_wave}). 
\end{example}

In the cases where $\cEL$ is only presymplectic, one can thus distinguish two distinct origins for the kernel of $\omega_{\cEL|\Sigma}$, related to the failure of injectivity of $\pi_{\Sigma|\cEL}$ and surjectivity of $\wt\iota$, respectively.

In the absence of local (gauge) symmetries, the lack of injectivity of $\pi_{\Sigma|\cEL}$ is typically due to failure of the phase surface $\Sigma$ to be Cauchy, i.e.\ $D(\Sigma) \subsetneq M$. However, in this case, one still expects $\wt\iota$ to be bijective so that $\cEL_\Sigma \simeq \cE_\Sigma$ and symplectic.

In a theory with \emph{local} (gauge) symmetries, on the other hand, both expectations fail to be met even if $\Sigma$ is Cauchy. Indeed, the freedom to perform symmetry transformations with support disjoint from $\Sigma$ tells us that $\pi_{\Sigma\vert \cEL}$ cannot be injective. On the other hand, the presence of constraints, such as the Gauss constraint in Yang-Mills theories or the vector and scalar constraints in General Relativity, affirms that not all configurations in $\cE_\Sigma$ can be chosen as initial values for solutions in $\cEL$. These two facts are related.

Indeed, Theorem \ref{thm:hamflow-closed} below says that if $\Sigma$ is a \emph{closed} Cauchy surface, namely $\pp\Sigma=\emptyset$ and $M=D(\Sigma)$, and $\fG$ is a \emph{local} symmetry, then $\rho(\xi)$ is in the kernel of $\omega_{\cEL|\Sigma}=\omega_{\cEL}$. 
The reason for this is that---in the absence of boundaries---the integral of the Noether current is constrained to take on a value fixed by the external current $j$ (see Noether's second theorem \ref{thm:N2}).
This relationship among local symmetries, coisotropic submanifolds of $(\cEL(\bL),\omega_\cEL)$, and the (integrals of) the Noether current, foreshadows the paradigm of Hamiltonian reduction (Sections \ref{sec:hamiltoniantheory} and \ref{sec:hamwithcorners}).

Generally speaking, however, we do not know whether the symmetry distribution $\Im(\rho)$ exhausts the kernel of the presymplectic form, i.e.\ we do not know whether $\cEL(\bL)/\fG$ has symplectic properties. Something we can say in fair generality is that, whenever one can view the $\fG$ action as Hamiltonian on the \emph{pre}-symplectic manifold $\cEL(\bL)$, with the Noether current as momentum map (in a sense to be explained below), one can perform Hamiltonian reduction.

In order to compare the covariant and the geometric phase spaces effectively, it is convenient to distinguish four scenarios, labelled by two binary numbers, based on whether $\Sigma$ as boundaries or not, and whether $\Sigma$ is Cauchy or not. We assume $\iota:\cEL \hookrightarrow\cE$ smooth.
    \begin{itemize}
        \item[(0,1)] \emph{$\Sigma$ is a Cauchy and has no boundary.} 
        By Theorem \ref{thm:hamflow-closed} $\cEL$ is \emph{pre-}sympelctic and the action $\rho$ is contained in the characteristic distribution of $\omega_\cEL$. The action $\rho$ is Hamiltonian and the momentum map, given by the integral of the Noether current, takes \emph{only one} value determined by $j$. When this is a regular value, Hamiltonian reduction by the symmetries yields a $\cEL/\fG$ which is (expected to be) a symplectic, reduced, covariant phase space.
        
        In good cases $\rho$ descends to a Hamiltonian action by some $\fG_\Sigma$ on the geometric phase space $\cE_\Sigma$, whence one expects $\cEL/\fG$ to be symplectomorphic to the Hamiltonian reduction $\cE_\Sigma//\fG_\Sigma$.
        \item[(1,0)] \emph{$\Sigma$ is not Cauchy, but it has a boundary.} Only the evolution of fields in the Cauchy development $D(\Sigma)$ can be determined by configurations at $\Sigma$, hence the kernel of the pre-symplectic form $\omega_{\cEL|\Sigma}$ on $\cEL(\bL)$ contains degrees of freedom to account for this mismatch, which are not related to any symmetry consideration. Moreover, the value of the momentum map is not anymore constrained by $j$. This is because the integral of the Noether current now picks boundary terms, called \emph{fluxes}, which do not necessarily vanish. This changes the Hamiltonian reduction paradigm and one expects $\cEL/\fG$ to be foliated by presymplectic manifolds, one per every value of the fluxes. This structure yields, in general, a Dirac structure.\footnote{This is true up to gluing conditions on the presymplectic structures. We thank J.\ Schnitzer for this comment.}
        \item[(1,1)] \emph{$\Sigma$ is a Cauchy submanifold with boundary.} This is an improvement of the case above where one focuses on a causal domain $D$ of $\Sigma$ in a larger spacetime (cf. Figure \ref{fig:causalregions}) and thinks about it as a ``whole universe.'' In the absence of local symmetries one expects $\cEL$ to be symplectomorphic to $\cE_\Sigma$, owing to the Cauchy nature of $\Sigma$. In the presence of local symmetries, one expects, the reduction  $\cEL/\fG$ to yield an honest Poisson manifold in the form of a foliation of \emph{symplectic} manifolds \cite{RSATMP}. 
        If furthermore $\rho$ descends to $\cE_\Sigma$, each leaf is expected to be symplectomorphic to the Hamiltonian reduction of $\cE_\Sigma$ by $\fG_\Sigma$ at one particular (orbit of a) value of the flux. See Section \ref{sec:hamiltoniantheory} for an analogue scenario in mechanical systems, and Section \ref{sec:remarks} for a more in-depth discussion of this case.
        \item[(0,0)] \emph{$\Sigma$ is not Cauchy and has no boundary.} We discard this case, as it relates to non-connected domains.
    \end{itemize}

    The goal of the next two sections, \ref{sec:hamiltoniantheory} and \ref{sec:hamwithcorners}, is to construct these spaces as explicitly as possible. To do so, we focus on the  ``good cases" (cf. Assumptions \ref{ass:GPS-local}, \ref{ass:local-ham} and \ref{ass:Ho-equi}) analysed from the perspective of the geometric phase space.

    \subsection{Local symmetries in phase space}
As we mentioned, in the absence of local (gauge) symmetries, $\cEL$ and $\cE_\Sigma$ are expected to be isomorphic---but not so in their presence. A direct way to showcase this failure is to observe that $\omega^h_{\cEL|\Sigma}$ is degenerate, i.e.\ only \emph{pre}-symplectic. To show this, we start by investigating the properties of the Noether current $\bJ_\bullet$ in relation to $\bom$.

    \begin{lemma}\label{lemma:flow-density}
    For all $\xi\in\fG$, 
    \[
    \i_{\rho(\xi)}\bom = -\d \bJ_\xi + dh^\geq ( \i_{\rho(\xi)}\bom + \d \bJ_\xi ) -  h^\geq\L_{\rho(\xi)}\bbE\bL.
    \]
    \end{lemma}
    \begin{proof}
    From Noether's first theorem \ref{thm:N1} and the on-shell conservation of $\bom$ (Theorem \ref{thm:dbom}), we have
    \[
    d(\i \bom + \d \bJ) = - \i d\bom - \d d\bJ =- \i \d \bbE \bL - \d (\cancel{-p^*S} + \i\bbE\bL) = - \L\bbE\bL.
    \]
    The lemma follows by applying $h^\geq$, using the homotopy equation $\id = dh^\geq + h^\geq d + \bbI$ on $\i \bom + \d \bJ$, and recalling that $\bbI$ annihilates $(\geq1,\top-1)$ forms.
    \end{proof}

It can be instructive to rewrite the right-hand side of the previous lemma. Recalling Definition \ref{def:NoetherCurrent} of the Noether current $\bJ_\bullet$ \eqref{eq:Jh}, we see that it can be rewritten as:
    \[
        \bJ_\bullet = \i_{\rho(\bullet)}\btheta + \mathbf{B}_\bullet, 
    \]
where $\mathbf{B}_\bullet \doteq h^0 \L_{\rho(\bullet)}\bL$ captures the $d$-exact term by which $\bL$ fails to be invariant under the action of the symmetry.
With this notation one can use the following version of the above lemma:\footnote{Viz. $\i \bom =  - \d \i \btheta + \L\btheta = - \d \bJ + \d \mathbf{B} + \L\btheta$. }
    \[
    \i_{\rho(\xi)}\bom = -\d \bJ_\xi + dh^\geq ( \d \mathbf{B}_\xi + \L_{\rho(\xi)}\btheta ) -  h^\geq\L_{\rho(\xi)}\bbE\bL.
    \]
This can be useful since, at times, it is fairly easy to compute the variation of the Lagrangian density (giving $\mathbf{B}_\xi$) and that of the presymplectic potential current (i.e.\ $\L_{\rho(\xi)}\btheta$), especially when these quantities vanish. 

Most importantly, from the lemma we immediately deduce the following theorem and its consequence in the covariant phase space:
    
    \begin{theorem}\label{thm:hamflow-closed}
        Let $\Sigma\into M$ be a closed phase submanifold, $\pp\Sigma=\emptyset$. Then, 
        \[
        \left( \i_{\rho(\xi)}\textstyle{\int_\Sigma} \bom + \textstyle{\int_\Sigma} \d\bJ_\xi\right)_{|\cEL } =0.
        \]
        Moreover, if $\rho$ is a local symmetry 
        \[
        \rho(\xi) \in \ker \left(\textstyle{\int_\Sigma} \bom _{|\cEL } \right).
        \]
    \end{theorem}
    \begin{proof}
        The first part follows from the previous lemma and Theorem \ref{thm:inv-eom}. The second part follows from Noether's two theorems \ref{thm:N1} and \ref{thm:N2}: for then $\bJ =\mathbf{C} + d\mathbf{K} $ with $\mathbf{C}|_\cEL = j$ and $\d j = 0$. Indeed, from these two formulas, using that $\pp\Sigma =\emptyset$, we find $\int_\Sigma \d \bJ |_\cEL = \int_\Sigma \d \bC|_\cEL = \int_\Sigma \d j = 0$.
    \end{proof}
    \begin{corollary}
    Let $\Sigma\into M$ be a closed phase submanifold, $\pp\Sigma=\emptyset$. Assume $\iota\colon \cEL(\bL)\hookrightarrow\cE$ is smooth, the covariant phase space is a presymplectic manifold $(\cEL(\bL), \omega^h_{\cEL|\Sigma})$ and the Noether Charge $Q_\Sigma(\xi) \doteq \iota^*\int_\Sigma \bJ_\xi$ is the Hamiltonian function of the symmetry vector field $\rho(\xi)$.\footnote{In physics lingo one says that $Q$ is the \emph{generator} of the symmetry $\rho$.}
    If the symmetry is local, $Q_\Sigma$ vanishes and $\ker(\omega^h_{\cEL|\Sigma})\neq\emptyset$.
    \end{corollary}

    We thus deduce that, in the presence of local symmetries, while the covariant phase space carries a degenerate 2-form $\omega_{\cEL|\Sigma}$, the 2-form $\omega_\Sigma$ on the geometric phase space is by construction non-degenerate. Therefore, the two spaces cannot be isomorphic. As the proof of Theorem \ref{thm:hamflow-closed} shows, the existence of local symmetries is intimately tied through Noether's second theorem to the constraint equation $\mathbf{C}_\bullet|_{\cEL} = j_\bullet$. 
    \begin{definition}[Constraint set]\label{def:constraintset}
        Let us denote 
        \[
        \cE_{j}\doteq \{\phi\in \cE\ |\ \mathbf{C}_\bullet(\phi) = j\} \subset \cE
        \] 
        and 
        \[\cC_\Sigma\doteq \pi_\Sigma(\cE_j).
        \]
        The set $\cC_\Sigma$ is called the  Constraint Set of the field theory. A field theory is said to be \emph{regular} if $\cC_{\Sigma}\hookrightarrow \cE_\Sigma$ is a coisotropic embedding.
    \end{definition}
    
    The commutative diagram \eqref{eq:comm-diag-1} can be refined to give 
    \begin{equation}\label{eq:diagram-C}
    \xymatrix{
    \cEL \ar[d] \ar[r] & \cE_j \ar[r]\ar[d] & \cE \ar[d]^{\pi_\Sigma} \\
    \cEL_{\Sigma} \ar[r]^-{(\simeq)} & \cC_\Sigma \ar[r] & \cE_\Sigma
    }
    \end{equation}
    showing that not all field configurations in $\cE_\Sigma$ are compatible with the equations of motion in $M$, and one should instead look at $\cC_\Sigma$. Indeed, one expects $\cEL_\Sigma \simeq \cC_\Sigma$.

As hinted at the beginning of Section \ref{sec:covVgeom}, other typical expectations are that the local symmetries exhaust the kernel of $\omega_{\cEL|\Sigma}$ (in infinite dimension this statement must be qualified, see \cite{DiezPHD}), and that 
\begin{equation}\label{eq:can-vs-cov-reduction}
\cEL/\ker(\omega_{\cEL|\Sigma}) = \cEL/\fG \simeq \cC_\Sigma/\ker(\iota_\cC^*\omega_\Sigma).
\end{equation}
Note that if this expectation holds, it implies that $\cC_\Sigma\subset \cE_\Sigma$ is a coisotropic submanifold.

In other words, one expects that although the pre-symplectic manifolds $\cEL$ and $\cC_\Sigma$ are not isomorphic, their symmetry reductions are.

While these expectations can be safely turned into assumptions, for they are verified in all cases of physical interested we are aware of, it is certainly not warranted to assume that the characteristic distribution of the constraint set $\cC_\Sigma$ is itself the action of a symmetry group or algebra. This requires that $\rho:\fG\to\mathfrak{X}(\cE)$ descends to an action $\rho_\Sigma : \fG_\Sigma \to \mathfrak{X}(\cE_\Sigma)$ -- which holds, e.g., in Yang-Mills and Chern-Simons theories, but notably fails in General Relativity \cite{Teitelboim1973,HKT1976, LeeWald, BFW, BSW}.

    An instance where this failure can arise is if the Euler-Lagrange form $\bbE\bL$ fails to be symmetry-invariant on the nose (as opposed to on-shell, cf. Theorem \ref{thm:inv-eom}). Indeed, $\bbE\bL$ typically involves second derivatives of the configuration variables which are not expressible in terms of the canonical variables alone (positions and velocities). For this reason, when non-vanishing, the quantity $\L_{\rho(\xi)}\bbE\bL$ usually fails to be basic with respect to the boundary map $\cE \to \cE_\Sigma$, strongly hinting to the fact that the right hand side of Lemma \ref{lemma:flow-density} will likely not descend to $\cE_\Sigma$. In this case, then, the left hand side won't descend either. Since $\bom$ descends by (a mild) assumption, this means that it is $\rho$ which won't descend to $\cE_\Sigma$. For the details of this scenario within General Relativity, see \cite{LeeWald}. In the next remark we explain why $\bbE\bL$ is not invariant in General Relativity.

    \begin{remark}\label{rmk:diffeos-not-descend}
    In General Relativity, or (essentially by definition) in any other background-independent theory with diffeomorphism symmetry $\fG =\mathfrak{diff}(M) \simeq \mathfrak{X}(M)$, one has that any local form $\boldsymbol{\alpha}$ built exclusively out of the dynamical fields (background independence) satisfies\footnote{This equation can be broken if $\boldsymbol{\alpha}$ is defined with reference to auxiliary, non-dynamical, structures, of the kind used to introduce, e.g., asymptotic boundary conditions \cite[Sect. 7.6]{FreidelRiello2024}.}
    \[
        \L_{\rho(\xi)}\boldsymbol{\alpha}= L_\xi \boldsymbol{\alpha}
        \quad \forall \xi \in \mathfrak{diff}(M),
    \]
    where $L_\xi = [d,i_\xi]$.
    Thus, $\L_{\rho(\xi)}\bbE\bL = L_\xi \bbE\bL = d i_\xi \bbE\bL\not\equiv 0$. Note that upon pullback to $\Sigma$, this expression vanishes whenever $\xi$ is tangent to $\Sigma$. 
    
    Similarly, the term $dh^\geq(\d \mathbf{B}_\xi + \L_{\rho(\xi)}\btheta) )= dh^\geq (\d i_\xi \bL + L_\xi \btheta)$ will fail to vanish when integrated over $\Sigma$ whenever $\xi$ is transverse to $\pp\Sigma$, thus hinting at the fact fact the failure of the action $\rho$ to descend to an action $\rho_\Sigma$ is due to the  diffeomorphisms transverse to $\Sigma$.
    \end{remark}    

    At the end of section \ref{sec:covPS} we have explained that the geometric phase space is built as a proxy for the covariant phase space, which in the presence of local symmetry is recovered after reduction. In order to investigate this construction, we now turn our attention to a particularly simple scenario---one where the symmetries can be realised in a Hamiltonian fashion. This will allow us to infer precise information on the set $\cC_\Sigma$ and its reduction, at the expense of some generality. However, although restrictive, this scenario still underpins many examples of interest such as Yang-Mills, $BF$, and Chern-Simons theories. We will refer to Section \ref{sec:BV}, and in particular Section \ref{sec:BFV}, for a more general approach.

\subsection{Local Hamiltonian field theory - the canonical theory.}\label{sec:hamiltoniantheory}

Henceforth, on top of Assumption \ref{ass:GPS-local}, we will also assume that 
\begin{assumption}\label{ass:local-ham}
~
\begin{enumerate}
    \item The strong local Hamiltonian flow equation $\i_{\rho(\xi)}\bom = - \d \bJ_\xi$ holds $\forall \xi\in \fG$.\label{ass:localham-1}
    \item  $\pi_\Sigma$ extends to a surjective submersion of (local, action) Lie algebroids:
    \[
    (\mathsf{A}\doteq\cE\times \fG, \rho, M) \to (\mathsf{A}_\Sigma\doteq\cE_\Sigma\times \fG_\Sigma, \rho_\Sigma, \Sigma),
    \]
    with $\pi_\Sigma^\fG\colon \fG \to \fG_\Sigma \doteq \fG/\ker(\iota_\Sigma^*\circ \check\bJ)$.\label{ass:localham-2}
\end{enumerate}
\end{assumption}

\begin{remark}
    If one insists on $\mathsf{A}_{\Sigma}$ being an \emph{action} Lie algebroid, one needs to require $\mathfrak{N}\doteq\ker(\iota_\Sigma^*\circ \check\bJ)\subset \fG$ to be an ideal. In diffeomorphism invariant theories (e.g.\ General Relativity) the space $\mathfrak{N}$ of vector fields that vanish at $\Sigma$ fails to be an ideal of $\mathfrak{diff}(M)\simeq \mathfrak{X}(M)$ \cite{BFW,blohmann2018hamiltonian}. This is related to the fact that point 2 fails entirely in such theories, even if one were to allow for more general Lie algebroids.\footnote{In other words, Dirac's Hypersurface Deformation Algebra does \emph{not} define a Lie algebroid on $\cE_\Sigma$ \cite{BFW,BSW}.} See Remark \ref{rmk:diffeos-not-descend}.
    
    Our construction can likely be generalised by demanding that the (local, action) Lie algebroid $\mathsf{A}_\Sigma$ be only defined over the pre-symplectic manifold $\cC_\Sigma \subset \cE_\Sigma$ instead of the entirety of $\cE_\Sigma$. This way one might be able to include General Relativity (see \cite[Section 5]{BSW})---this time at the expense of locality. We will not purse this generalisation here.
\end{remark}

\begin{lemma}\label{lem:Jdescends}
    Under Assumption \ref{ass:local-ham} there is a commuting diagram
    \[
    \xymatrix{
    \ar[d]^{\pi_\Sigma}\fG\times\cE \ar[r]^-{\bJ} & \Omega^{\top-1}(M) \ar[d]^{\iota^*_\Sigma}\\
    \fG_\Sigma\times\cE_\Sigma \ar[r]^-{\bH} & \Omega^\top(\Sigma)
    }
    \]
    Moreover, if $\rho$ is a local symmetry, $\bH$ uniquely decomposes into
    \[
    \bH = \bHo + d \bh \quad\text{with}\quad \bHo|_{\cC}={j_\Sigma}
    \]
    for ${j_\Sigma} \circ \pi_\Sigma^\fG \doteq \iota_\Sigma^* \circ j$ a constant $(\d {j_\Sigma}=0)$.
\end{lemma}
\begin{proof}
    From Assumption \ref{ass:local-ham}, it follows that the pullback to $\Sigma$ of the Noether current $\bJ$ seen as local map $\fG\times \cE\to \Omega^{\top-1}(M)$, i.e.\ $\iota_\Sigma^*\circ \bJ_\bullet$, is basic w.r.t the algebroid map $\pi_\Sigma:\cE \times \fG\to \cE_\Sigma\times \fG_\Sigma$. Indeed,
    \[
    \pi_\Sigma^*(\i_{\rho_\Sigma} \bom_\Sigma) = \pi_\Sigma^*(\i_{(T\pi_\Sigma) \rho}  \bom_\Sigma)
    = \i_{\rho+\mathbb{K}} \pi_\Sigma^*\bom_\Sigma = \i_\rho \iota_\Sigma^*\bom =  -\iota_\Sigma^* \d\bJ_\xi,
    \]
    where we used that the images of $\rho_{\Sigma}$ and $\rho$ are $\pi_\Sigma$-related via a surjective submersion, and that $\mathbb{K}$ is in the kernel of $T\pi_\Sigma$ and thus in $\ker(\iota_\Sigma^*\bom)$. As a consequence, contracting with a $\mathbb{K}$ in the kernel of $T\pi_\Sigma$, we find $\L_\mathbb{K}\iota_\Sigma^*\bJ_\xi=0$.

    Now, if $\rho$ is a local symmetry, according to Theorem \ref{thm:N2}, $\bJ$ decomposes into $\bJ = \mathbf{C} + d \mathbf{K}$. Then $\bHo$ and $\bh$ are respectively defined from $\mathbf{C}$ and $\mathbf{K}$ in the same way as $\bH$ is defined from $\bJ$, so that we have the desired decomposition. Finally, since $\mathbf{C}=j$ on $\cE_j$, $\bHo={j_\Sigma}$ on $\cC_\Sigma = \pi_\Sigma(\cE_j)$, cf. diagram \ref{eq:diagram-C}.
\end{proof}

    This allows us to introduce the following canonical versions of the Noether, Constraint, and Flux currents by means of a corresponding split for $\bH$:
    \begin{definition}\label{def:bH}
    The $\Sigma$-Noether form  is $\bH_\bullet : \fG_\Sigma \to \oloc^{0,\top}(\cE_\Sigma\times\Sigma)$, while its summands are the \emph{constraint form} ($\bHo$) and the \emph{flux form} ($\bh$). The Noether map is $H=\int_\Sigma \bH$, the constraint map is $\Ho \doteq \int_\Sigma \bHo$ and the flux map is $h \doteq \int_\Sigma d \bh$.
    \end{definition}

In the physics literature, the flux map is sometimes called the ``{charge aspect}''.

    \begin{definition}[Strongly Hamiltonian Field Theory]\label{def:loc-ham-ft}
    A field theory with local symmetries $(\cE,[\bL],\fG,\rho)$ satisfying Assumption \ref{ass:local-ham} is said \emph{Strongly Hamiltonian}, and the data $(\mathsf{A}_\Sigma,\bom_\Sigma,\bH_\bullet)$ is its associated Hamiltonian $\fG_{\Sigma}$-space.
    \end{definition}

    Indeed, in a strongly Hamiltonian field theory, it is immediate to verify that one has the \emph{strong Hamiltonian flow equation}
    \begin{align}\label{eq:loc-ham-flow}
    \i_{\rho_\Sigma(\xi)}\bom_\Sigma = - \d \bH_\xi = - \d \bHo{}_\xi - d \d \bh_\xi .
    \end{align} 

    \begin{remark}
        Typically if $\d\mathbf{B}_\xi+\L_\rho(\xi)\btheta$ and $\L_{\rho(\xi)}\bbE\bL$ vanish, and therefore assumption \ref{ass:local-ham}(2) holds, the rest is easy to verify.
        Examples of strongly Hamiltonian field theories are $BF$ theory, Chern-Simons theory, as well as Yang-Mills theory with $\Sigma$ either spacelike or null. The strongly Hamiltonian YM field theories for $\Sigma$ null or spacelike differ in important ways \cite{RSATMP,RSAHP}. As already mentioned, a notable counterexample failing assumption \ref{ass:local-ham}(\ref{ass:localham-1}-\ref{ass:localham-2}) is General Relativity.
    \end{remark}

With the following theorem we investigate the equivariance properties of $\bH$, $\bHo$ and $d\bh$. A map from $\boldsymbol{\alpha}_\bullet:\fG \to \oloc^{\bullet,\bullet}$ is said equivariant if $\L_{\rho(\xi)}\boldsymbol{\alpha}_\eta = \boldsymbol{\alpha}_{[\xi,\eta]}$, and \emph{weakly} equivariant if it is equivariant up to a $d$-exact cocycle.

    \begin{theorem}\label{thm:equivariance}
        In a strongly Hamiltonian field theory, if $\cC_\Sigma\into \cE_\Sigma$ is smooth and nonempty, the $\Sigma$-Noether form $\bH$ satisfies the following equivariance equation,
        \[
        \L_{\rho(\xi)}\bH{}_\eta = \bH{}_{[\xi,\eta]} - j_\Sigma([\xi,\eta]) + d{{\kappa}}(\xi,\eta).
        \]
        where 
        \[
        d{{\kappa}} : {\bigwedge^2} \fG_\Sigma \to d\Omega^{\top-1}(\Sigma)
        \]
        is a $d\Omega^{\top-1}(\Sigma)$-valued Chevalley-Eilenberg 2-cocycle.

        Moreover, among the following statements: 
        \begin{enumerate}
            \item $\bC$ is equivariant;
            \item $\bHo$ is equivariant;
            \item  $[\fG_\Sigma,\fG_\Sigma]\subseteq \ker(j_\Sigma)$, i.e.\ ${j_\Sigma}([\xi,\eta])=0$ for all $\xi,\eta\in\fG_\Sigma$;
            \item $\bH$ is weakly equivariant with
            \[
            (\L_{\rho_\Sigma(\xi)}d\bh_\eta - d\bh_[\xi,\eta])|_\cC = d{\kappa}(\xi,\eta),
            \]
        \end{enumerate}
        we have the following implications $(1) \implies (2) \implies (3) \implies (4)$.
    \end{theorem}
    (Note that if the Abelianised algebra $\fG_\Sigma/ [\fG_\Sigma,\fG_\Sigma]$ is trivial, then ${j_\Sigma}([\fG_\Sigma,\fG_\Sigma])=0$ implies ${j_\Sigma}=0$. Cf. Example \ref{ex:jext-YM}.)
    
    \begin{proof} 
        From the local Hamiltonian flow equation \ref{eq:loc-ham-flow}, one readily deduces that $\L_{\rho_\Sigma(\xi)}\bom_\Sigma =0$, whence
        \(
        -\d \bH_{[\xi,\eta]} = \L_{\rho_\Sigma(\xi)} \i_{\rho_\Sigma(\eta)}\bom_\Sigma = - \d \L_{\rho_\Sigma(\xi)}\bH_{\eta}
        \).
        Thus, 
        \[
        \L_{\rho_\Sigma(\xi)}\bH_{\eta} = \bH_{[\xi,\eta]} + \wt{{\kappa}}(\xi,\eta)
        \]
        with $\wt{{\kappa}}$ constant ($\d$-closed). It is standard to verify that $\wt{{\kappa}}$ is also a CE 2-cocycle. 
        
        From the above, it follows that
        \[
        \L_{\rho_\Sigma(\xi)}\bHo{}_{\eta} - \bHo{}_{[\xi,\eta]} - \wt{{\kappa}}(\xi,\eta) = -d(\L_{\rho_\Sigma(\xi)}\bh{}_{\eta} - \bh{}_{[\xi,\eta]})
        \]
        and since $\bHo |_\cC = {j_\Sigma}$ (Lemma \ref{lem:Jdescends}, cf.\ diagram \eqref{eq:diagram-C}),
        \[
         {j_\Sigma}{}_{[\xi,\eta]} +\wt{{\kappa}}(\xi,\eta) = d(\L_{\rho_\Sigma(\xi)}\bh{}_{\eta} - \bh{}_{[\xi,\eta]})|_\cC.
        \]
        Using that $j_\Sigma$ and $\wt{{\kappa}}$ are constant on $\cE_\Sigma$, we deduce that there exists a CE 2-cocycle ${{\kappa}}:\bigwedge^2 \fG_\Sigma \to \Omega^{\top-1}(\Sigma)$ such that
        \[
        \wt{{\kappa}}(\xi,\eta) = - {j_\Sigma}{}_{[\xi,\eta]} + d {\kappa}(\xi,\eta).
        \]
        This concludes the proof of the first part of the theorem.
        
        The implication $(1)\implies(2)$ follows from Lemma \ref{lem:Jdescends}, diagram \eqref{eq:diagram-C} and assumption \ref{ass:local-ham}. 
        The implication $(2)\implies (3)$ follows from $\bHo|_\cC = j_\Sigma$ and the constancy of $j_\Sigma$ along the same line as the proof of Theorem \ref{thm:jext=0} which showed that if $\bC$ is equivariant then $[\fG,\fG]\subset \ker(j)$.
                Finally, $(3)\implies(4)$ is proved by combining the two expressions for $\tilde\kappa$ given above.
    \end{proof}

    \begin{remark}
        From the proof of Theorem \ref{thm:equivariance} one can also deduce that if $\bHo$ is equivariant, then $\L_{\rho_\Sigma(\xi)}d\bh_\eta - d\bh_[\xi,\eta] = d{\kappa}(\xi,\eta)$ even off of the constraint surface.
    \end{remark}

    \begin{remark}[Erratum]
        In Proposition 4.10 of \cite{RSATMP} we stated that a weakly equivariant $\bH$ implies a equivariant $\bHo$. However, the argument provided works only if the bracket $[\cdot,\cdot]$ on $\fG_\Sigma$ is ultralocal\footnote{Meaning that it does not depend on the $k\geq1$ jets of its arguments.}, whereby from $\check\bH_\circ = \bHo{}_{\d\eta}$ being source we can deduce that $\bHo{}_{[\xi,\d\eta]}$ is also source. A well-known counterexample to this is the diffeomorphism (a.k.a.\ vectorial, a.k.a.\ momentum) constraint of General Relativity\footnote{In a local chart, this is $\bHo{}_\xi = D_i\Pi^{ij} h_{jk} \xi^k$, for $h$ a Riemannian metric and $\Pi$ a symmetric two-tensor density, where we identify densities and top forms on $\Sigma$.}, which is equivariant only up to total divergences, i.e., $d$-exact terms. Simpler counterexamples can be built by considering diffeomorphisms of $\Sigma$ as local symmetries in e.g.\ Chern-Simons or $BF$ theory. This error does not affect the reminder of \cite{RSATMP} if one simply \emph{assumes} that $\bHo$ is equivariant. (Note, in \cite[Def 3.17]{RSATMP} we used the term ``strongly equivariant'' as a synonym of ``equivariant".)
    \end{remark}

\begin{assumption}\label{ass:Ho-equi}
    Henceforth we will assume that $\bC$ is equivariant.
\end{assumption}

\subsection{Local Hamiltonian field theory - the case ``with corners''}\label{sec:hamwithcorners}

Recall diagram \eqref{eq:diagram-C} which states that in a theory with local symmetries physically viable configuration must lie in the constraint space $\cC_\Sigma \doteq \pi_\Sigma(\cE_j)$ with $\cE_j = \{\varphi\in\cE \ | \ \mathbf{C}|_\varphi = j\}$. From the above construction of the $\Sigma$-Noether form $\bH$ and constraint form $\bHo$, one readily sees that the constraint set is given by
\[
\cC_\Sigma = \{ \phi \in \cE_\Sigma \ |  \ \bHo|_\phi = {j_\Sigma}\} \quad \text{where}\quad  {j_\Sigma} \circ \pi_\Sigma^\fG\doteq \iota_\Sigma^* \circ j.
\]
The \emph{charactistic distribution} of $\cC_\Sigma$ in $(\cE_\Sigma,\omega_\Sigma)$ is the subset $(\cC_\Sigma)^{\omega_\Sigma} \subset T_{\cC}\cE_\Sigma$ given by the pairs $(\bar\phi,\mathbb{X})$ such that $\omega_{\Sigma}|_{\bar\phi}(\mathbb{X},\cdot)=0$.

In a strongly Hamiltonian field theory (Definition \ref{def:loc-ham-ft}) one can then prove that the characteristic distribution of the constraint set is given by the symmetry subalgebra generated by $\Ho$ or, more precisely\footnote{The previous statement is correct only if $\fGo = \fGo^\mathrm{off} \doteq \ker(h)$. This is the case in non-Abelian gauge theories or gauge theories with matter, but not in pure Maxwell theory where $\mathfrak{u}(1) \simeq \fGo/ \fGo^\mathrm{off}$.}, that:

\begin{theorem}[Constraint reduction \protect{\cite[Theorem 1(ii)]{RSATMP}}]\label{thm:RSGo}
Consider a strongly Hamiltonian field theory $(\mathsf{A}_\Sigma,\bom_\Sigma,\bH_\bullet)$, with $\bHo$ equivariant (Assumption \ref{ass:Ho-equi}).
Then, 
    \[
    \fGo \doteq \ker(\iota_\cC^*\circ h) \subset \fG_\Sigma,
    \]
with $h : \fG_\Sigma \to \Omega^0(\cE_\Sigma)$ the flux map, is the (maximal, just, constraining\footnote{This terminology is briefly explained in the proof.}) ideal of $\fG_\Sigma$. Moreover, if $\rho_\Sigma(\fGo)$ is symplectically closed,
    \[
    (\cC_\Sigma)^{\omega_\Sigma} \simeq \rho_\Sigma(\cC_\Sigma \times \fGo)  \subset T_{\cC}\cE_\Sigma
    \]
and therefore, when smooth, one has the symplectic space $(\underline{\cC}{}_\Sigma,\underline{\omega}{}_\Sigma)$, where
\[
    \underline{\cC}{}_\Sigma \doteq \cC_\Sigma/\fGo
    \quad\text{and}\quad
    \pi_\circ^*\underline{\omega}{}_\Sigma \doteq \iota_\cC^*\omega_\Sigma,
\]
for $\pi_\circ : \cC_\Sigma\to \underline{\cC}{}_\Sigma$ the corresponding surjective submersion.

We call $\fGo$ the \emph{constraint gauge group} and $(\underline{\cC}{}_\Sigma,\underline{\omega}{}_\Sigma)$ the \emph{constraint-reduced phase space} of the strongly Hamiltonian field theory $(\mathsf{A}_\Sigma,\bom_\Sigma,\bH_\bullet)$.
\end{theorem}
\begin{proof}
    The proof of the theorem for $j=0$ can be found in \cite{RSATMP}.
    Its key step, presented in Section 5.3 ibidem, is showing that $\fGo$ is the \emph{maximal}, just, constraining ideal of $\fG_\Sigma$, where a subalgebra $\iota_\mathfrak{H}:\mathfrak{H}\into \fG_\Sigma$ is said to be constraining and just if its momentum map $J_\mathfrak{H} \doteq H \circ \iota_\mathfrak{H}$ is such that $\cC_\Sigma = J_\mathfrak{H}^{-1}(0)$.\footnote{E.g.\
    the subspace of elements of $\fG_\Sigma=\Gamma(\Sigma,\Xi_\Sigma)$ which are  compactly supported in the interior of $\Sigma$ is a just, constraining, ideal which is not maximal, and indeed it is contained in $\fGo$.} The relationship between the characteristic distribution $\cC_\Sigma^{\omega_\Sigma}$ and $\rho_\Sigma(\fGo)$ follows from the hypothesis of symplectic closure \cite[Prop. 4.1.7 and Lem. 4.2.1]{DiezPHD}, which can be verified on a case by case basis (see \cite{DiezHuebschmann-YMred,RSATMP} for discussions on Yang-Mills theory).
    
    If $j\neq 0$ the argument can be adapted simply by replacing $\bH$, $\bHo$, and $\bh$, with the shifted $\Sigma$-Noether current $\bH^j \doteq \bHo^j + d\bh$ where $\bHo^j \doteq \bHo - j_\Sigma$. Indeed, although constant shifts of a covariant Hamiltonian  $\bHo \mapsto  \bHo - j_\Sigma$ invalidate its equivariance by a CE coboundary $j_\Sigma([\cdot,\cdot])$, in this case the CE coboundary $j_\Sigma([\cdot,\cdot])$ vanishes by Theorem \ref{thm:equivariance}(3). Therefore, the shift can be performed without any further consequences.
\end{proof}

\begin{remark}[Symmetry stabilizers]
   An important obstruction to the smoothness of a quotient such as $\underline{\cC}{}_\Sigma=\cC_\Sigma/\fGo$ is due to the presence of stabilisers for the action. In theories with local symmetries stabilisers are usually finite dimensional subalgebras of infinite dimensional symmetry algebras. This is the case in gauge theories over a principal bundle (as well as general relativity, although the letter is not strongly Hamiltonian). If $\pp\Sigma\neq\emptyset$, the extra boundary conditions involved in the definition of $\fGo = \ker(\iota_\cC\circ h) \subsetneq \fG$ often implies that $\fGo$ acts freely on $\cC_\Sigma$, i.e.\ without fixed points. 
\end{remark}

\medskip
Since the constraint gauge subalgebra $\fGo\subset \fG_\Sigma$ is an ideal, one can construct the quotient algebra
\begin{definition}
The \emph{flux gauge algebra} is the quotient of the gauge algebra $\fG$ by the constraint gauge ideal
\[
\underline{\fG}{}_\Sigma \doteq \fG_\Sigma/\fGo,
\]
with its induced Lie algebra structure.
\end{definition}

The flux gauge algebra is guaranteed to act in a Hamiltonian fashion on the constraint-reduced phase space $(\underline{\cC}{}_\Sigma,\underline{\omega}{}_\Sigma)$, which can therefore be itself reduced by its action.
Therefore the Hamiltonian reduction $\fG_\Sigma\circlearrowright \cC_\Sigma$ can be split into two steps, making this an example of \emph{Hamiltonian reduction by stages} \cite{Marsdenstages,OrtegaRatiu}.

The first-stage is Hamiltonian reduction with respect to the ideal $\fGo\circlearrowright\cC_\Sigma$ and---under the assumption of symplectic closure---it corresponds to the coisotropic reduction of the constraint set $\cC_\Sigma$, thus yielding the symplectic manifold $(\underline{\cC}{}_\Sigma,\underline{\omega}{}_\Sigma)$. 

The second stage corresponds to the residual reduction $\underline{\fG}{}_\Sigma \circlearrowright \underline{\cC}{}_\Sigma$, which instead produces a \emph{collection} of symplectic spaces, labeled by (the coadjoint orbit of) a value $f\in (\underline{\fG}{}_\Sigma)^*$ of the second-stage momentum map. We call these symplectic leaves the \emph{flux superselection sectors} $(\underline{\underline{\mathcal{S}}}{}_{\mathcal{O}_f}, \underline{\underline{\omega}}{}_{\mathcal{O}_f})$ of the reduced theory. Taken together, these sectors form a Poisson manifold\footnote{\label{fnt:Poisson}See \cite{PelletierCabau,RSATMP} for the precise notion of Poisson manifold in infinite dimensions used here.} $\underline{\underline{\cC}}{}_\Sigma$---the result of the Poisson reduction of $(\underline{\cC}{}_\Sigma,\underline{\omega}{}_\Sigma)$ by the Hamiltonian, and therefore Poisson, action of $\underline{\fG}_\Sigma$.

In the remainder of this section, we emphasise the physical interpretation of each of the two stages, and provide a constructive definition of the flux superselection sector $\underline{\underline{\mathcal{S}}}{}_{\mathcal{O}_f}$.

\medskip

As mentioned, the first stage implements the coisotropic reduction of the constraint set $\cC_\Sigma \hookrightarrow \cE_\Sigma$. Recall that $\cC_\Sigma = \{\phi \ | \ \bHo(\phi)=j_\Sigma\}$ is the set of fields on $\Sigma$ that satisfy conditions which are \emph{necessary} for them to serve as the initial data for a solution to the equations of motion in $D(\Sigma)$, when not empty.\footnote{Indeed, the constraints typically exhaust all the \emph{local} conditions, but extra ``global", e.g.\ topological, conditions on the initial data might emerge.} 

Therefore, within a gauge theory, $\underline{\cC}{}_\Sigma$ is fully and uniquely determined by the choice of surface $\Sigma$.
Ignoring a couple of technical aspects for the sake of expository clarity, one could restate this point by saying that, in the first-stage reduction, the value of the momentum map\footnote{Indeed, it is not always correct to state that $\bHo$ is the momentum map for $\fGo \subset \fG$, although it induces one. See \cite[Section 5.4]{RSATMP}.} $\bHo$ for the constraint ideal $\fGo$ is fixed once and for all to be given by the external currents, i.e.\ $\bHo=j_\Sigma$. 

If $\pp\Sigma = \emptyset$, since $\bH$ and $\bHo$ differ only by a $d$-exact term, one has $\fGo = \fG_\Sigma$ and the reduction procedure is constituted by a single stage.

If $\pp\Sigma\neq\emptyset$, on the other hand, second-stage reduction is nontrivial. It is controlled by the the flux map $h = \int_\Sigma d\bh$. The name of this map is inspired by Maxwell theory, where $\langle h,\xi\rangle$ equals the value of the electric flux through $\pp\Sigma$ integrated against the gauge parameter $\xi$ \cite{RielloSciPost,RSATMP,RSAHP,fewster2025semilocalobservablesedgemodes}. 

Note that in the second-stage Hamiltonian reduction, the value of the ``flux'' momentum map $h$ at which reduction is performed is not fixed a priori, as opposed to what happens in the first stage, where $\Ho$ is forced to take on a distinguished value, owing to the \emph{physical} requirement that reduced configurations be initial data for solutions of the equations of motion (the ``constraint"). This makes the second stage, or flux reduction, more akin to the reduction of a global symmetry (cf. Section \ref{sec:HamiltonianMechRed} and Section \ref{sec:remarks}).

To proceed further we need first to introduce some notation. Note that the dual of the flux gauge algebra $\underline{\fG}{}^*_\Sigma$ can be understood as the subset of $\fG^*_\Sigma$ given by elements that annihilate $\fGo\subset \fG_\Sigma$. Thinking of the flux map $h:\fG_\Sigma \to \Omega^0(\cE_\Sigma)$ as the $\fG^*_\Sigma$-valued function\footnote{This is akin to the duality between comomentum ($h$) and momentum ($\hat h$) maps.} 
\begin{equation}\label{e:momentum/comomentum}
\hat{h} : \cE_\Sigma \to \fG_\Sigma^*,  
\end{equation}
we denote $f  \in\mathfrak{F} \doteq \mathrm{Im}(\iota_\cC^* \hat h)$. We call $f$ an (on-shell) flux and $\mathfrak{F}$ the space of (on-shell) fluxes.

Using the equivariance of $\hat h$ up to the CE 2-cocycle $k \doteq \int_\Sigma d{\kappa}$ (Theorem \ref{thm:equivariance} and Assumption \ref{ass:Ho-equi}), we introduce $\mathcal{O}_f$ to be the orbit of $f$ under the affine action $\ad_k^*$ defined as the $k$-twist of the coadjoint action of $\fG_\Sigma$ on $\fG^*_\Sigma$:
\[
\ad_{{k}}^* : {\fG}{}_\Sigma \times {\fG}{}_\Sigma^*\to {\fG}{}_\Sigma^*,
\quad
({\xi}, \alpha) \mapsto \ad_k^*({\xi}) \alpha \doteq \ad^*({\xi})\alpha+ {k}({\xi},\bullet).
\]
The affine orbit $\mathcal{O}_f$ carries a canonical symplectic structure $\Omega_{\mathcal{O}_f}$, customarily named after Kirillov, Konstant, and Souriau \cite{OrtegaRatiu}.

With this notation, we define the set of constrained configurations in $\cC_\Sigma$ whose fluxes belong to $\mathcal{O}_f$:
\[
\mathcal{S}_{\mathcal{O}_f} \doteq \{ \phi \in \cC_\Sigma \ | \ H(\phi) \equiv h(\phi) \in \mathcal{O}_f\}.
\]
We also denote its first- and second-stage quotients
\[
\underline{\mathcal{S}}{}_{\mathcal{O}_f} \doteq \mathcal{S}_{\mathcal{O}_f} / \fGo
\qquad\text{and}\qquad
\underline{\underline{\mathcal{S}}}{}_{\mathcal{O}_f} \doteq  \mathcal{S}_{\mathcal{O}_f} / \fG_\Sigma \simeq \underline{\mathcal{S}}{}_{\mathcal{O}_f} / \underline{\fG}{}_\Sigma.
\]

Finally, we recall that Assumption \ref{ass:local-ham} states that the symmetry action $\rho_\Sigma$ descends to the constraint-reduced phase space $(\underline{\cC}{}_\Sigma,\underline{\omega}{}_\Sigma)$ as per the following commutative diagram
\[
    \xymatrix{
    \fG \times \cC_\Sigma \ar[d]_-{\cdot/\fGo} \ar[r]^-{\rho_\Sigma} & T\cC_\Sigma \ar[d]^-{\cdot/\fGo} \\
    \underline{\fG}{}_\Sigma  \times 
    \underline{\cC}{}_\Sigma \ar[r]_-{\underline{\rho}{}_\Sigma} &  T\underline{\cC}{}_\Sigma
    }
\]

After these preliminaries, we can finally state the flux reduction theorem:

\begin{theorem}[Flux reduction \protect{\cite[Theorem 1 (iv-v)]{RSATMP}}]\label{thm:FluxRed}
In addition to the hypotheses of Theorem \ref{thm:RSGo}, assume that $\underline{\rho}{}_\Sigma(\underline{\fG}{}_\Sigma \times \underline{\mathcal{S}}{}_{\mathcal{O}_f}) \subset T\underline{\mathcal{S}}{}_{\mathcal{O}_f}$ is symplectically closed in $(\underline{\cC}{}_\Sigma,\underline{\omega}{}_\Sigma)$. Then, when smooth, for each $f\in\mathfrak{F}$ one has the symplectic space $(\underline{\underline{\mathcal{S}}}{}_{\mathcal{O}_f}, \underline{\underline{\omega}}{}_{\mathcal{O}_f})$ defined by
\[
\underline{\underline{\mathcal{S}}}{}_{\mathcal{O}_f} \doteq  \mathcal{S}_{\mathcal{O}_f} / \fG_\Sigma
\quad\text{
and
}\quad
\pi_{\mathcal{O}_f}^*\underline{\underline{\omega}}{}_{\mathcal{O}_f} \doteq \iota_{\mathcal{O}_f}^*\omega - (\hat h \circ \iota_{\mathcal{O}_f})^* \Omega_{\mathcal{O}_f},
\]
where  $\pi_{\mathcal{O}_f} : \mathcal{S}_{\mathcal{O}_f} \to \underline{\underline{\mathcal{S}}}{}_{\mathcal{O}_f} $ is the corresponding surjective submersion and $\hat{h}$ is defined in Equation \ref{e:momentum/comomentum}. We call $(\underline{\underline{\mathcal{S}}}{}_{\mathcal{O}_f}, \underline{\underline{\omega}}{}_{\mathcal{O}_f})$ the \emph{$\mathcal{O}_f$-flux superselection sector} over $\Sigma$.

Moreover,
\[
\underline{\underline{\cC}}{}_\Sigma \doteq \cC_\Sigma /\fG \simeq \bigsqcup_{\mathcal{O}_f \subset \mathfrak{F}} \underline{\underline{\mathcal{S}}}{}_{\mathcal{O}_f} .
\]
is a (partial \cite{PelletierCabau}) Poisson manifold, of which the superselection sectors $(\underline{\underline{\mathcal{S}}}{}_{\mathcal{O}_f}, \underline{\underline{\omega}}{}_{\mathcal{O}_f})$ are the symplectic leaves. We call $\underline{\underline{\cC}}{}_\Sigma$ the \emph{fully reduced phase space} over $\Sigma$.
\end{theorem}

To conclude, we summarise the two-stage reduction procedure by the following diagram, with the flux reduction theorem \ref{thm:FluxRed} following the dashed arrows:

\medskip

\begin{equation*}
\xymatrix@C=.75cm@R=1cm{
(\cE_\Sigma,\omega_\Sigma)
    \ar@{~>}[rr]^-{\tbox{2.2cm}{constraint reduction \\(w.r.t. $\fGo$ at $j_\Sigma$)}}
&&(\underline{\cC}{}_\Sigma,\uomegao{}_\Sigma)
    \ar@{~>}[rr]^-{\tbox{2.4cm}{flux superselection (w.r.t. $\fGred{}_\Sigma$ at $\mathcal{O}_f$)}}
&&(\uuS_{\mathcal{O}_f},\uuomegao_{\mathcal{O}_f})\\
&{
    \;\qquad\cC_\Sigma\qquad\;
    \ar@{_(->}[ul]^-{\iota_\cC}
    \ar@{->>}[ur]_-{\pi_\circ}
    }
&&{
    \;\;\uS{}_{\mathcal{O}_f} 
    \ar@{_(->}[ul]^-{\underline{\iota}_{\mathcal{O}_f}}
    \ar@{->>}[ur]_-{\underline{\pi}_{\mathcal{O}_f}}
    }\\
&&{
    \;\;\mathcal{S}_{\mathcal{O}_f} 
    \ar@{_(->}[ul]^-{{\iota}^\cC_{\mathcal{O}_f}}
    \ar@{->>}[ur]_-{\pi_{\circ,\mathcal{O}_f}}
    \ar@{_(-->}@/^2.7pc/[uull]^{\iota_{\mathcal{O}_f}}
    \ar@{-->>}@/_2.7pc/[uurr]_{\pi_{\mathcal{O}_f}}
    }\\
}
\end{equation*}

\medskip

\subsection{The corner algebra and the space of superselection sectors}\label{sec:corneralgebra}

We have seen that the densitised momentum map $\bH$ of a strongly Hamiltonian field theory splits into the sum of a constraint form $\bHo$ and a flux form $d\bh$. Integration of the flux form yields the flux map $h = \int_\Sigma d\bh$ which can be alternatively be thought of as a map $\hat h : \cE_\Sigma \to \fG_\Sigma^*$.

Dropping the subscript $\bullet_\Sigma$,
define the following 1-form on the action algebroid $\mathsf{A} = \cE \times \fG$
\[
\alpha(\phi,\xi) \doteq \langle \hat h(\phi) , \d \xi\rangle \in \Omega^{1}(\mathsf{A}).
\]
Mimicking the Kijowsky--Tulczyjew (KT) procedure, one can think of $\d\alpha$ as a presymplectic form to construct a surjective submersion $\pi_\pp: \mathsf{A}\to \mathsf{A}_\pp $ on a space of \emph{corner fields} $\mathsf{A}_\pp$ supported on $\pp\Sigma$.

Indeed, assuming that the action of $\fG$ on $\cE$ is faithful\footnote{In fact, it is sufficient to assume that all $\phi\in\cE$ has the same stabiliser $\mathfrak{N} \subset \fG$, i.e.\ $\rho(\xi)\phi =0$ iff $\xi\in\mathfrak{N}$. Cf. Theorem 4 and Assumption G of \cite{RSATMP}.}, one can prove that $\mathsf{A}_\pp \simeq \cE_\pp \times \fG_\pp$ is itself an action Lie algebroid on a space of boundary fields $\cE_\pp$ whose anchor $\rho_\pp$ is induced by $\rho$. In other words, $\pi_\pp = \pi_{\pp,\cE}\times \pi_{\pp,\fG}$ factorises and there is a Lie algebroid morphism (see \cite[Theorem 4]{RSATMP} for details):
\[
\xymatrix{
\mathsf{A} \ar[r]^-{\rho} \ar[d]_-{\pi_\pp} & T\cE \ar[d]^-{\d \pi_{\pp,\cE}}\\
\mathsf{A}_\pp \ar[r]^-{\rho_\pp} & T\cE_\pp.
}
\]

In fact, the Lie algebroid $\mathsf{A}_\pp$ is by construction symplectic, so that its base $(\cE_\pp,\Pi_\pp)$ is a Poisson manifold\footnote{See footnote \ref{fnt:Poisson}.} such that $\Pi_\pp^\sharp = \rho_\pp$.

Moreover, one can show that $h$ descends to $h_\pp$, i.e.\ $h = \pi_\pp^* h_\pp$, and the corner function $h_\pp : \cE_\pp \to \fG_\pp^*$ can be used as a variable on $\cE_\pp$ to define a functional derivative $\frac{\delta}{\delta h_\pp}$ valued in $\fG_\pp$ (see \cite[Section 7.2]{RSATMP} for details). Then, the Poisson bivector $\Pi_\pp$ on the space of boundary fields $\cE_\pp$ reads
\begin{equation}\label{e:OScornerPoisson}
\Pi_\pp = \frac12 \langle h_\pp, \left[\frac{\delta}{\delta h_\pp}, \frac{\delta}{\delta h_\pp} \right] \rangle + \frac12 k\left( \frac{\delta}{\delta h_\pp}, \frac{\delta}{\delta h_\pp}\right).
\end{equation}
Note: in this section, the notion of boundary refers to $\Sigma$, therefore from a spacetime perspective it corresponds to a codimension 2 submanifold of $M$, a.k.a. a corner. Therefore, from a spacetime perspective, $(\cE_\pp,\Pi_\pp)$ are \emph{corner} data.

\begin{remark}
    In the physics literature, this Poisson bivector is often guessed in the pursuit of an algebraic structure defining the so-called ``corner charge algebra''. Our discussion shows how, under reasonable assumptions, starting from \ref{ass:local-ham}, one can give a rigorous meaning to this algebra in terms of the corner Poisson manifold $(\cE_\pp,\Pi_\pp)$.
\end{remark} 

\begin{remark}[Flux superselection sectors]
There is a relationship between $(\cE_\pp,\Pi_\pp)$ and the space of flux superselection sectors, namely: the symplectic leaves of $(\cE_\pp,\Pi_\pp)$ are in bijective correspondence with (the connected components of) the preimage of the affine orbits $\mathcal{O}_f$, seen as subsets of $(\fG_\pp)^*$, along $h_\pp$.
Which means that studying the symplectic leaves of $(\cE_\pp,\Pi_\pp)$ or, more loosely speaking, the Casimirs of the corner charge algebra, one can classify the (off-shell\footnote{This qualifier is needed if the constraints impose global conditions that disallow certain sectors, i.e.\ certain values of the corner fields. E.g. in vacuum Maxwell theory, due to the Gauss constraint, only electric fields whose flux across $\pp\Sigma$ vanishes are allowed.}) flux superselection sectors. 
\end{remark}

We conclude this section with the following observation.\footnote{See \cite[Section 7]{RSATMP} for a detailed elaboration.} Introducing an odd variable $c$ for the graded vector space $\fG_\pp[1]$, the bivector $\Pi_\pp$ can be rewritten in terms of the degree-2 function 
\begin{equation}\label{e:earlyBFFV}
S_\pp \doteq \frac12 \langle h_\pp, [c,c] \rangle + \frac12 k\left( c,c\right) \in C^\infty(\mathsf{A}_\pp[1])^{(2)}.
\end{equation}
Similarly, the symplectic potential on $\mathsf{A}_\pp$ is encoded in the degree $1$ function
\[
\alpha_\pp \doteq \langle h_\pp, c\rangle \in C^\infty(\mathsf{A}_\pp[1])^{(1)}.
\]
We have put forward this observation because the functions $S_\pp$ and $\alpha_\pp$ will naturally appear in the approach to (gauge) field theories with corners presented in the following section.

\section{Resolution of phase space reduction}\label{sec:BV}

In the previous section we have discussed how the phase space in field theory, in the presence of symmetry, naturally emerges as the quotient space of a space of solutions of PDE's by the equivalence relations induced by the action of local symmetries. Indeed, on the one hand, the covariant phase space $\cEL(\bL)$, equipped with the closed two form $\omega^h_{\cEL|\Sigma}$ is reduced by the associated characteristic distribution which necessarily includes (bulk) gauge transformations
(Definition \ref{def:cps});
while, on the other hand, the constraint-reduced phase space $\underline{\cC}{}_\Sigma$ arises as the reduction of $(\cC_\Sigma,\iota^*_\cC\omega^h_\Sigma)$ by the (Hamiltonian) action of $\fGo\subset \fG_\Sigma$ (as defined in Theorem \ref{thm:RSGo}).

In (a large number of) good cases, both $\cEL(\bL)$ and $\cC_\Sigma$ are smooth pre-symplectic manifolds defined through a choice of phase submanifold $\Sigma \hookrightarrow M$ (Definition \ref{def:achronal}), and a choice of homotopy (i.e.\ a consistent choice of primitives) for $d$-exact forms such as $d\btheta\doteq \d \bL - \bbE(\bL)$.

However, in general, the corresponding \emph{reduced} phase spaces may have a singular structure given by the lack of regularity of the kernel of $\omega^h_{\cEL|\Sigma}$ or $\iota^*_{\cC}\omega_\Sigma$, which renders their description in terms of smooth (Fr\'echet) manifolds more involved. Even assuming this is possible (typically as smooth, but stratified spaces \cite{ArmsMarsdenMoncrief1981,DiezRudolph_slice,DiezHuebschmann-YMred}) their description is complicated by the emergence of non-localities, which render their viability for quantisation less optimal.

We recall, in fact, that quantisation for a field theory is generally understood as quantisation of its reduced phase space, in terms of a suitable representation of (a subalgebra of) its algebra of functions over a (possibly Hilbert) space of quantum states.\footnote{Here we are being deliberately vague to contain various attempts, such as geometric and deformation quantisations, which achieve different goals.}

A possible route to handle these complications requires forgoing a direct description of said algebra, which is then replaced by the cohomology in degree zero of a suitably constructed cochain complex. Often, one further aims for a cohomological description in terms of a \emph{resolution} of the space of functions over the reduction, which is an additional requirement on the vanishing of the negative cohomology groups. We will use the term ``resolution'' more loosely here, simply to refer to a complex whose cohomology in degree zero reproduces the desired space.

The techniques we will present are the culmination of decades-long investigations pioneered by Batalin, Fradkin and Vilkovisky (in various combinations, hence the acronyms BV and BFV, see below), which built on previous work of DeWitt, Faddeev and Popov, as well as Becchi, Rouet, Stora and Tyutin (BRST). Apart from the main idea of describing spaces via cohomologies, the BV/BFV/BRST mechanisms have another crucial feature: they retain locality. Indeed, this formalism takes as input a Lagrangian field theory (with symmetries) $(\cE,\bL,\fG)$ and outputs an extended Lagrangian field theory on a larger space of fields $\cE_{\BV}$, extended by ``fields'' of non-zero internal degree, endowed with an extension $\bL_{\BV}$ of the original Lagrangian density. Similarly, one can construct an extension of the strongly Hamiltonian field theory $(\cE_\Sigma, \bH, \fG_\Sigma)$ by means of the BFV data $(\cE_{\BFV}, \bH_{\BFV})$. 

In this section we will describe both the BV (Lagrangian) and the BFV (Hamiltonian) formalism, as well as how their mutual relation within the BV-BFV formalism.
In so doing, we will highlight both how the BV-BFV formalism greatly generalises the discussion of reduction outlined in Section \ref{sec:PS} and---conversely---how the geometric discussion in Section \ref{sec:PS} can be used to illuminate the geometric and physical meaning of certain algebraic structures appearing in the BV/BFV/BRST framework.

Our discussion will focus on spaces of fields which are vector (and affine) bundles, although most of it can be extended---with appropriate \emph{caveats}---to general fibre bundles. In particular, we work with graded vector spaces that are spaces of sections of $\mathbb{Z}$-graded vector bundles $V^\bullet=\{V^k\}_{k\in\mathbb{Z}}\to N$. If $\mathcal{V}$ denotes the space of sections of $V^\bullet$, we can construct the bicomplex of local forms over $\mathcal{V}\times N$ in the usual way. Tensor fields over $\mathcal{V}$ inherit a grading from the internal $\mathbb{Z}$-degree of the graded vector bundle. This is usually called \emph{ghost degree} and it is denoted by $\mathsf{gh}$.

Hence, on $\oloc^{\bullet,\bullet}(\mathcal{V}\times N)$ one has three relevant gradings: the vertical form degree, the horizontal form degree and the internal (ghost) degree. We refer to \cite{QiuZabzine} for a primer on graded geometry, relevant to this section. 

\begin{definition}\label{def:HamVF}
    Let $\mathcal{V}=\Gamma(N,V^\bullet)$ be a graded vector space as above. A $k$-\emph{symplectic} local form is an element $\bom\in \oloc^{2,\top}(\mathcal{V}\times N)$ such that $\mathsf{gh}(\bom)=k$ and $\bbI\circ\bom^\flat\colon T\mathcal{V}\to \Omega_\src^{1,\top}(\mathcal{V}\times N)$ is injective.\footnote{$\bom^\flat_\src$ injective implies $\bom^\flat$ injective, while the converse implication does not hold. See \cite{RSATMP}.}

    A $(0,\top)$ local form $\mathbf{F}$ is said to be Hamiltonian w.r.t.\ a $k$-symplectic local form iff there exists a local vector field $X_{\mathbf{F}}\in \fX_{\loc}(\mathcal{V})$ such that
    \[
    \bbI \i_{X_{\mathbf{F}}}\bom = \bbI\d\mathbf{F}. 
    \]
    We denote the space of Hamiltonian forms as $\Omega_{\mathrm{ham}}^{0,\top}(\mathcal{V}\times N)$, and define the antisymmetric bilinear map \cite{SchiavinaSchnitzer}
    \[
    \{\cdot,\cdot\}_{\bom}\colon \Omega_{\mathrm{ham}}^{0,\top}
    \otimes \Omega_{\mathrm{ham}}^{0,\top}
    \to \Omega_{\mathrm{ham}}^{0,\top}, \quad \{\mathbf{F},\mathbf{G}\}_{\bom} \doteq \i_{X_{\mathbf{F}}}\i_{X_{\mathbf{G}}}\bom.
    \]
\end{definition}

\begin{remark}
    It is easy to check that Definition \ref{def:HamVF} of $\{\cdot,\cdot\}_{\bom}$ is well posed as a (skew) antisymmetric linear bi-derivation map, however it fails to be a Lie bracket due to a $d$-exact Jacobiator. Indeed, one can show that it is the 2 bracket of an $L_\infty$ algebra structure. See \cite{SchiavinaSchnitzer} for more details on this construction.
\end{remark}

\subsection{BV formalism}
The Batalin--Vilkovisky (BV) formalism is a procedure by which one builds a symplectic differential graded (dg) manifold associated to the critical locus of (the differential of) a function $S$ on a manifold $M$ possibly equipped with a distribution $D\subset TM$, representing ``symmetries'', such that $L_X S=0$ for every section $X\in\Gamma(D)$.

In order to do this, one considers the contraction $i_{d S}$ on the space of multivector fields $\fX^\bullet(M)$ and looks at the Koszul complex $\left(\fX^\bullet(M)\simeq C^\infty(T^*[-1]M), i_{d S}\right)$, which is defined so that the cohomology in degree zero is given by $C^\infty(M)/I_{dS}\simeq C^\infty(\mathrm{Crit}(dS))$, where $\mathrm{Crit}(dS)$ is the critical locus of $S$,  and $I_{dS}$ its vanishing ideal.\footnote{Let $\iota_N : N \hookrightarrow M$ be an immersed submanifold; its vanishing ideal is $I_N \doteq \{ f\in C^\infty(M) \ | \ \iota_N^* f =0 \}$.}
This can be thought of as a resolution of the critical locus of $dS$.\footnote{Another way of thinking about this is as building the derived critical locus, in the language of algebraic geometry.}

Now, if the distribution $D$ is given by a Lie algebra action on the whole of $M$ it is easy to see that the Chevalley--Eilenberg differential $d_{CE}$ acts on the resolution $\fX^\bullet(M)$. Since these are symmetries, the CE differential commutes with the Koszul differential, so that $d_K + d_{CE}$ is a new differential on $\fX^\bullet(M)\otimes \bigwedge^\bullet \fg^*$. This construction can be generalised to more general distributions $D$, leading to a more complicated perturbation of $d_K$ which was constructed recursively by Batalin, Fradkin, and Vilkovisky \cite{BV1,BV2,BV3}. The result is the differential $Q_{\BV}\doteq d_K + d_\infty$ such that $Q_{\BV}^2=0$. See e.g.\ \cite{Fisch,Stasheff_secret,RejznerBook}.

The output of this procedure is the data $(M_{\BV}, \omega_{\BV}, S_{\BV}, Q_{\BV})$ where $(M_{\BV}, \omega_{\BV})$ is a $(-1)$ symplectic graded manifold extending $M$ in non-zero degree, and $S_{\BV}$ is a higher-degree extension of $S$ with the property of being the Hamiltonian function of the BV differential $Q_{\BV}$ reinterpreted as a degree 1 cohomological vector field on $M_{\BV}$. This way, the BV complex $\mathfrak{BV}^\bullet =(C^\infty(M_{\BV}), Q_{\BV})$ gives a resolution of the critical locus of $S$ modulo symmetries. Note that the cohomological vector field $Q_{\BV}$ and its Hamiltonian function $S_{\BV}$ are not independent data, and can in fact be reconstructed from each other.

When starting from a Lagrangian field theory $(\cE,\bL)$, the construction of the BV resolution for $S=\int \bL$ is local in nature and can thus be phrased in terms of local forms. 
\begin{definition}[BV data]\label{def:BVdata}
    A BV extension of a local Lagrangian field theory $(\cE,\bL)$ is a set of data $(\cE_{\BV},\bom_{\BV},\bL_{\BV},Q_{\BV})$, where
    \begin{enumerate}
        \item $(\cE_{\BV},\bom_{\BV})$ is a $(-1)$ symplectic graded manifold\footnote{We denote by $T^\vee$ the densitised cotangent bundle. This is obtained, fibrewise, by taking section of the densitised linearised dual bundle: $ T^\vee_\phi\cE \simeq \cE^\vee = \Gamma(M,E^*\otimes \mathrm{Dens}(M))$, see \cite[Remark 3]{RSATMP}.}
        \[
        \cE_{\BV} \doteq T^\vee[-1]\cE'
        \]
        endowed with the canonical, (weakly)-symplectic density $\bom_{\BV}\in \oloc^{1,\top}$, where $\cE'$ is the space of sections of a graded vector bundle $E'\to M$ such that $\mathsf{Body}(\cE')\equiv (\cE')^{(0)}\simeq\cE$; 
        \item $Q_{\BV}\in \fX(\cE_{\BV})$ is a degree $1$ vector field such that $[Q_{\BV},Q_{\BV}]=0$ (cohomological);
        \item $\bL_{\BV}\in\oloc^{1,\top}$ is a degree $0$ Hamiltonian form such that $\bL_{\BV}\vert_{\mathsf{Body}(\cE')} = \bL$;
    \end{enumerate}
    which must satisfy the Hamiltonian flow equation
        \[
        \bbI \i_{Q_{\BV}} \bom_{\BV} = \bbI\d \bL_{\BV} .
        \]
    The BV complex associated to the BV extension $(\cE_{\BV},\bom_{\BV},\bL_{\BV},Q_{\BV})$ of $(\cE,\bL)$ is $(\mathfrak{BV}^\bullet \doteq C^\infty(\cE_{\BV}),Q_{\BV})$.
\end{definition}

Observe that,  while in finite dimensions $Q_{\BV}$ and $S_{\BV}$ form a Hamiltonian pair in the standard sense, in infinite dimensions we require that $Q_{\BV}$ admit a Hamiltonian form $\bL_{\BV}$ only in the ``weak'' sense of Definition \ref{def:HamVF}, i.e.\ up to a $d$-exact term. As a consequence, while in the finite-dimensional BV framework one immediately sees that the Hamiltonian function of $Q_{\BV}$ satisfies the \emph{classical master equation}, $\{S_{\BV},S_{\BV}\}=0$, in field theory we see that the correct condition is:\footnote{Note that even when defining $S_{\BV}=\int\bL_{\BV}$ one has that $\{S_{\BV},S_{\BV}\}$ vanishes only up to boundary terms.}

\begin{lemma}
    Let $(\cE_{\BV},\bom_{\BV},\bL_{\BV},Q_{\BV})$ as above, then $\bL_{\BV}$ satisfies the densitised classical master equation, i.e.
    \[
    \{\bL_{\BV},\bL_{\BV}\}_{\bom_{\BV}} \in \mathrm{Im}(d).
    \]
\end{lemma}
\begin{proof}
    From Cartan's calculus and the fact that $Q$ is of ghost degree 1 and cohomological, one computes:
    \[
    0 =  \i_{[Q,Q]}\bom_\BV 
    = [\L_Q ,\i_Q ]\bom_\BV 
    = -\d \i_Q \i_Q \bom_\BV + 2 \i_Q \d \i_Q \bom_\BV
    \]
    whence, from the definition of $\{\cdot,\cdot\}_{\bom}$,
    \[
    \d\{ \bL_\BV,\bL_\BV\}_{\bom_\BV} 
    = \d \i_Q \i_Q \bom_\BV = 2 \i_Q \d \i_Q \bom_\BV.
    \]
    Now, using the horizontal homotopy equation $\id = d h^\geq + h^\geq d + \bbI$ twice, together with the Hamiltonian flow equation relating $\bom_\BV$ and $\bL_\BV$, as well as the fact that $\bbE = \bbI \d$ is nilpotent, one finds:
    \[
    \d \{ \bL_\BV,\bL_\BV\}_{\bom_\BV} = d ( 2 \i_Q h^\geq \d \bbI \d \bL_\BV - 2 \i_Q \d h^\geq \i_Q \bom_\BV) \doteq d \boldsymbol{R}
    \]
    (recall that $\bL_\BV$ and $\bom_\BV$ are top forms and hence $d$-closed).
    
    To conclude, we apply the vertical homotopy equation $\id = \bbh \d + \d \bbh + p^*0^*$ together with the fact that there exists no non-zero (field-space) constant of ghost degree 1:
    \[
    \{ \bL_\BV,\bL_\BV\}_{\bom_\BV} = \bbh \d \{ \bL_\BV,\bL_\BV\}_{\bom_\BV} = \bbh d \boldsymbol{R} = - d \bbh \boldsymbol{R}
    \]
    (recall that Anderson's vertical homotopy $\bbh$ commutes with $d$, and observe that since $\{ \bL_\BV,\bL_\BV\}_{\bom_\BV}$ is a vertical zero-form one has $\bbh \{ \bL_\BV,\bL_\BV\}_{\bom_\BV} =0$).
\end{proof}

    Definition \ref{def:BVdata} is not constructive and suggests non-uniqueness.\footnote{Indeed the very notion of resolution is somewhat empty if not complemented with additional requirements.} A result discussing uniqueness of the BV data under certain assumptions was presented in \cite{FelderKazhdan}. A particularly useful instance of this non-uniqueness arises when considering the so-called \emph{non-minimal sector}, used to frame gauge fixing in a particularly simple way (see e.g.\ \cite{RejznerBook}). However, as anticipated above, in the simpler cases where the Lie algebra $\fG$ acts faithfully on $\cE$, the following data provide a BV resolution:
    \[
    \cE_{\BV} = T^\vee[-1](\cE \times \fG[1])
    \quad\text{and}\quad
    \bL_{\BV} = \bL + \langle\Phi^\dag, Q_{CE} \Phi\rangle.
    \]
    Here, ${Q}_{CE}$ is the Chevalley--Eilenberg differential $d_{CE}$, now thought of as a degree $1$ cohomological vector field on $\mathsf{A}[1]\doteq\cE\times \fG[1]$, while $(\Phi,\Phi^\dag)$ are variables in the densitised cotangent bundle $T^\vee[-1](\cE\times \fG[1])$, paired through the nondegenerate pairing $\langle\cdot, \cdot \rangle$. Note, indeed, that the space of functions over $\mathsf{A}[1]$ can be thought of as 
    \[
    C^\infty(\mathsf{A}[1])\equiv C^\infty(\cE\times \fG[1])\simeq C^\infty(\cE)\otimes \wedge^\bullet \fG^*\simeq C_{CE}^\bullet(\fG,C^\infty(\cE)),
    \]
    after the tensor is appropriately completed. The CE differential is then a degree $1$ derivation of this algebra.

    Then, it is easy to gather that $Q_{\BV} = Q_K + \check{Q}_{CE}$, where $\check{Q}_{CE}$ is the cotangent lift of $Q_{CE}$, and $Q_K$ is the Koszul differential: for any pair $(\varphi,\varphi^\dag) \in T^\vee[-1]\cE^{(0)}$, respectively of degree $0$ and $-1$, we have $Q_K\varphi^\dag \doteq (\bbE\bL)_\varphi= \bE_\varphi$ and zero otherwise, so that the image of the Koszul differential is the vanishing ideal of the Euler--Lagange locus of $S$, $\Im(Q_K)\simeq I_{\cEL(\bL)}$. Indeed, one has that 
    \[
    H^0(\mathfrak{BV}) \simeq C^\infty(\cEL(\bL))^\fG,
    \]
    namely: the zero-th cohomology of the BV complex so constructed is given by the space of $\fG$-invariant functions on $\cEL$, that is to say, the space of functions on the symmetry-reduced covariant phase space $\cEL/\fG$.

    \begin{remark}
    One can make contact with the construction outlined at the beginning of the section by choosing $S_{\BV} =\int_M \bL_{\BV}$ to be the action functional built out of the Lagrangian density $\bL_{\BV}$.\footnote{Here $M$ is a closed manifold. For $M$ noncompact one has to smear the integral with a compactly supported function, see \cite{RejznerBook}.}
    One then formulates the complex that results from the resolution procedure in terms of a local differential graded manifold $(\cE_{\BV},\omega_{\BV},Q_{\BV})$, e.g.\ reinterpreting $i_{\bd S}$ as a cohomological vector field on the graded manifold $T^\vee[-1]\cE$.
    \end{remark} 
    
    Thus, a BV extension of a Lagrangian field theory provides a resolution of the symmetry-reduced covariant phase space $\cEL/\fG$, with the non-trivial advantage of allowing us to work with smooth, local, field theory data---albeit enlarged via a graded extension.

It is natural to ask whether one can provide a similar description of the constraint-reduced phase space $\underline{\cC}_\Sigma \doteq \cC_\Sigma/\fGo$. The answer is affirmative, and the analogous procedure is due to Batalin, Fradkin and Vilkovisky (BFV).

\subsection{BFV formalism}\label{sec:BFV}

Consider a coisotropic embedding $\iota_C:C\hookrightarrow M$ in a symplectic manifold $(M,\omega)$, and denote by $\underline{C}=C/C^\omega$ the coisotropic reduction (i.e.\ the space of leaves of the characteristic foliation of the coisotropic submanifold $C$). It is a well known fact that $C^\infty(\underline{C}) \simeq N(I_C)/I_C$ where $I_C$ is the vanishing ideal of $C$, and $N(I_C)$ is its normaliser\footnote{Recall: $N(I_C) \doteq \{ f \in C^\infty(M) \ | \ \iota_C^*(X_g f) =0 \ \forall g\in I_C \}$.} in $C^\infty(M)$ (see e.g.\ \cite[Lemma 2.2]{Dippel}).

A resolution of the coisotropic reduction of the embedding $C\hookrightarrow (M,\omega)$ is provided by a graded $0$-symplectic manifold $(M_{\BFV},\omega_{\BFV})$ together with a cohomological vector field $Q_{\BFV}\in\fX(\cE_{\BFV})$ of degree $1$, such that the cohomology in degree $0$ of the complex
    \[
    \mathfrak{BFV}(C,M)\doteq \left(C^\infty(\cE_{\BFV}),Q_{\BFV}\right)
    \]
is the space of functions on the coisotropic reduction, i.e.
    \[
    H^0\left(\mathfrak{BFV}(C,M)\right) \simeq N(I_C)/I_C \simeq C^\infty(\underline{C}).
    \]

Building $(\cE_\BFV,\omega_\BFV, Q_\BFV)$ and the associated complex $\mathfrak{BFV}(C,M)$ is the goal of the BFV construction. In particular the complex $\mathfrak{BFV}(C,M)$ controls deformations of the coisotropic submanifold $C\hookrightarrow M$ \cite{StasheffConstraints88,OhPark,SchaetzBFV}. We are going to apply, and illustrate, this framework by resolving the coisotripic constraint surface in the geometric phase space of a gauge field theory. Note that the BFV data we use in field theory are once again a densitised version of the BFV resolution described here.

Consider a (gauge) Lagrangian field theory $(\cE,\bL)$ with a locally symplectic geometric phase space $(\cE_\Sigma,\omega_\Sigma)$. We are going to assume that the Lagrangian field theory is regular (Definition \ref{def:constraintset}), i.e.\ that the constraint set $\cC_\Sigma\doteq \pi_\Sigma(\cE_j) \hookrightarrow \cE_\Sigma$ (Definition \ref{def:constraintset}) is coisotropic, and denote by $\underline{\cC}_\Sigma$ its (possibly non-smooth) coisotropic reduction. The BFV resolution of $\cC_\Sigma \hookrightarrow \cE_\Sigma$ is defined as follows:

\begin{definition}[BFV data]
    A BFV resolution of the coisotropic constraint set $\cC_\Sigma\hookrightarrow \cE_\Sigma$ of a regular Lagrangian field theory $(\cE,\bL)$ is given by data $(\cE_{\BFV},\omega_{\BFV}, \bL_{\BFV},Q_{\BFV})$ where
    \begin{enumerate}
        \item $(\cE_{\BFV},\omega_{\BFV})$ is a local $0$-symplectic manifold, with $\mathsf{Body}(\cE_\BFV)\equiv (\cE_\BFV)^{(0)} = \cE_\Sigma$ and $(\bom_\BFV)^{(0)} = \bom_\Sigma$;
        \item  $Q_{\BFV}\in\fX(\cE_{\BFV})$ is a cohomological vector field of degree $1$;
        \item $\bL_{\BFV}\in\oloc^{0,\top}$ is a $(0,\top)$ form of degree $1$ such that
        \[
        \bbI\i_{Q_{\BFV}}\bom_{\BFV} = \bbI \bd \bL_{\BFV};
        \]
        \item The form $\bC_{\BFV}$ defined via\footnote{Note, $\d c$ is of zeroth polynomial order in $c$.} 
        \[
        \bbI\d \bL_{\BFV} \doteq \langle \bC_{\BFV}, \d c\rangle + \text{higher polynomial orders in $c$}
        \]
        is such that $\cC_\Sigma = \{ \bC_{\BFV} = 0\}$.
    \end{enumerate}
\end{definition}

\begin{remark}
    In a strongly Hamiltonian field theory, the degree-1 part of $\bL_{\BFV}$ is essentially given by the Noether current. More precisely, in view of Lemma \ref{lem:Jdescends} one has
    \[
    \bL_{\BFV} = \langle \bHo - j_\Sigma,c\rangle + \text{(higher polynomial orders in $c$)} + \text{($d$-exact terms)}
    \]
    where $\langle\bHo-j_\Sigma,c\rangle$ is linear in $c$. Note that, then, $\bC_{\BFV} = \bHo - j_\Sigma$ and indeed $\cC_\Sigma = \mathrm{Zero}(\bHo - j_\Sigma)$ coincides with the zero-level set of the momentum map for the action of the constraint Lie algebra $\fGo$ (see Theorem \ref{thm:RSGo} above, as well as  \cite[Proposition 5.16 and Corollary 5.19]{RSATMP}). 
    Coisotropicity of $\cC_{\Sigma}$ follows from this observation(cf.\ \cite[Proposition 2.5]{BSW}).

    More generally, although the assumption that $\cC_\Sigma$ be coisotropic (i.e.\ regularity of the field theory) is restrictive, it only discards somewhat pathological situations. In particular, it is satisfied even in many models of General Relativity which are not strongly Hamiltonian \cite{LeeWald,BSW}. The fact that one can hope to construct a BFV resolution even in less than optimal situations, i.e.\ when the symmetries of the theory either are not given by a Lie group action in the first place, or when the Lie group action on $\cE$ does not induce a Lie group action on $\cE_\Sigma$ (cf. the discussion around Remark \ref{rmk:diffeos-not-descend}), suggests that the resolution approach is genuinely more general and flexible. 
\end{remark}

\subsection{BV-BFV formalism}

We illustrated the BV and the BFV formalisms as resolutions of the symmetry-reduced covariant and geometric phase spaces, respectively. In many cases of interest the two can be related to each other by considering the $BV$ data on a manifold with boundary.

Thus, let $M$ be a manifold with boundary, $\pp M\neq \emptyset$, and consider on it a Lagrangian field theory $(\cE,\bL)$ with BV extension $(\cE_{\BV},\omega_{\BV},\bL_{\BV},Q_{\BV})$. Recall from Definition \ref{def:BVdata} that $Q_{\BV}$ is the Hamiltonian vector field of $\bL_{\BV}$ only up to boundary terms, which means that using appropriate homotopies (here omitted) we can define a $\btheta_{\BV}\in \oloc^{1,\top -1}$ such that
\[
\i_{Q_{\BV}}\bom_{\BV} = \bd \bL_{\BV} + d \btheta_{\BV}.
\]

Following the procedure outlined in Section \ref{sec:covPS}, we can now look for a symplectic space of fields configurations $(\cE_{\pp M},\bom_{\pp M})$ on $\pp M$ and a surjective submersion $\pi_{\BV}\colon \cE \to \cE_{\pp M}$ such that $\iota^*_{\pp M}\bd\btheta_{\BV} = \pi_{\BV}^* \bom_{\pp M}$. More generally:

\begin{definition}[BV-BFV extension]
    A BV-BFV extension of a regular Lagrangian field theory $(\cE,\bL)$ is the data 
    \[
    (\cE_{\BV},\omega_{\BV},\bL_{\BV}, Q_{\BV}, \cE_{\BFV},\omega_{\BFV}, \bL_{\BFV},Q_{\BFV}, \pi_{\BV})
    \]
    such that 
    \begin{enumerate}
        \item $\pi_{\BV}\colon \cE_{\BV}\to \cE_{\BFV}$ is a surjective submersion,
        \item the compatibility conditions
            \[
            \iota^*_{\pp M}\bd\btheta_{\BV} = \pi_{\BV}^* \bom_{\pp M}, \qquad \{\bL_{\BV},\bL_{\BV}\}_{\bom_{\BV}} = d\bL_{\BFV}, \qquad Q_{\BV}\pi_{\BV}^* = \pi_{\BV}^*Q_{\BFV}
            \]
        are satisfied, 
        \item $(\cE_{\BFV},\omega_{\BFV}, \bL_{\BFV},Q_{\BFV})$ is a BFV resolution of $\cC_{\Sigma}=\pi_{\Sigma}(\cE_j)$, the constraint set of the Lagrangian field theory $(\cE,\bL)$.
    \end{enumerate} 
\end{definition}

\begin{remark}
    Note that, in practical situations, $\pi_{\BV}$ can be constructed (essentially) as restriction of fields to a connected component of the boundary of $M$, via the Kijowsky--Tulczyjew (KT) construction. It is a result of \cite{CMR1} that it is sufficient to build a surjective submersion $\pi_{\BV}$, in order to obtain the subsequent compatibility conditions as a byproduct.
    
    A BFV extension requires a compatibility between the BV extension of the theory and the BFV resolution of its constraint set. Examples are known of theories where such compatibility fails in otherwise equivalent models \cite{CStime,CS_PCH,CCS_AKSZ}, and the failure emerges during the presymplectic reduction necessary for the definition of $\cE_{\BFV}$ via the KT construction. To our knowledge this is more likely to happen in theories that enjoy diffeomorphism symmetries. Even so, some theories of gravity allow for a BV-BFV extension \cite{CS_EH,CCS_AKSZ,CaSc_3d1,CaSc_3d2}.
\end{remark}

\subsection{Higher BFV extensions to manifolds with corners}
The BV-BFV construction can be further extended to manifold with corners, thus transposing the relationship between the BV and BFV constructions to higher codimensions.

There is no universal consensus on the definition of a manifold with corners. We are going to work with Joyce's definition from \cite{Joyce_corners} (which we refer to also for a clear disambiguation on the notion). We summarise his construction and some results in the following definition:\footnote{The terminology ``skeleton" used here is not employed in \cite{Joyce_corners}.}
\begin{propdef}\label{def:corners}
    An $n$ manifold with corners is a topological space $M$ together with a partition $M=\bigsqcup_{k}M^{(k)}$ such that $M^{(k)}$, called the depth $k$ stratum, is an $(n-k)$-dimensional smooth manifold (without boundary) and $M^{(0)}=\mathring{M}$. The closure of the depth $k$ stratum is the union $\overline{M}^{(k)}=\bigcup_{i= k}^n M^{(k)}$, also called the depth $k$ skeleton. A local boundary component $\sigma$ of $M$ at $x$ is a choice of connected component of $M^{(1)}$ near\footnote{That is, one that intersects any open neighborhood centered at $x$. If $x$ is an element of $M^{(k)}$ and $k=0$ then the local boundary component of $M$ at $x$ is empty.} $x$,
    \[
    \pp M = \{(x,\sigma),\ x\in M\ |\ \sigma \ \text{local boundary component of $M$ at $x$} \}.
    \]
    The $k$-th boundary $\pp^k M$ of $M$ is the boundary of the $(k-1)$-th boundary of $M$ and
    \[
    \pp^k M \simeq \{(x,\sigma_1,\dots ,\sigma_k),\ x\in M\ |\ \sigma_i \ \text{distinct local boundary components for $M$ at $x$}\}.
    \]
    The $k$-th corner of $M$ is a manifold with corners obtained as the quotient of the $k$-th boundary by the action of the $k$-th permutation group $\mathcal{S}_k$ acting on the labels $\sigma_i$ of the local boundary components:
    \[
    \angle^k M \doteq \pp^k M / \mathcal{S}_k.
    \]
\end{propdef}

\begin{remark}
    To illustrate, according to this definition, the boundary of a cube seen as a manifold with corners is given by the union of 6 squares. This differs from a 2-sphere in that we remember the depth of points on the edges and vertices of the cube. A cube's depth 2 skeleton is the connected graph composed by 12 links and 8 vertices. 
\end{remark}

The name \emph{depth} $k$ strata $M^{(k)}$ comes from the notion of depth, which counts the number of ``boundary faces'' that contain any point $x\in M^{(k)}$ in any local chart, modeled on $\mathbb{R}^{n-k}\times [0, \infty )^k$, i.e.\ the number of local coordinates that are zero when parametrising a point in the local chart.

Note that the \emph{boundary} of a manifold with corners includes in its definition the choice of a connected component of the $1$-stratum. However, the immersion $\iota\colon \pp M\to M$ need not be injective, because a point $x$ may belong to multiple connected components of the $1$-stratum.

We can now give a notion of $k$-fold extension of Lagrangian field theory in the BV sense, as follows:

\begin{definition}[$n$-BV-BFV extension]
    Let $M$ be a manifold with corners and $n$ an integer in $[0,\dim(M)]$. An $n$-BV-BFV extension (or a BF${^n}$V) extension of a Lagrangian field theory  $(\cE,\bL)$ on $M$ is the data $(\cE^{(k)},\bom^{(k)},\bL^{(k)},Q^{(k)})$ for all $0\leq k\leq n$ given by
    \begin{enumerate}
        \item A $(k-1)$-symplectic local form $\bom^{(k)}\in\oloc^{1,\top}(\cE^{(k)}\times \pp^k M)$, with $\cE^{(k)}$ the space of sections of a graded bundle over $\pp^k M$;
        \item A degree $1$ cohomological vector field $Q^{(k)}\in\fX(\cE^{(k)})$;
        \item A degree $k$ Hamiltonian local form $\bL^{(k)}\in \oloc^{0,\top}(\cE^{(k)}\times \pp^k M)$:
        \[
        \i_{Q^{(k)}}\bom^{(k)} = \d\bL^{(k)} + d\btheta^{(k+1)},
        \]
    \end{enumerate}
    together with maps $\pi^{(k)}\colon \cE^{(k)}\to \cE^{(k+1)}$ such that $\iota_{\pp^{(k+1)}M}^*\bd \btheta^{(k+1)} = (\pi^{(k)})^*\bom^{(k+1)}$ and 
    \[
    \qquad (\pi^{(k)})^*Q^{(k+1)} = Q^{(k)}(\pi^{(k)})^*, \qquad \{\bL^{(k)},\bL^{(k)}\}_{\bom^{(k)}} = d \bL^{(k+1)}.
    \]
    
\end{definition}

We have seen that a $0$-BV-BFV (or simply BV) extension of a Lagrangian field theory provides a resolution of the symmetry-reduced coavariant phase space (a.k.a. moduli space) of the theory, and a $1$-BV-BFV (or simply BV-BFV) extension furthermore provides a resolution of the coisotropic reduction of its constraint set together with a compatibility relation in the form of a chain map between the respective complexes.
It is then natural to ask what of the original Lagrangian field theory is encoded by a $2$-BV-BFV extension. 

We must observe immediately that a $2$-BV-BFV extension only exists when the 2 boundary $\pp^2 M$ of $M$ is nontrivial, and thus when $\pp^1 M$ is a manifold with (at least) a boundary. On the other hand, in Section \ref{sec:hamwithcorners} we have seen that the presence of boundaries in the codimension $1$ hypersurface $\Sigma\hookrightarrow M$---which defines the constraint set $\cC_\Sigma$---determines a modification of the reduction procedure, as one ends up with the fully reduced phase space $\underline{\underline{\cC}}_\Sigma$ being a Poisson manifold whose symplectic leaves (the ``superselection sectors'') are labeled by codimension $2$ field configurations (the ``fluxes'').

Moreover, in section \ref{sec:corneralgebra} we have argued that, even before restricting to $\cC_\Sigma$, one can build a space of \emph{off-shell} corner configurations\footnote{There denoted $\cE_{\pp\Sigma} \leadsto \cE_\pp$. The same notational shift applies to other related quantities.} $\cE{_{\pp\Sigma}}$, which, in good cases, is itself a Poisson manifold. The Poisson bivector $\Pi_{\pp\Sigma}$ on $\cE_{\pp\Sigma}$ is given in Equation \eqref{e:OScornerPoisson}, where it was also observed that $\Pi_{\pp\Sigma}$ is associated to the existence of a degree $2$ function $S_{\pp\Sigma}$ on the shifted algebroid $\mathsf{A}_{\pp\Sigma}[1]\to \mathcal{E}_{\pp\Sigma}$, given in Equation \eqref{e:earlyBFFV}. As it turns out, if we interpret $\Sigma$ as one connected component of the boundary $\pp^1 M$ of the manifold with corners $M$, the data $(\mathsf{A}_{\pp\Sigma}[1],\omega_{\pp\Sigma}=\langle \d h_{\pp\Sigma}, \d c\rangle, S_{\pp\Sigma})$ is precisely the data of a 2-BV-BFV extension.

To see this, note that the Lie algebra of functions over the $(+1)$ symplectic graded manifold $(\mathsf{A}_\pp[1],\omega_\pp=\langle \d h_\pp, \d c\rangle)$ is isomorphic to the Lie algebra of multivector fields on $\mathcal{E}_\pp$ with Schouten bracket, $(\mathfrak{X}^\bullet(\cE_{\pp\Sigma}), [\cdot,\cdot]_\mathrm{SN})$. Then, under this identification, a function that is homogeneous of degree $2$ in fibre variables, such as $S_\pp$, satisfies the classical master equation if and only if the associated bivector field is Poisson:
\begin{gather}
\varphi\colon \left(C^\infty(\mathsf{A}_{\pp\Sigma}[1]),\{\cdot,\cdot\}_{\pp\Sigma}\right) \to (\mathfrak{X}^\bullet(\cE_{\pp\Sigma}), [\cdot,\cdot]_\mathrm{SN}), \\ 
\{S_{\pp\Sigma},S_{\pp\Sigma}\}_{\pp\Sigma}=0 \iff [\varphi(S_{\pp\Sigma}),\varphi(S_{\pp\Sigma})]_\mathrm{SN}=0
\end{gather}

More generally a 2-BF-BFV extension provides a generalisation of this picture, i.e.\ a degree $1$ symplectic manifold with a function of degree $2$ that solves the classical master equation given by $S^{(2)} = \int_{\angle^2 M}\bL^{(2)}$. Upon identifying 
\[
\mathcal{E}^{(2)}\simeq T^*[1]\mathcal{E}_{\angle^2M}
\]
we conclude that $\cE_{\angle^2M}$ comes equipped with a $P_\infty$ (``Poisson-infinity") structure \cite{CanepaCattaneo_corner}. 

This data codifies the existence of a Poisson manifold on some residual space of \emph{on-shell} corner configurations. This can be seen quite clearly in the example of $4d$ BF theory, where vanishing curvature conditions survive all the way down to codimension 2, and thus define a nontrivial submanifold of all possible corner configurations (see \cite{RSATMP,CanepaCattaneo_corner}.)

\section{Concluding remarks on the notion of symmetry}\label{sec:remarks}
    The discussion surrounding ``symmetries'' in field theory is often confusing, owing to the term itself being too all-encompassing. To disambiguate, let us analyse various scenarios.
    
    \begin{enumerate}
        \item Given an ordinary mechanical system, as discussed in Section \ref{sec:HamiltonianMechRed}, a symmetry is encoded by a $G$-dynamical system $(M,\omega,H,\Phi)$. It is important to note that, generally speaking, in thinking of mechanics as a 1d field theory for $C^\infty(I,M)$, symmetries arise as a group action on the \emph{target} manifold $(M,\omega)$. Note that one can think of these as constant maps in $C^\infty(I,G)$.
        \item  In field theory we can act via a local Lie algebra $\fG$, giving rise to what is usually called \emph{gauge symmetries} when the local action preserves the Lagrangian of the system (up to boundary terms).
        If the algebra is Abelian, a notion of constant gauge transformations exists\footnote{For Abelian gauge theories of connections, constant gauge transformations are well defined, but in the non Abelian case what is constant depends on a choice of trivialization. This is because, for a $P$ a principal fibre bundle, $\fG = \mathrm{ad}(P)$ and the adjoint action is trivial in Abelian case, but not in non-Abelian one. Generalizations of global symmetries in the non-Abelian case rely on $\fG$ having a nontrivial Abelianization, cf. Theorem \ref{thm:jext=0}.}. 
        These transformations act trivially on the connections $A_\mu$, but not on the (charged) matter fields, where they generate the corresponding finite dimensional charge group. These constant Abelian symmetries can be unambiguously thought of as \emph{global} symmetries. This is how the global $U(1)$ and the related nontrivial charge conservation equation is singled out in electrodynamics \cite{AbbottDeser,GomesRielloSciPost}.
        \item The field theory may depend on geometric data such as a \emph{fixed} metric on the source (or target) manifold. Automorphisms of that structure, such as isometries, may then act on the field theory (via diffeomorphisms). We think that the terminology \emph{global symmetry} is more appropriate to describe this scenario.
    \end{enumerate}
    
    Note that there is potential overlap between the three scenarios. For example, between (1) and (3) when the target group action is also an isometry; or between (2) and (3) when the global ${\rm U}(1)$ is a common automorphism for all $A_\mu$. Moreover, one can render the symmetries of the dynamical system with target $(M,\omega)$ \emph{local} by adding connections on a trivial $G$ bundle on $I$, in a procedure called \emph{gauging}, linking (1) and (2). This is why there is often a colloquial conflation of the various notions, and significant confusion.
    
    It is standard lore to argue that the ontological difference between these three notions of symmetry is that gauge transformations are un-physical or, as it is sometimes put, mere ``descriptive fluff". 

    In a spatially closed universe with no external sources, i.e.\ $M=\Sigma \times \mathbb{R}$ with $\pp\Sigma=\emptyset$ and $S_\bullet=0$, a more precise statement that can be used as a discriminating feature of gauge transformations is that, following Noether's prescription, the ``charge'' associated to gauge transformations vanishes on shell\footnote{I.e.\ on configurations that satisfy the Euler--Lagrange equations.} by Noether's second theorem \ref{thm:N2} (recall $S_\bullet = d j_\bullet^h$):
    \[
    \text{on $\cEL$, } \
    Q_\xi  = \int_\Sigma \bJ^h_\xi= \int_\Sigma \bC^h_\xi = \int_\Sigma j^h_\xi = 0 .
    \]
    Thus, the Noether charge $Q_\xi$, which in symmetries of type (1) would be conserved and an observable marker of the symmetry, is trivial for gauge symmetries when $\pp\Sigma = \emptyset$ and $S_\bullet=0$.

    However, in the presence of codimension $2$ corners $\pp\Sigma\neq\emptyset$ (see Definition \ref{def:corners}), Noether's second theorem gives
    \[
    \text{on $\cEL$, } \
    Q_\xi = \int_\Sigma \bJ^h_\xi= \int_{\pp\Sigma} \bK^h_{\xi\vert_{\pp\Sigma}}.
    \]
    This expression fails to automatically vanish whenever $\xi|_{\pp\Sigma}\neq 0$.
    It has therefore been argued in the literature that gauge transformations that do not vanish at the corner ``turn physical". The typical meaning of this statement that is found in the literature purports that one \emph{can} distinguish on-shell configurations that differ by a corner gauge transformations. We do not agree with this interpretation, as we will argue in what follows.
    
    What can \emph{objectively} be said about this scenario is that the charge prescribed by the Noether procedure no longer vanishes. From this point of view, the scenario is completely analogous to that of dynamical systems of type (1) described in Section \ref{sec:HamiltonianMechRed}, where one can choose to reduce at nonzero values of the momentum map (the Noether charge).
    
    We then argue that these \emph{corner symmetries} should be thought of as generalising the notion of symmetry of ordinary (finite dimensional) $G$-dynamical systems in an ``otherwise empty universe''. As we will now argue, both can be treated as ``descriptive, yet consequential, fluff".

    For gauge symmetries, an argument in favor of treating corner gauge symmetries as descriptive fluff is simply that if no measurement apparatus can distinguish between two gauge-related field configurations in the bulk, the same must be true when the apparatus is placed at the corner. Any claim of the opposite, we think, ought to be supported by a strong and explicitly thought-out counterexample (cf. \cite[Section 7]{RSAHP,GomesRielloSciPost}).

    Conversely, our argument in favor of treating symmetries for \emph{ordinary} $G$-dynamical systems as fluff, too, is that it seem implausible that the absolute position and velocity of the center of mass, or the overall orientation, of a mechanical system \emph{in an otherwise empty universe} have any physical significance (cf. \cite{Newton,earman1992world,GomesGribb}). For these ``absolute'' purposes, the fully reduced phase space with respect to the Galilean, or Poincaré, group offers a complete description of the dynamics of the universe in a reference-frame invariant fashion. Its superselection sectors account for different mass-energies and spins of the system.\footnote{In the relativistic case, mass-energy corresponds to the square of the total 4-momentum, and spin to the square of the Pauli-Lubanski 4-vector. In the Galilean case, mass and internal energy are to be considered separately; together with spin, they yield the three Casimirs of the Bargmann central extension of the Galilean group. Internal energy and spin are then defined as the energy and angular momentum in the centre-of-mass reference frame.}

    In this sense, the group of symmetries of a generalised dynamical system, including corner symmetries, can be considered as ``descriptive, yet consequential, fluff". By this we mean that while all the physical information is encoded in the \emph{fully-reduced} phase space---and hence symmetry-related configurations are to be identified---this reduction procedure has the nontrivial consequence of generating a merely Poisson, rather than symplectic, phase space, where superselection sectors emerge as its symplectic leaves. These leaves, on which the classical dynamics is restricted, constitute a classical analogue of quantum superselection sectors \cite{wightman1995}.
    (In the case of field theory, by ``dynamics" we mean the dynamics within a causal domain such as those of Fig. \ref{fig:causalregions} so that the superselection data is assigned to the corner that is there denoted $\pp\Sigma$; for the mechanical case, see Section \ref{sec:HamiltonianMechRed}, especially Remark \ref{rmk:reduced-all-relevant-info}, where for the sake of simplicity we focus on rotational symmetry only.)

    An important objection, however, is that one might want to model the physics of \emph{sub-systems} that do not exhaust the content of the universe \cite{aharonov2008book,BartlettEtAl2007}. In that specific scenario, keeping track of the ``absolute" position/velocity/orientation of a system allows one to compare subsystems and to compose them in a way that allows and accounts for a plethora of different physical situations.

    For example, the \emph{relative} orientation of two subsystems is indeed physical, in the sense that it is \emph{relatively observable}. This is based, e.g., on the experimental fact that the vector value of the angular momentum of a subsystem is indeed measurable \emph{relative to} the lab's reference frame. Of course, to make this argument precise, one would need to elaborate a measurement protocol for classical systems enriched by symmetry considerations.
    
    This raises the natural question: can a similar \emph{relational} phenomenon be at play when composing adjacent subsystems in the presence of corner gauge symmetries? 
    
    The gluing theorem of \cite[Sect. 6]{GomesRielloSciPost} suggests that care is needed when answering this question, since it finds that in electrodynamics\footnote{The same result applies to the gluing of linear perturbations in the non-Abelian theory.} no freedom exists when (smoothly) gluing two adjacent gauge subsystems across a common corner surface: the possible \emph{relative} (corner) gauge transformation of the electromagnetic field $A_\mu$ in two adjacent subsystems is fully fixed when asking that the glued, global, configuration is smooth across the interface. (A caveat to this result arises in the presence of charged matter fields sensitive to the \emph{global} $\rm U(1)$ symmetry mentioned at the point (2) above \cite[Sect. 6.4]{GomesRielloSciPost}.)

    This observation hints at a potential divergence in behaviour between composition of systems in these two scenarios (classical mechanics and gauge field theory). A systematic comparison requires facing the issue that composition in the 1d-field theory representing classical mechanics happens in the \emph{target} manifold, while the composition in field theory happens in the \emph{source}. We are then in need of a compelling protocol that is able to handle open/closed systems in both classical mechanics and field theory and a compatible notion of relative observability, and thus of classical measurement.

    In summary, we argued that it is consistent to treat \emph{corner symmetries} in field theory on equal footing with symmetries for $G$-dynamical systems, in an otherwise empty universe. In the field theoretical case, the role of ``the otherwise empty universe" is played by one of the causal regions of Fig. \ref{fig:causalregions}. In both cases the fully reduced phase space contains all the physically relevant information, and thus both types of symmetries can be safely reduced without loss of information, as some of this information is recovered through the superselection data. The difference between these two scenarios emerges when dealing with the composition and relative measurement of subsystem. In order to make any further statement, however, more effort has to be put to model either operation.

    \newpage

     \begingroup
     \sloppy
     \printbibliography
     \endgroup
\end{document}